\documentclass[jair,twoside,11pt,theapa]{article}
\usepackage[hidelinks,breaklinks]{hyperref}
\usepackage{breakurl}
\hypersetup{%
    bookmarks=false,    
    pdftitle={Query Answering with Transitive and Linear-Ordered Data},    
    pdfauthor={Amarilli, Benedikt, Bourhis, \& Vanden Boom},                     
    linkcolor=black,       
    citecolor=black,       
    filecolor=black,        
    urlcolor=purple,        
    linktoc=page            
}
\usepackage{jair, theapa, rawfonts}

\usepackage{times}
\usepackage{latexsym}
\usepackage{tabularx,booktabs,multirow}
\usepackage{amsmath}
\jairheading{63}{2018}{X-Y}{8/17}{10/18}
\ShortHeadings{Query Answering with Transitive and Linear-Ordered Data}
{Amarilli, Benedikt, Bourhis, \& Vanden Boom}

\usepackage{tikz}
\usetikzlibrary{automata,positioning,calc}
\usetikzlibrary{decorations.pathreplacing}
\usetikzlibrary{decorations.markings}
\usetikzlibrary{shapes.geometric}
\usetikzlibrary{shapes.misc}
\usetikzlibrary{shapes.arrows} 
 
\usepackage{xspace}
\usepackage{amsthm}
\usepackage{amssymb}
\usepackage{amsfonts}
\usepackage{paralist}
\usepackage{url}
\usepackage{subcaption}
\usepackage{colonequals}

\usepackage{bookmark}
\bookmarksetup{startatroot}

\newcommand{\fo}{\kw{FO}}

\newcommand{\types}{\kw{Types}}
\newcommand{\instance}{{\cal F}}
\newcommand{\gnfpup}{\kw{GNFP^{UP}}}
\newcommand{\fgtgd}{\kw{FGTGD}}
\newcommand{\afgtgd}{\kw{BaseFGTGD}}
\newcommand{\acfgtgd}{\kw{BaseCovFGTGD}}

\newcommand{\agnf}{\kw{BaseGNF}}
\newcommand{\acgnf}{\kw{BaseCovGNF}}
\newcommand{\mapsfrom}{:=}
\newcommand{\qr}{\kw{QR}}

\newcommand{\gnf}{\kw{GNF}}

\newcommand{\gf}{\kw{GF}}

\newcommand{\decode}{\kw{decode}}

\newcommand{\dom}{\elems}
\newcommand{\child}{\kw{Child}}
\newcommand{\incd}{\kw{ID}}
\newcommand{\aincd}{\kw{BaseID}}
\newcommand{\did}{\kw{DID}}
\newcommand{\owqa}{\kw{QA}}

\newcommand{\owqalin}{\owqa{\kw{lin}}}
\newcommand{\owqatr}{\owqa{\kw{tr}}}
\newcommand{\owqatrans}{\owqatr}

\newcommand{\owqatc}{\owqa{\kw{tc}}}

\newcommand{\ie}{i.e.,~}

\newcommand{\kw}[1]{{\mathsf{#1}}\xspace}
\newcommand{\kwr}[1]{{\mathrm{#1}}\xspace}

\newtheorem{theorem}{Theorem}[section]

\newtheorem{proposition}[theorem]{Proposition}
\newtheorem{corollary}[theorem]{Corollary}

\newtheorem{definition}[theorem]{Definition}

\theoremstyle{definition}
\newtheorem{example}[theorem]{Example} 

\newtheorem{lemma}[theorem]{Lemma}

\renewcommand{\phi}{\varphi}

\newcommand{\myparagraph}[1]{\paragraph{#1.}}
\newcommand{\mysubparagraph}[1]{\subparagraph{\normalfont \textit{#1.}}}
\newcommand{\myeat}[1]{}

\makeatletter
\def\@Opargbegintheorem#1#2#3#4{#4\trivlist
      \item[\hskip\labelsep{#3#1}]{#3#2\@thmcounterend\ }}

\newcommand{\definerep}[2]{%
\spnewtheorem*{#1rp}{#2}{\bf}{\itshape}
\newenvironment{#1rep}[2]{%
  \ifthenelse{\equal{##1}{*}}
  {\begin{#1rp}[\ref{##2}]}
  {\begin{#1}\label{##2}}}
{\ifthenelse{\equal{\@currenvir}{#1}}{\end{#1}}{\end{#1rp}}}
}

\newcommand{\calA}{\mathcal{A}}
\newcommand{\calB}{\mathcal{B}}

\newcommand{\calD}{\mathcal{D}}

\newcommand{\calF}{\mathcal{F}}
\newcommand{\calG}{\mathcal{G}}

\newcommand{\calI}{\mathcal{I}}

\newcommand{\calN}{\mathcal{N}}

\newcommand{\calS}{\mathcal{S}}

\newcommand{\NN}{\mathbb{N}}

\newcommand{\sigmab}{\sigma_{\calB}}

\newcommand{\sigmad}{\sigma_{\calD}}

\newcommand{\sigmaf}{\sigma_{\calF}}
\newcommand{\sigmas}{\sigma_{\calS}}
\newcommand{\sigmaor}{\sigma_{\Or}}
\newcommand{\sigmapb}{\sigma'_{\calB}}
\newcommand{\sigmapd}{\sigma'_{\calD}}

\newcommand{\arity}[1]{\kwr{arity}(#1)}

\makeatletter
\newcommand*{\defeq}{\mathrel{\rlap{%
  \raisebox{0.3ex}{$\m@th\cdot$}}%
  \raisebox{-0.3ex}{$\m@th\cdot$}}%
  =}
\makeatother

\newcommand{\guardedb}{\kw{guarded}_{\sigmab}}

\newcommand{\guardedbg}{\kw{guarded}_{\sigmab \cup \set{G}}}

\newcommand{\drel}{S}

\newcommand{\acgnfplus}{\acgnf^{+}}
\newcommand{\acgnfminus}{\acgnf^{-}}
\newcommand{\tree}{T}
\newcommand{\mysize}[1]{|#1|}

\newcommand{\witness}{\kw{Witness}}
\newcommand{\twoexp}{\kw{2EXPTIME}}
\newcommand{\twoexptime}{\twoexp}
\newcommand{\exptime}{\kw{EXPTIME}}
\newcommand{\conp}{\kw{CoNP}}
\newcommand{\np}{\kw{NP}}
\newcommand{\ptime}{\kw{PTIME}}

\newcommand{\trans}{+}

\newcommand{\true}{\mathfrak{t}}
\newcommand{\false}{\mathfrak{f}}

\newcommand{\gn}{\text{GN}}
\newcommand{\fA}{{\cal F}}
\newcommand{\fB}{{\cal G}}

\newcommand{\set}[1]{\{ #1 \}}
\newcommand{\elems}[1]{\kwr{elems}(#1)}
\newcommand{\structureunravelk}[1]{#1^k}
\newcommand{\structureunravelkint}[1]{#1^k_{\cal B}}

\newcommand{\Or}{\mathrm{Or}}

\newcommand{\card}[1]{\left|#1\right|}

\newcommand{\perm}[1]{\langle #1 \rangle}
\newcommand{\cA}{\calA}
\newcommand{\cB}{\calB}
\newcommand{\powerset}[1]{\mathcal{P}(#1)}
\newcommand{\Dir}{\text{Dir}}
\newcommand{\dleft}{\mathsf{left}}
\newcommand{\dright}{\mathsf{right}}
\newcommand{\dup}{\mathsf{up}}
\newcommand{\dstay}{\mathsf{stay}}
\newcommand{\N}{\mathbb{N}}
\newcommand{\sset}{\set}
\newcommand{\paramsk}{U}

\newcommand{\sigcode}[2]{\Sigma^{\text{code}}_{#1,#2}}

\newcommand{\bagnames}[1]{\text{names}(#1)}
\newcommand{\mydecode}[1]{\decode(#1)}
\newcommand{\autsig}{\Gamma}

\newcommand{\fakethm}[3]{
\medskip
\noindent {\bf #1~\ref{#2}.} \emph{#3}
}

\newcommand{\True}{\mathrm{True}}

\newcommand{\ptrue}{{+\true}}

\newcommand{\ptwowayaltinf}{\kw{2APT}}
\newcommand{\ponewayndinf}{\kw{1NPT}}

\renewcommand{\sqcup}{\cup}

\newcommand{\rwto}{\leadsto}

\newcommand{\quoteresult}[2]{
\begin{quote}
{\bf #1.}
#2
\end{quote}
}

\newcommand{\findstart}{\textit{start}}
\newcommand{\findend}{\textit{end}}
\newcommand{\choosenext}{\textit{next}}

\newcommand{\accsink}{\textit{win}}
\newcommand{\rejsink}{\textit{lose}}

\newcommand{\treelab}{\beta}

\usepackage{jair, theapa, rawfonts}
\input{fixjair}

\title{Query Answering with Transitive and Linear-Ordered Data}
\author{\name Antoine Amarilli
\email antoine.amarilli@telecom-paristech.fr \\
\addr LTCI, T{\'e}l{\'e}com ParisTech, Universit{\'e} Paris-Saclay
\AND
\name Michael Benedikt
\email michael.benedikt@cs.ox.ac.uk \\
\addr University of Oxford
\AND
\name Pierre Bourhis
\email pierre.bourhis@univ-lille1.fr \\
\addr CNRS CRIStAL, Universit\'e Lille 1, INRIA Lille
\AND
\name Michael Vanden Boom
\email michael.vandenboom@cs.ox.ac.uk \\[-.18em]
\addr University of Oxford
}

\begin{document}

\maketitle

\begin{abstract}
We consider entailment problems involving powerful constraint
languages such as \emph{frontier-guarded existential rules}
in which we impose additional semantic restrictions
on a set of distinguished relations. We consider 
restricting a relation to be transitive,  restricting a relation  to be the transitive closure of another relation, 
and restricting a relation to be a linear order. We give some natural variants of guardedness
that allow inference to be decidable in each case, and isolate the complexity of the corresponding decision problems. Finally
we show that slight changes in these conditions lead to undecidability.

This version of the paper includes one modification from the version originally
published at JAIR in 2018: we fix an incorrect proof for one of our undecidability results.

\end{abstract}

\section{Introduction} \label{sec:intro}

The \emph{query answering problem} (or certain
answer problem), abbreviated here as $\owqa$, 
is a fundamental reasoning problem
in both knowledge representation and databases.
It asks
whether a query (e.g.,~given by an existentially-quantified conjunction of atoms) is entailed
by a set of constraints and a set of facts. That is, we generalize the standard
querying problem in databases to take into account not only the explicit
information (the facts) but additional ``implicit information''  given by
constraints.
A common class of constraints used
for $\owqa$ are \emph{existential
rules}, also known as \emph{tuple generating dependencies}
(TGDs). 

Although query answering is known to be undecidable  for general TGDs,
there are a number of subclasses 
that admit decidable $\owqa$, such as those
based on \emph{guardedness}.
For instance, \emph{guarded} TGDs
 require all variables in the body
 of the dependency to appear in a single body atom (the \emph{guard}).
 \emph{Frontier-guarded} TGDs ($\fgtgd$s)
relax this condition
and require only that some
guard atom 
contains the variables that occur in both head and body \cite{bagetcomplexityfg}.
These classes include standard
SQL referential constraints as well as important constraint
classes (e.g., role inclusions) arising in knowledge representation.
Guarded existential rules
can be generalized to
\emph{guarded logics} that allow disjunction and negation and still
enjoy decidable $\owqa$, e.g., the
guarded fragment of first-order logic ($\gf$)
of \citeA{andreka1998modal} which captures guarded TGDs,
and the guarded negation fragment
($\gnf$) of \citeA{gnficalp} which captures $\fgtgd$s.

A key challenge is to extend these decidability results to capture
additional semantics of the relations that are important in practice but cannot
be expressed in these classes. For example, the property
that a binary relation is \emph{transitive} or is the \emph{transitive closure} of another relation
cannot be expressed
directly in guarded logics.  Yet, transitive
relations, such as the ``part-of'' relationship among components,
are common in data modelling and have received significant attention (see, e.g., \citeR{oldhorrocks,mugnier15}).
For example,
using standard reasoning
$Q = \exists xy ( \kw{broken}(x) \wedge \kw{part\text{-}of}(x,y) \wedge
\kw{car}(y) )$,
is not entailed by the set of facts $\set{\kw{car}(c), \kw{motor}(m), \kw{sparkplug}(p), \kw{part\text{-}of}(p,m), \kw{part\text{-}of}(m,c), \kw{broken}(p) }$.
But enforcing that $\kw{part\text{-}of}$ is transitive---a very natural condition
when modelling this situation---makes a difference, and means that $Q$ is entailed.
Hence, we would like to be able to
capture these special semantics when reasoning.

A semantic restriction related to transitivity is the fact
that a binary
relation is a strict \emph{linear order}:
a transitive relation
which is also irreflexive and total.
For example, when the data stored in a relation is numerical, it may happen
that the data satisfies integrity constraints involving the standard linear order $<$ on integers;
e.g., for every
tuple in  a ternary relation the  value in the first
position of a binary relation  is  less than the  value in the third.
Again, the ability to reason about the additional semantics of $<$
may be crucial in inference, but it is not possible to express
this in guarded logics.
Query answering with additional semantic relations, e.g., linear orders or arithmetic, 
has been studied in the description logic
and semantic web communities \cite{datatypes1,datatypes2,datatypes3}. However 
query answering with linear orders 
has  not received much attention for arbitrary arity relations.

In this work we look at conditions that make query
answering decidable.
We study the three semantic restrictions above: transitivity, transitive
closure, and linear ordering. We will show that there are common techniques
that can be used to analyze all three cases, but also significant differences.

\subsection{State of the Art}
There has been extensive work on decidability results for guarded logics
extended with such semantic restrictions.

We first review known results for 
the
\emph{satisfiability problem},
which asks whether some logical constraints are
satisfiable.
\citeA{undecidgf2} showed that satisfiability is not
decidable for $\gf$ when two relations
are restricted to be transitive, even on arity-two signatures
(i.e., with only unary and binary relations).
For linear orders, \citeA{kieronski2011decidability} showed that $\gf$ is
undecidable when three relations are restricted to be non-strict linear
orders, even with only two variables (so on arity-two signatures). 
\citeA{ottofo2order} showed that satisfiability is decidable for
two-variable first-order logic with
one relation restricted to be a linear order.
For transitive relations,
one way to regain decidability for $\gf$ satisfiability
is the \emph{transitive guards} condition introduced by \citeA{gftgdecid}:
allow transitive
relations \emph{only} in guards.

We now turn to the \emph{$\owqa$ problem}, where we also consider a query and an
initial set of facts.
\citeA{andreaslidia} showed that query answering for
$\gf$ with transitive relations only in guards
is undecidable,
even on arity-two signatures.
\shortciteA{mugnier15} studied
$\owqa$ 
with respect to a collection
of linear TGDs (those with
only a single atom in the body).
They showed that the query
answering problem is decidable with such TGDs and transitive relations, if the
signature is arity-two
or if other additional restrictions are obeyed.

The case of TGDs mentioning relations with a restricted interpretation 
has been studied in the database community mainly in the setting of acyclic schemas, such as those
that map source data to target data. Transitivity restrictions have not been studied,
but there has been work on inequalities \cite{abdus} and TGDs with arithmetic \cite{afratiarith}.
Due to the acyclicity assumptions, $\owqa$ is decidable in these settings, and has data complexity in $\conp$.
The fact that the data complexity can be $\conp$-hard
was shown by \citeA{abdus}, while polynomial cases
were isolated by \citeA{abdus} with inequalities,
and by \citeA{afratiarith}
with arithmetic.

Query answering 
that features transitivity restrictions
has also been studied for constraints expressed in description
logics, i.e., in an arity-two setting where
the signature contains unary relations (concepts) and
binary relations (roles). 
$\owqa$ is then decidable for many description logics
featuring 
transitivity, such as $\mathcal{SHIQ}$ \cite{glimmetal}, $\mathcal{ZIQ}$,
$\mathcal{ZOQ}$, $\mathcal{ZOI}$ \cite{calvanese2009regular}, 
Horn-$\mathcal{SROIQ}$ \cite{ortiz2011query}, OWL2 EL with the regularity restriction
\cite{stefanoni},
or
regular-$\mathcal{EL}^{++}$ \cite{kroetzsch2007conjunctive}.
All of these logics are incomparable to the ones we consider.
For example, the language considered by \cite{stefanoni} includes
powerful features beyond transitive closure operators, such as
role composition. On the other hand, it allows only arity $2$ relations, and
further restricts the use of inverse roles, which has a significant impact on complexity.
For even more expressive description  logics with transitivity,
such as 
$\mathcal{ALCOIF^*}$ \cite{ortiz2010query2}
and $\mathcal{ZOIQ}$ \cite{ortiz2010query}, $\owqa$ becomes undecidable. Decidability of $\owqa$ is open for
$\mathcal{SROIQ}$ and $\mathcal{SHOIQ}$
\cite{ortiz2012reasoning}.

\subsection{Contributions}
The main contribution
of this work is to introduce a broad class of constraints over arbitrary-arity
vocabularies where query answering is decidable even when we assert that some distinguished relations
follow one of three semantics:
being transitive,
being the transitive closure of another relation, or being a linear order.

\begin{itemize}
\item We provide new results on $\owqa$ with transitivity and transitive closure
  assertions.
We show that
  query answering is decidable
    with guarded and frontier-guarded constraints,
as long as these distinguished relations are \emph{not} used as guards.
We call this new kind of restriction \emph{base-guardedness}, and similarly
    extend frontier-guarded to \emph{base-frontier-guardedness}, and so forth.
    The base-guarded restriction is orthogonal to
the prior decidable cases such as transitive guards \cite{gftgdecid} for
    satisfiability, or linear rules
\cite{mugnier15}.  

On the one hand, we show that 
our restrictions make query answering decidable even with very expressive and flexible decidable logics, capable
of expressing  not only guarded existential rules, but also guarded rules with negation and disjunction in the head.
These logics can express integrity constraints, as well as conjunctive queries and their negations.
On the other hand, as a by-product of our results we obtain  new query answering schemes for some previously-studied
classes of guarded existential rules with extra semantic restrictions.
For example, our base-frontier-guarded constraints encompass
all \emph{frontier-one TGDs}
\cite{baget2009extending},
where at most one variable is shared between the body and head.
Hence, our results imply that $\owqa$ 
    with transitivity assertions  (or even transitive closure assertions)
    is decidable with
frontier-one TGDs, which answers a question of
\citeA{mugnier15}.

Our results are shown by arguing that it is enough to consider entailment over  ``tree-like'' sets of facts. By representing
the set of witness representations
as a  tree automaton,
we derive upper bounds for the combined complexity of the problem. 
The sufficiency of tree-like examples also enables a refined analysis of \emph{data complexity} (when the query and
constraints are fixed).
Further, we use a set of coding techniques 
to  show  matching lower bounds
within our fragment.
We also show that loosening our conditions leads to undecidability.

\item We provide both upper and lower bounds on
 the $\owqa$ problem
when the distinguished relations are \emph{linear orders}.

We show that it is undecidable 
    even assuming base-frontier-guardedness, so we introduce a
stronger condition called \emph{base-coveredness}:
not only are distinguished relations never used as guards, they are always \emph{covered}
by a non-distinguished atom.
Under these conditions, our decidable technique for $\owqa$ works by ``compiling away'' linear order restrictions, obtaining an entailment problem
without any special restrictions.
The correctness proof for our reduction to classical $\owqa$ again relies
on the ability to  restrict reasoning to  sets of facts with tree-like representations.
To our knowledge,
these are the first
decidability results for the $\owqa$ problem with linear orders, and
again we provide tight complexity bounds for the problem.
\end{itemize}
Both classes of results apply  to the motivating scenarios for distinguished
relations mentioned earlier.
Our results on transitivity
show that $\owqa$ with distinguished relations that are transitive or are the transitive closure of a base relation is decidable for $\agnf$,
the restriction of $\gnf$ that follows our base-guardedness requirement. 
In particular  this means that  in query answering with rules 
\begin{inparaenum} 
\item in the \emph{head} of a rule, or in the query, one can freely restrict some relations
to be transitive, or to be the transitive closure of some other relation;
\item in the \emph{body} of a rule, one can restrict some relations to be transitive
or to be the transitive closure, provided that there is a frontier-guard  available
that is not restricted.
\end{inparaenum}
In particular, we can use a transitive relation
such as ``part-of'' (or even its transitive closure) whenever
only one variable is to be exported to the head: that is in ``frontier-$1$'' rules. This latter condition
holds in the translations of many classical description logics.
For example,
$\forall x y ( \kw{part\text{-}of}(x,y) \wedge \kw{broken}(x) \rightarrow \kw{broken}(y) )$
and $\forall p (\kw{sparkplug}(p) \rightarrow \exists m (\kw{motor}(m) \wedge \kw{part\text{-}of}(p,m) ) )$
can be rewritten in $\agnf$.

Our results on  $\owqa$
with linear orders
show that the problem
is decidable for $\acgnf$,
the base-covered version of $\gnf$.
This allows constraints  that arise from data integration and
data exchange over attributes with linear orders---e.g., views
defined by selecting rows of a table where some order constraint involving
the attributes is satisfied.

\subsection{Organization}
In Section~\ref{sec:prelims}, we formally define
the query answering problems that we study,
and the constraint languages that we use.
We present our main decidability results on
query answering with transitive 
data in Section~\ref{sec:decid},
and with linear-ordered data in Section~\ref{sec:decidlin};
we analyze both the combined complexity and data complexity
of these decidable cases.
We prove lower bounds for these problems in Section~\ref{sec:hardness},
and show that slight changes to the conditions lead
to undecidability in Section~\ref{sec:undecid}. Section \ref{sec:undecid}
also compares the undecidability results with prior results in the literature.

Our main results are summarized in Figure~\ref{tab:complexity}, and
the languages that we study are illustrated in Figure~\ref{fig:tax}
(please see Section~\ref{sec:prelims} for the definitions).
Some technical material that is not essential for understanding
our main results can be found in the appendices.

\section{Preliminaries} \label{sec:prelims}

We work on a \emph{relational signature}  $\sigma$, where each relation $R \in\nolinebreak
\sigma$ has an associated \emph{arity}
written $\arity{R}$; we write $\arity{\sigma} \defeq \max_{R \in \sigma}
\arity{R}$.
A \emph{fact} $R(\vec{a})$, or \emph{$R$-fact}, consists of a relation $R\in \sigma$ and
elements $\vec{a}$, with $\card{\vec{a}} = \arity{R}$. $\vec{a}$ is the \emph{domain}
of the fact. Queries and constraints
will be evaluated over a (finite or infinite) set of facts over $\sigma$.
We will often use $\instance$ to denote a 
 set of facts.
We write $\elems{\instance}$ for the set of elements
that appear as arguments in the facts in $\instance$. We also refer to this
as the \emph{domain of $\instance$}.

We consider \emph{constraints} and \emph{queries} given in fragments of first-order
logic with equality ($\fo$) without constants.
Given
a set of facts $\instance$ 
and 
a sentence $\phi$ in $\fo$, we talk of $\instance$ \emph{satisfying} $\phi$ in the
usual way.
The \emph{size} of $\varphi$, written $\mysize{\varphi}$, is defined
to be the number of symbols in $\varphi$. 

A \emph{conjunctive query} (CQ) 
is a 
first-order formula of the form $\exists \vec{x} \, \varphi$ for $\varphi$ a conjunction of atomic formulas
using equality or a relation from $\sigma$.
Likewise, a \emph{union of
conjunctive queries} (UCQs) is a disjunction of CQs.
We will only use queries that are \emph{Boolean} CQs or UCQs (\ie CQs or UCQs with no free variables).

\begin{figure}
\centering
  \begin{tabularx}{\linewidth}{X@{\quad}l@{~~~~}l@{\quad\quad}l@{~~~~}l@{\quad\quad}l@{~~~~}l}
\toprule
\multirow{ 2}{*}{\begin{tabular}[t]{@{}l@{}}{\bf Fragment}\\\null\end{tabular}} & \multicolumn{2}{c}{\bf $\!\!\!\!\!\!\!\!\!\!$$\owqatr$} &
\multicolumn{2}{c}{\bf $\!\!\!\!\!\!\!\!\!\!$$\owqatc$} & \multicolumn{2}{c}{\bf
$\!\!\!\!\!\!\!\!\!\!$$\owqalin$}\\
& \bf data & \bf combined
& \bf data & \bf combined
& \bf data & \bf combined \\
\midrule
{$\agnf$}
& coNP-c & 2EXP-c
& coNP-c & 2EXP-c
& \multicolumn{2}{c}{{$\!\!\!\!\!\!\!$undecidable}} \\
{$\acgnf$}
& coNP-c & 2EXP-c
& coNP-c & 2EXP-c
& coNP-c & 2EXP-c \\
{$\afgtgd$}
& in coNP & 2EXP-c
& coNP-c & 2EXP-c
& \multicolumn{2}{c}{{$\!\!\!\!\!\!\!$undecidable}} \\
{$\acfgtgd$}
& P-c & 2EXP-c
& coNP-c & 2EXP-c
& coNP-c & 2EXP-c \\
\bottomrule
\end{tabularx}
\caption{Summary of $\owqa$ results. On the rows that concern
base-covered fragments, the queries are also assumed to be base-covered.
For complexity class $X$, we write ``$X$-c'' for ``$X$-complete''.
Please refer to Sections~\ref{sec:decid}~and~\ref{sec:decidlin} for upper bounds,
Section~\ref{sec:hardness} for lower bounds,
and Section~\ref{sec:undecid} for undecidability results.}
\label{tab:complexity}
\end{figure}

\subsection{Problems Considered}\label{subsec:problemsconsidered}
Given a 
\emph{finite} set of facts $\instance_0$, constraints $\Sigma$ and query $Q$ (given as $\fo$
sentences), we say that 
$\instance_0$ and $\Sigma$ \emph{entail} $Q$  if for  every 
$\instance \supseteq \instance_0$
satisfying $\Sigma$ (including infinite~$\instance$),
we have that $\instance$ satisfies~$Q$.
This amounts to asking whether $\calF_0 \wedge \Sigma \wedge \neg Q$
is unsatisfiable, over all finite and infinite sets of facts.
We write $\owqa(\instance_0, \Sigma, Q)$ for this decision problem, called
the \emph{query answering} problem.

In this paper, we study the $\owqa$ problem when imposing semantic constraints
on some \emph{distinguished} relations.
We thus work with signatures of the form $\sigma \defeq \sigmab \sqcup \sigmad$, where $\sigmab$ is the \emph{base signature}
(its relations are the \emph{base} relations), and
$\sigmad$ is the \emph{distinguished} signature (with \emph{distinguished}
relations),
and $\sigmab$ and $\sigmad$ are disjoint.
All distinguished relations are
required to be binary,
and they will be
assigned special semantics.

We study three kinds of special semantics:

\begin{itemize}
\item We say \emph{$\instance_0, \Sigma$ entails $Q$ over transitive relations}, and
write $\owqatrans(\instance_0, \Sigma, Q)$ for the corresponding problem,
if there is no set of facts $\instance$ where $\instance_0 \wedge \Sigma \wedge
    \neg Q$ holds and 
each relation $R_i^\trans \in \sigmad$ is
\emph{transitive}.\footnote{Note that we work with \emph{transitive}
relations, which may not be \emph{reflexive}, unlike, e.g., $R^*$  roles in
$\mathcal{ZOIQ}$ description logics \cite{calvanese2009regular}. This being
    said, all our
    results extend with the same complexity to relations that are both reflexive
    and transitive.
Please refer to Section~\ref{sec:reflexive-transitive} for more information.}

\item We say \emph{$\instance_0, \Sigma$ entails $Q$ over transitive closure}, and
write
$\owqatc(\instance_0, \Sigma, Q)$ for this problem,
if there is no set of facts $\instance$ where
    $\instance_0 \wedge \Sigma \wedge \neg Q$ holds and 
for each relation $R_i \in \sigmab$, the relation $R^\trans_i \in \sigmad$
is interpreted as the transitive closure of $R_i$.

\item We say \emph{$\instance_0, \Sigma$ entails $Q$ over linear orders},
and write $\owqalin(\instance_0, \Sigma, Q)$ for this problem,
if there is no set of facts $\instance$ where $\instance_0 \wedge \Sigma \wedge
    \neg Q$ holds and
each relation ${<_i} \in \sigmad$ is
a strict linear order on the elements
of~$\instance$.
\end{itemize}

\begin{example}
  Referring back to the example used in the introduction, we consider the query 
  $Q = \exists xy ( \kw{broken}(x) \wedge \kw{part\text{-}of}(x,y) \wedge
\kw{car}(y) )$,
and the set of facts $\instance_0 = \set{\kw{car}(c), \kw{motor}(m), \kw{sparkplug}(p), 
  \allowbreak \kw{part\text{-}of}(p,m), \kw{part\text{-}of}(m,c), \kw{broken}(p) }$.
  The query $Q$ is not entailed by~$\instance_0$ in general, but it is entailed
  in both $\owqatr$, $\owqatc$, and $\owqalin$ when the distinguished relation
  $\kw{part\text{-}of}$ is asserted to be transitive.

  The difference in semantics between $\owqatr$ and $\owqatc$ can be exemplified
  with the rule
  \[\kw{same\text{-}supplier}(x, y), \kw{part\text{-}of}^+(x,y),
  \neg \kw{part\text{-}of}(x,y) \rightarrow \kw{indirect\text{-}pair}(x, y),\]
  which applies to all pairs of objects $x, y$ produced by the same supplier
  such that $x$ is \emph{indirectly} a part of~$y$. This rule can be expressed in
  the $\owqatc$ setting, whereas in $\owqatr$ we cannot distinguish between 
  $\kw{part\text{-}of}$ and its transitive closure
  $\kw{part\text{-}of}^+$.

  As for the semantics of $\owqalin$, it differs from both $\owqatr$ and
  $\owqatc$. Consider the set of facts $\instance_0' = \set{
    \kw{motor}(m), \kw{sparkplug}(p),
    \kw{prod\text{-}date}(m,d),
  \kw{prod\text{-}date}(p,d') }$ and the UCQ~$Q'$ with the following disjuncts:
\begin{align*}
  \exists xyzz' &
  (
  \kw{motor}(x)\wedge \kw{sparkplug}(y)\wedge
  \kw{prod\text{-}date}(x,z) \wedge 
  \kw{prod\text{-}date}(y,z') \wedge z < z'
  )\\
  \exists xyzz' &
  (
  \kw{motor}(x)\wedge \kw{sparkplug}(y)\wedge
  \kw{prod\text{-}date}(x,z) \wedge 
  \kw{prod\text{-}date}(y,z') \wedge z' < z
  )\\
  \exists xyz &
  (
  \kw{motor}(x)\wedge \kw{sparkplug}(y)\wedge
  \kw{prod\text{-}date}(x,z) \wedge 
  \kw{prod\text{-}date}(y,z)
  ).
\end{align*}
  The query $Q$ is entailed in $\owqalin$ over~$\instance_0'$ where~$<$ is asserted to be a total
  order, but it is not entailed in either $\owqatr$ or $\owqatc$.
\end{example}

We now define the constraint languages 
for which we study these $\owqa$ problems.
We will also give some examples of sentences in these languages in Example~\ref{ex:syntax}.

\subsection{Dependencies}
\label{sec:dependencies}

The first constraint languages that we study are
restricted classes of
\emph{tuple-generating dependencies} (TGDs).
A TGD is an $\fo$ sentence $\tau$ of the form
$\forall \vec x~ \big(\bigwedge_i \gamma_i(\vec x) \rightarrow \exists \vec y ~ \bigwedge_i
\rho_i(\vec x, \vec y)\big)$
where $\bigwedge_i \gamma_i$ and $\bigwedge_i \rho_i$ are non-empty conjunctions
of atoms,
respectively called the \emph{body} and \emph{head} of~$\tau$.
 
We will be interested in TGDs that are \emph{guarded}
in various ways.
A \emph{guard} 
for a tuple $\vec{x}$ of variables, or for an atom $A(\vec x)$,
is an atom from~$\sigma$ or an equality
using (at least) every variable in $\vec{x}$. That is, an atom where every variable of $\vec{x}$ appears as an argument.
For example, $R(z,y)$, $C(y,w,z)$, and $y=z$
are all guards for $\vec{x} = (y, z)$.
In this work, we will be particularly interested in \emph{base-guards} (sometimes denoted $\sigmab$-guards),
which are guards coming from the base relations in~$\sigmab$ or equality.

 A \emph{frontier-guarded TGD} or $\fgtgd$
 \cite{bagetcomplexityfg}
 is a TGD $\tau$ whose body contains a guard
for the \emph{frontier variables}, i.e., the variables that occur in both head and body.
We introduce the \emph{base frontier-guarded TGDs} ($\afgtgd$s) as the TGDs with
a \emph{base frontier guard}, i.e., 
an equality or \mbox{$\sigmab$-atom}
including all  the frontier variables.
We allow equality atoms $x=x$ to be guards, so $\afgtgd$ subsumes 
\emph{frontier-one TGDs}
 \cite{bagetcomplexityfg}, which have
one frontier variable.
We also introduce the more restrictive class of 
\emph{base-covered frontier-guarded TGDs} ($\acfgtgd$):
they are the $\afgtgd$s
where, for every $\sigmad$-atom $A$ in the body,
there is a base-guard $A'$ of~$A$
in the body; note that this time $A'$ may be different for each~$A$.

\emph{Inclusion dependencies} ($\incd$) are
an important special case of frontier-guarded TGDs used in many applications.
An $\incd$ is a $\fgtgd$ 
of the form $\forall \vec x ~ R(\vec x) \rightarrow \exists \vec y ~ S(\vec x,
  \vec y)$, i.e., where the body and head
contain a single atom, and where we further impose that no variable occurs twice in the same atom. A
\emph{base inclusion dependency} ($\aincd$) is an $\incd$ where
the body atom is in~$\sigmab$,
so the body atom serves as  the base-guard for the frontier
variables, and the constraint is trivially base-covered.

\subsection{Guarded Logics}
\label{sec:guardedlogics}
Moving beyond TGDs, we also study constraints coming from
\emph{guarded logics}. In particular, the \emph{guarded negation fragment} ($\gnf$) over a signature $\sigma$ is the fragment of $\fo$
given by the grammar
\[
\varphi ::= 
A(\vec{x})
~|~
x = y
~|~
\varphi \vee \varphi
~|~
\varphi \wedge \varphi
~|~
\exists \vec{x}  \ \varphi
~|~
\alpha \wedge \neg \varphi
\]
where
$A$ ranges over relations in $\sigma$
and
$\alpha$ is a guard for the free variables in $\neg \phi$.

Note that since a guard $\alpha$ in $\alpha \wedge \neg \varphi$ is allowed
to be an equality, $\gnf$  can express 
all formulas of the form $\neg \phi$
when $\phi$ has at most one free variable. In particular, if a sentence
$\phi$ is in $\gnf$, then $\neg \phi$ is expressible in $\gnf$, as
$\exists x ~ x=x \wedge \neg \phi$.

The use of these ``equality-guards'' is convenient in the proofs. But
in the presentation of examples within the paper, we do not wish to write out
these ``dummy guards'', and thus as a convention we allow in examples
unguarded subformulas $\neg \phi$ where $\phi$ has at most one free variable.

$\gnf$ can express all $\fgtgd$s since
an $\fgtgd$ of the form
$\forall \vec{x} (\bigwedge \gamma_i \rightarrow \exists \vec{y} \bigwedge \rho_i)$ can be written in $\agnf$
as
$\neg \exists \vec{x} ( \bigwedge \gamma_i  \wedge \alpha \wedge \neg \exists \vec{y} \bigwedge \rho_i )$
where $\alpha$ is the guard for the frontier variables in $\bigwedge \gamma_i$.
It can also express non-TGD constraints and UCQs.
For instance, as it allows disjunction, $\gnf$ can express
\emph{disjunctive
inclusion dependencies}, $\did$s
\cite{bourhispieris}, which  generalize $\incd$s: a $\did$ is a first-order sentence of the
  form $\forall \vec x ~ R(\vec x) \rightarrow \bigvee_{1 \leq i \leq n} \exists
  \vec{y_i} ~ S_i(\vec x,
  \vec{y_i})$ such that, for every~$1 \leq i \leq n$, the sentence $\forall \vec x ~ R(\vec x)
  \rightarrow \exists \vec{y_i} ~ S_i(\vec x, \vec{y_i})$ is an $\incd$. In particular, any
  $\incd$ is a $\did$, as is seen by taking $n=1$ in the disjunction.

In this work, we introduce the \emph{base-guarded negation fragment} $\agnf$ over $\sigma$: it
is defined
like $\gnf$, but requires \emph{base-guards} instead of guards.
The \emph{base-covered guarded negation fragment} $\acgnf$
over $\sigma$
consists of $\agnf$ formulas such that
every $\sigmad$-atom $A$ that appears negatively (i.e., under the scope of an odd
number of negations) appears in conjunction with a
base-guard for its variables.
This condition is designed to generalize $\acfgtgd$s; indeed, any
$\acfgtgd$ can be expressed in $\acgnf$.

We call a CQ $Q$ \emph{base-covered} if, for each $\sigmad$-atom $A$ in $Q$,
there is a base-guard $A'$ of~$A$ in~$Q$.
A UCQ is \emph{base-covered} if each disjunct is.
Note that every base-covered UCQ can easily be rewritten in $\acgnf$.

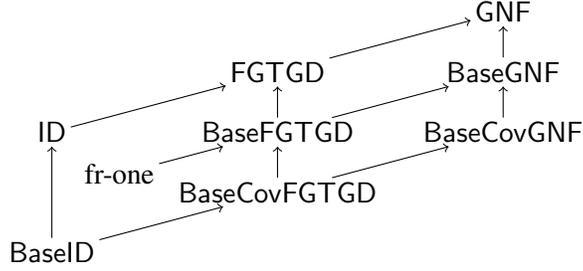
\begin{figure}[t]
\centering
\begin{tikzpicture}[xscale=3,yscale=.8,inner sep=1.5pt]
\node (GNF) at (4, 0) {$\gnf$};
\node (AGNF) at (4, -1) {$\agnf$};
\node (ACGNF) at (4, -2) {$\acgnf$};
\node (FGTGD) at (3, -1) {$\fgtgd$};
\node (AFGTGD) at (3, -2) {$\afgtgd$};
\node (ACFGTGD) at (3, -3) {$\acfgtgd$};
\node (ID) at (2, -2) {$\incd$};
\node (AID) at (2, -4) {$\aincd$};
\node (FR1) at (2.3, -2.7) {fr-one};

\draw[->] (ACGNF) -- (AGNF);
\draw[->] (AGNF) -- (GNF);
\draw[->] (ACFGTGD) -- (AFGTGD);
\draw[->] (AFGTGD) -- (FGTGD);
\draw[->] (FGTGD) -- (GNF);
\draw[->] (AFGTGD) -- (AGNF);
\draw[->] (ACFGTGD) -- (ACGNF);
\draw[->] (ID) -- (FGTGD);
\draw[->] (AID) -- (ACFGTGD);
\draw[->] (AID) -- (ID);
\draw[->] (FR1) -- (AFGTGD);
\end{tikzpicture}
\caption{Taxonomy of fragments}
\label{fig:tax}
\end{figure}

We illustrate the different constraint languages and queries by giving a few examples.

\begin{example}\label{ex:syntax}
Consider a signature with a binary base relation $B$,
a ternary base relation $C$, and a distinguished relation $D$.
\begin{itemize}
\item $\forall x y z \big( (B(x,y) \wedge B(y,z)) \rightarrow D(x,z) \big)$ is a TGD,
but is not a $\fgtgd$ since the frontier variables $x,z$ are not guarded.
It cannot even be expressed in $\gnf$.
\item $\forall x y  \big(D(x,y) \rightarrow B(x,y) \big)$
is an $\incd$, hence a $\fgtgd$.  It is not a $\aincd$ or even in $\agnf$, since the frontier
variables are not base-guarded.
\item $\forall x y z \big((B(z,x) \wedge D(x,y) \wedge D(y,z)) \rightarrow D(x,z) \big)$
is a $\afgtgd$.
However, it is not a $\acfgtgd$
since there are no base atoms in the body to cover $x,y$ and $y,z$.
\item $\exists w x y z \big( D(w,x) \wedge D(x,y) \wedge D(y,z) \wedge D(z,w) \wedge C(w,x,y) \wedge C(y,z,w) \big)$
is a base-covered CQ.
\item $\exists x y
\big(
B(x,y) \wedge \neg (D(x,y) \wedge D(y,x) ) \wedge  (D(x,y) \vee D(y,x) )
\big)$ cannot be rewritten as a TGD. But it can be rewritten in $\acgnf$
as
\begin{align*}
\exists x y
\big( \quad
&\big[ (B(x,y) \wedge \neg D(x,y))
\vee
( B(x,y) \wedge \neg D(y,x) ) \big] \\
\ \wedge \
&\big[ (B(x,y) \wedge  D(x,y) )
\vee
( B(x,y) \wedge D(y,x) ) \big]
\qquad \big) .
\end{align*}
\end{itemize}
\end{example}
\subsection{Normal Form}\label{app:nf}
\newcommand{\nf}{\text{normal form}\xspace}

The fragments of $\gnf$
that we consider can be converted into a normal form
that is related to the $\gn$ normal form
introduced in the original paper on $\gnf$ \cite{gnficalp}.
The idea is that $\gnf$ formulas can be seen as being built
up by nesting UCQs using guarded negation.
We introduce this normal form here,
and discuss related notions
that we will use in the proofs.

The \emph{$\nf$ for $\agnf$} over $\sigma$
can be defined recursively as the formulas of the form
$\delta[Y_1 \mapsfrom \alpha_1 \wedge \neg \phi_1, \ldots, Y_n \mapsfrom \alpha_n \wedge \neg \phi_n]$
where
\begin{itemize}
\item $\delta$ is a UCQ over signature $\sigma \cup \{Y_1,\ldots,Y_n\}$
for some fresh relations $Y_1,\dots,Y_n$,
\item $\phi_1, \ldots, \phi_n$ are in $\nf$ $\agnf$ over $\sigma$, and
\item $\alpha_1, \ldots, \alpha_n$ are base-guards for
the free variables in $\phi_1, \ldots, \phi_n$ such that the number of free variables in each $\alpha_i \wedge \neg \phi_i$ matches the arity of $Y_i$, and
\item $\delta[Z \mapsfrom \psi]$ is the result of replacing every occurrence of $Z(\vec{x})$ in $\delta$ with $(\psi(\vec{x}))$.
\end{itemize}
The base case of this recursive definition is a UCQ over $\sigma$ (take $n=0$ above).

In other words, formulas in $\nf$ $\agnf$ are built up from UCQs (over both base and distinguished relations) using base-guarded negation.
We also refer to these as \emph{UCQ-shaped formulas}.
A single disjunct of a UCQ-shaped formula is called a \emph{CQ-shaped formula}.
Note that an atomic formula can be seen as a simple UCQ with no disjunction, projection or negation.

The \emph{$\nf$ for $\acgnf$}
over $\sigma$
consists of $\nf$ $\agnf$ formulas such that
for every CQ-shaped subformula $\delta[Y_1 \mapsfrom \alpha_1 \wedge \neg \phi_1, \ldots, Y_n \mapsfrom \alpha_n \wedge \neg \phi_n]$ that appears negatively
(in the scope of an odd number of negations),
and for every distinguished atom $\alpha'$ that appears as a conjunct in $\delta$,
there must be some base-guard for the free variables of $\alpha'$ that appears in $\alpha_1,\dots,\alpha_n$ or as a conjunct in $\delta$.
Formulas in $\nf$ $\acgnf$ might not syntactically satisfy the condition that
every distinguished atom appears in (direct) conjunction with a base atom using its variables,
but it can always be converted into one satisfying this condition with only a linear blow-up in size by duplicating guards:
e.g., if $D$ is the only distinguished relation,
$\neg \exists x y z ( D(x,y) \wedge D(y,z) \wedge C(x,y,z) )$
could be converted to $\neg \exists x y z ( C(x,y,z) \wedge D(x,y) \wedge C(x,y,z) \wedge D(y,z) )$.
We allow this slightly more relaxed definition for $\nf$ $\acgnf$ since it is more natural
when talking about formulas built up using UCQ-shaped subformulas.

We revisit some of the sentences from Example~\ref{ex:syntax} to see how to
rewrite them in \nf.
\begin{example}
The $\afgtgd$ $\forall x y z \big((B(z,x) \wedge D(x,y) \wedge D(y,z)) \rightarrow D(x,z) \big)$
can be expressed in $\nf$ $\agnf$ as
$\neg \exists x y z \big(D(x,y) \wedge D(y,z) \wedge (B(z,x) \wedge \neg D(x,z)) \big) .$
We now explain how this is built following the definition of \nf $\agnf$.
We first build the inner CQ-shaped formula by
taking the CQ $\exists x y z \big(D(x,y) \wedge D(y,z) \wedge Y(x,z) \big)$
and substituting $B(z,x) \wedge \neg D(x,z)$ for $Y(x,z)$.
We can then build the final formula by substituting
$\exists x y z \big(D(x,y) \wedge D(y,z) \wedge (B(z,x) \wedge \neg D(x,z)) \big)$
for the 0-ary relation $Z$
in
$\neg Z$.

The sentence $\exists x y
\big(
B(x,y) \wedge \neg (D(x,y) \wedge D(y,x) ) \wedge  (D(x,y) \vee D(y,x) )
\big)$ can
be expressed in $\nf$ $\acgnf$ as
$\delta[Y_1 := B(x,y) \wedge \neg D(x,y), Y_2 := B(x,y) \wedge \neg D(y,x)]$
where
\begin{align*}
\delta = \qquad &\exists x y (B(x,y) \wedge D(x,y) \wedge Y_1(x,y))
\vee
\exists x y (B(x,y) \wedge D(y,x) \wedge Y_1(x,y)) \\
\vee \ &\exists x y (B(x,y) \wedge D(x,y) \wedge Y_2(x,y))
\vee
\exists x y (B(x,y)  \wedge D(y,x) \wedge Y_2(x,y)) .
\end{align*}
\end{example}

\subsubsection{Width, CQ-rank, and Negation Depth}
For $\varphi$ in $\nf$ $\agnf$, we define the
\emph{width} of $\varphi$ to be the maximum number
of free variables in any subformula of $\varphi$.
Note that by reusing variable names, a formula of width $k$ can always be written using only $k$ variable names.
  The \mbox{\emph{CQ-rank}} of~$\varphi$ is
the maximum number of conjuncts in any CQ-shaped subformula
$\exists \vec{x} ( \bigwedge \gamma_i )$.
The \emph{negation depth} is the maximal nesting depth of negations in the syntax tree.
These parameters will be important in later proofs.

\subsection{Conversion into Normal Form}
Observe that formulas in $\afgtgd$ or $\acfgtgd$ are of the form
$\forall \vec{x} (\bigwedge \gamma_i \rightarrow \exists \vec{y} \bigwedge \rho_i)$
and so can be naturally written in $\nf$ $\agnf$ or $\acgnf$ as
$\neg \exists \vec{x} ( \bigwedge \gamma_i  \wedge (\alpha \wedge \neg \exists \vec{y} \bigwedge \rho_i) )$
where $\alpha$ is the base-guard from $\bigwedge \gamma_i$ for the frontier-variables.
In this case, there is a linear blow-up in the size.
In general, $\agnf$ formulas can be converted into $\nf$,
but with an exponential blow-up in size:

\begin{proposition}\label{prop:nf}
Let $\varphi$ be a formula in $\agnf$.
We can construct an equivalent $\varphi'$ in $\nf$
in $\exptime$ such that
\begin{inparaenum}[(i)]
\item $\mysize{\varphi'}$ is at most exponential in $\mysize{\varphi}$,
\item the width of $\varphi'$ is at most $\mysize{\varphi}$,
\item the CQ-rank of $\varphi'$ is at most $\mysize{\varphi}$,
\item if $\varphi$ is in $\acgnf$, then $\varphi'$ is in $\nf$ $\acgnf$.
\end{inparaenum}
\end{proposition}

\begin{proof}[Proof sketch]
The conversion works by using the rewrite rules given by \citeA{gnficalp}:
\begin{align*}
\exists x (\theta \vee \psi) &\rwto (\exists x \theta) \vee (\exists x \psi)
  \qquad &\theta \wedge (\psi \vee \chi) &\rwto (\theta \wedge \psi) \vee (\theta \wedge \chi) \\
\exists x (\theta) \wedge \psi &\rwto \exists x' (\theta[x'/x] \wedge \psi ) \text{ where $x'$ is fresh}
\end{align*}
It is straightforward to check the bounds on the size, width, and CQ-rank after performing this rewriting.
Coveredness is also preserved since
\begin{inparaenum}[(i)]
\item the rewrite rules do not change the polarity of any subformulas,
so any distinguished atom in $\varphi'$ that appears negatively---and hence requires a base-guard---can be associated
with a distinguished atom in $\varphi$ that appears negatively and appears in conjunction
with a base-guard for its free variables,
and
\item the rewrite rules do not separate a conjunction of two atoms,
so the base-guard that appears in conjunction with a distinguished atom in $\varphi$ is propagated to any occurrence
of this distinguished atom in $\varphi'$.
\end{inparaenum}
\end{proof}

\subsection{Automata-Related Tools}\label{sec:automata-tools}
In Section~\ref{sec:decid},
we will use automata running on infinite binary trees,
so we briefly recall some definitions and key properties.
Our presentation is partially adapted from \citeA{lics16-gnfpup};
for more background please refer to Appendix~\ref{app:closure}
or surveys by \citeA{Thomas97} and \citeA{Loding-unpublished}.
In particular, we will need to use 2-way automata
that can move both up and down as they process
the tree, so we highlight some
less familiar properties
about the relationship between 2-way and 1-way versions
of these automata.
For readers not interested in the details of the automaton construction in Section~\ref{sec:aut-construction},
this section can be skipped.

\subsubsection{Trees}
The input to the automata will be infinite full binary trees $\tree$
over some \emph{tree signature} $\autsig$ consisting of a set of unary relations.
That is, each node~$v$
has exactly two children (one left child, and one right),
and has a label $\tree(v) \in \powerset{\autsig}$ that indicates the set of unary relations that hold
at~$v$.
We will assume that $\Gamma$ always includes unary relations $\dleft$ and $\dright$,
and the label of each node correctly indicates whether the node is the left or right child of its parent
(with the root being the unique node with neither $\dleft$ nor $\dright$ in its label).
We will also identify each node in a binary tree with a finite string over $\{ 0 , 1 \}$,
with $\epsilon$ identifying the root,
and $u 0$ and $u 1$ identifying the left child and right child of node $u$.

\subsubsection{Tree Automata}
We define a set of \emph{directions} 
$\Dir \colonequals \set{ \dleft,\dright,\dup,\dstay}$,
and write $\cB^+(X)$ for any set~$X$ to denote the set of positive Boolean
formulas over~$X$.
A \emph{2-way alternating parity tree automaton} ($\ptwowayaltinf$) $\cA$
is then a tuple $\perm{\autsig,Q,q_0,\delta,\Omega}$
where
$\autsig$ is a tree signature,
$Q$ is a finite set of states,
$q_0 \in Q$ is the initial state,
$\delta : Q \times \powerset{\autsig} \to \cB^{+}(\Dir \times Q)$
is the transition function,
and $\Omega : Q \to P$
is the priority function with a finite set of \emph{priorities} $P \subseteq \N$.
Intuitively, the transition function $\delta$ maps a state and a set of unary relations from $\autsig$
to a positive Boolean formula
over $\Dir \times Q$
that indicates possible next moves for the automaton.

Running the automaton $\cA$ on some input tree $\tree$
is best thought of in terms of an \emph{acceptance game} or
\emph{membership game} (see \citeR{Loding-unpublished} for more information).
The positions in the game are of the form $(q,v) \in Q \times \tree$.
In position $(q,v)$,
Eve and Adam play a subgame based on $\delta(q,\tree(v))$,
with Eve resolving disjunctions and Adam resolving conjunctions
until an atom $(d,q')$ in $\delta(q,\tree(v))$ is selected.
Then the game continues from
position $(q',v')$
where $v'$ is the node in direction $d$ from~$v$; in particular
$v' \colonequals v$ if $d = \dstay$.
For example,
if $\delta(q,\tree(v)) = (\dup,s_1) \vee ( (\dstay, s_2) \wedge (\dright, s_3) )$,
then the acceptance game starting from $(q,v)$ would work as follows:
Eve would select one of the disjuncts;
if she selects the first disjunct then the game would continue from $(s_1, u)$
where $u$ is the parent of $v$,
otherwise,
Adam would choose one of the conjuncts
and the game would continue from $(s_2,v)$ or $(s_3, v1)$
depending on his choice.

A play $(q_0,v_0) (q_1,v_1) \dots$
is a sequence of positions in such a game.
The play is winning for Eve if it
satisfies the \emph{parity condition}: the maximum priority occurring infinitely often
in $\Omega(q_0) \Omega(q_1) \dots$ is even.
A \emph{strategy} for Eve is a function that, given the history of the play
and the current position $(q,v)$ in the game, determines Eve's choices in the subgame
based on $\delta(q,\tree(v))$.
Note that we allow the automaton to be started from arbitrary positions in the tree,
rather than just the root.
We say that $\cA$ \emph{accepts $\tree$ starting from $v_0$} if Eve has a strategy
such that all plays consistent with the strategy starting from $(q_0,v_0)$ are winning.
$L(\cA)$ denotes the \emph{language} of trees accepted by $\cA$ starting from the root.

A 1-way alternating automaton is an automaton that processes the tree in a top-down fashion,
using only directions $\dleft$ and $\dright$.
A (1-way) nondeterministic automaton is a 1-way alternating automaton
such that every transition function formula is of the form $\bigvee_j \, (\dleft,q_j) \wedge (\dright,r_j)$.
We call such automata \emph{1-way nondeterministic parity tree automata},
and write it $\ponewayndinf$.

We review some closure properties of these automata in Appendix~\ref{app:automata}.

\subsubsection{Connections between 2-way and 1-way Automata}
It was shown by \citeA{Vardi98} that every $\ptwowayaltinf$ can be converted to an equivalent
$\ponewayndinf$,
with an exponential blow-up.

\begin{theorem}[\citeR{Vardi98}]\label{thm:2way-to-nd}
Let $\cA$ be a $\ptwowayaltinf$.
We can construct a $\ponewayndinf$ $\cA'$
such that $L(\cA) = L(\cA')$.
The number of states of $\cA'$ is exponential in the number of states of $\cA$,
and the number of priorities of $\cA'$ is polynomial in the number of states and priorities of $\cA$.
The running time of the algorithm is polynomial in the size of $\cA'$, and is
  hence in $\exptime$ in the input.
\end{theorem}

1-way nondeterministic tree automata can be seen as a special case of 2-way alternating automata,
so the previous theorem shows that $\ponewayndinf$ and $\ptwowayaltinf$ are
equivalent, in terms of their ability to recognize trees starting from the root.

We need another conversion from 1-way nondeterministic automata
to 2-way alternating automata,
that we call \emph{localization}.
This process takes a 1-way nondeterministic automaton that runs on trees
with extra information about membership in certain relations (annotated on the tree), 
and converts that automaton to an equivalent 2-way alternating automaton
that operates on trees without these annotations,
under the assumption that these relations hold only locally at the
position the 2-way automaton is launched from.
A similar localization theorem is present in prior work \cite{BourhisKR15,lics16-gnfpup}.
We sketch the idea here, and provide more details about the construction in Appendix~\ref{app:automata}.

\newcommand{\localization}
{
Let $\autsig' := \autsig \cup \sset{P_1,\dots,P_j}$.
Let $\cA'$ be a $\ponewayndinf$ on $\autsig'$-trees.
We can construct a $\ptwowayaltinf$ $\cA$ on $\autsig$-trees
such that
for all $\autsig$-trees $\tree$ and nodes $v$ in the domain of~$\tree$,
\begin{align*}
&\text{$\cA'$ accepts $\tree'$ from the root}
\ \text{ iff } \
\text{$\cA$ accepts $\tree$ from $v$}
\end{align*}
where $\tree'$ is the $\autsig'$-tree obtained from $\tree$ by setting $P_1^{\tree'} =  \dots = P_j^{\tree'} = \set{v}$.
The number of states of $\cA$ is linear in the number of states of $\cA'$,
the number of priorities of $\cA$ is linear in the number of priorities of
  $\cA'$,
and the overall size of $\cA$ is linear in the size of $\cA'$.
The running time is polynomial in the size of $\cA'$,
and hence is in $\ptime$.
}
\begin{theorem}\label{thm:localization}
\localization
\end{theorem}

\begin{proof}[Proof sketch]
$\cA$ simulates $\cA'$ by guessing in a backwards fashion an initial part of a run of $\cA'$
on the path from $v$ to the root and then processing the rest of the tree in a normal downwards fashion.
The subtlety is that the automaton $\cA$
is reading a tree without valuations for $P_1,\dots,P_j$
so once the automaton leaves node $v$, if it were to cross this node again,
it would be unable to correctly simulate~$\cA'$.
To avoid this issue, we only send downwards copies of the automaton in directions
that are not on the path from the root to $v$.
\end{proof}

\section{Decidability Results for Transitivity} \label{sec:decid}

We are now ready to explore query answering for $\agnf$
when the distinguished relations are transitively closed or are the transitive closure
of certain base relations.
We show that these query answering problems
can be reduced to tree automata emptiness testing.

\subsection{Deciding $\owqatc$ Using Automata}
\label{sec:decidtransautomata}

We first consider $\owqatc$,
where
$\sigmab$ includes binary relations $R_1, \dots, R_n$,
and $\sigmad$ consists of binary relations $R^\trans_1, \dots, R^\trans_n$
such that $R^\trans_i$ is the transitive
closure of $R_i$ for each~$i$.

\begin{theorem}
  \label{thm:decidtransautomata}
We can decide $\owqatc(\instance_0, \Sigma,Q)$ in $\twoexp$,
where $\instance_0$ ranges over finite sets of facts,
$\Sigma$ over $\agnf$ constraints, and $Q$ over UCQs.
In particular, this holds when $\Sigma$ consists of $\afgtgd$s.
\end{theorem}
In order to prove Theorem~\ref{thm:decidtransautomata},
we give a decision procedure to determine whether
$\instance_0 \wedge  \Sigma \wedge
\neg Q$
is unsatisfiable,
when $R^\trans_i$ is interpreted as the transitive closure of $R_i$.
When $\Sigma \in \agnf$ and $Q$ is a Boolean UCQ, then
$\Sigma \wedge \neg Q$ is in $\agnf$.
So it suffices to show that $\agnf$ satisfiability is decidable in $\twoexptime$,
when properly interpreting $R^\trans_i$.

As mentioned in the introduction, our proofs rely heavily on the fact that
in query answering problems for the constraint languages that we study, one
can restrict to sets of facts that have a ``tree-like'' structure. We now make this notion precise.
A \emph{tree decomposition} of $\instance$ consists
of a directed graph $(T, \child)$, where $\child \subseteq T \times T$ is the
edge relation, and of a labelling function $\lambda$ associating
each node of $T$ to a set of elements of $\instance$ and a set of facts over
these elements, called the \emph{bag} of that node. We impose the following
conditions on tree decompositions:
\begin{inparaenum}[(i)]
\item $(T, \child)$ is a tree;
\item each fact of $\instance$ must be in the image of $\lambda$;
\item for each element $e \in \elems{\instance}$,
the set of nodes that is associated with $e$ by $\lambda$ forms a connected
subtree of~$T$.
\end{inparaenum}
A tree decomposition is \emph{$\instance_0$-rooted} if the root node is associated with $\instance_0$.
It has \emph{width} $k-1$
if each bag other than the root is associated with at most $k$ elements.

For a number $k$, a 
sentence~$\phi$ over~$\sigma$ is said to have
\emph{transitive-closure friendly $k$-tree-like witnesses} if:
for every finite set of facts $\instance_0$,
if there is a set of facts $\instance$
(finite or infinite)
extending $\instance_0$ with additional $\sigmab$-facts
such that $\instance$ satisfies $\phi$
when each $R^\trans$ is interpreted as the transitive closure of $R$,
then there is such an $\instance$
that has an $\instance_0$-rooted $(k-1)$-width
tree decomposition with countable branching
(i.e., each node has a countable number of children).
Note that, in such an~$\instance$,
the only $\sigmad$-facts explicitly appearing in $\instance$
are from $\instance_0$,
and we require of~$\instance$ that these distinguished facts are actually
part of the transitive closure of the corresponding base relation.
Other $\sigmad$-facts may be implied by $\sigmab$-facts,
and both the explicit and implicit $\sigmad$-facts
must be considered when reasoning about~$\varphi$.
However, we emphasize that besides the $\sigmad$-facts in $\instance_0$,
there are no $\sigmad$-facts appearing in the tree decomposition---the explicit inclusion of such $\sigmad$-facts could make
it impossible to find a $k$-tree-like witness.

\begin{figure}[t]
\centering
\begin{tikzpicture}
  [
  yscale=1.18,xscale=1.1,
  mosaic/.style={draw,rectangle,rounded corners,minimum height = 2em,minimum width = 2em,font=\small},
  elim/.style={draw,cross out,color=red,minimum height=2em,minimum width=2em,thick},
  link/.style={->},
  linkdots/.style={-,dotted}
  ]
\node[mosaic,minimum size=3em,minimum height = 3.5em] (t1) at (0,0) {};
\node[mosaic,minimum size=3em,minimum height = 3.5em ] (t2) at (2,-2) {};
\node[mosaic,minimum size=3em,minimum height = 3.5em ] (t3) at (-2,-2) {};
\node[mosaic,minimum size=3em,minimum height = 3.5em ] (t3p) at (0,-2) {};
\node[mosaic,minimum size=3em,minimum height = 3.5em ] (t4) at (-2,-4) {};
\node[minimum size=3em] (tdot1) at (-1.5,-4) {};
\node[minimum size=3em] (tdot2) at (-.5,-4) {};
\node[minimum size=3em] (tdot3) at (1.5,-4) {};
\node[minimum size=0em] (tdot4) at (.25,-5.5) {};
\node[minimum size=0em] (tdot5) at (.75,-5.5) {};
\node[anchor=north west] (t1name) at (t1.north west) {};
\node[font=\tiny] at (t1) {
\begin{tabular}{c}
$B(a,c)$ \\
$B(b,c)$ \\
$R(a,b)$ \\
$B(a,d)$ \\
$R^\trans(a,c)$
\end{tabular}
};
\node[font=\tiny] at (t3p) {
\begin{tabular}{c}
$C(g,h,c)$ \\
$R(c,h)$ \\
$R(g,h)$
\end{tabular}
};
\node[font=\tiny] at (t2) {
\begin{tabular}{c}
$B(j,j)$ \\
$R(h,h)$ \\
$R(j,h)$
\end{tabular}
};
\node[font=\tiny] at (t3) {
\begin{tabular}{c}
$B(c,c)$ \\
$R(b,e)$
\end{tabular}
};
\node[font=\tiny] at (t4) {
\begin{tabular}{c}
$B(c,f)$ \\
$R(e,f)$ \\
$R(f,c)$
\end{tabular}
};
\draw[link] (t1.south) -- (t2.north);
\draw[link] (t1.south) -- (t3.north);
\draw[link] (t3.south) -- (t4.north);
\draw[link] (t1.south) -- (t3p.north);
\end{tikzpicture}
\hfill
\begin{tikzpicture}
  [
  xscale=.9,
  yscale=.8,
  mosaic/.style={draw,rectangle,rounded corners,minimum height = 2em,minimum width = 2em,font=\small},
  elim/.style={draw,cross out,color=red,minimum height=2em,minimum width=2em,thick},
  link/.style={->},
  linkdots/.style={-,dotted}
  ]
\node[mosaic,minimum size=3em,minimum height=3.5em ] (t) at (0,0) {};
\node[mosaic,minimum size=3em,minimum height = 3.5em ] (t0) at (-2,-2) {};
\node[mosaic,minimum size=3em,minimum height = 3.5em ] (t1) at (2,-2) {};
\node[mosaic,minimum size=3em,minimum height = 3.5em ] (t10) at (1,-4) {};
\node[mosaic,minimum size=3em,minimum height = 3.5em ] (t11) at (3,-4) {};
\node[mosaic,minimum size=3em,minimum height = 3.5em] (t110) at (2.15,-6) {};
\node[mosaic,minimum size=3em,minimum height = 3.5em] (t111) at (3.85,-6) {};
\node[mosaic,minimum size=3em,minimum height = 3.5em] (t00) at (-3,-4) {};
\node[mosaic,minimum size=3em,minimum height = 3.5em] (t01) at (-1,-4) {};
\node[minimum size=5em] (tdot) at ($(t00.south)+(0,-.25)$) {$\vdots$};
\node[minimum size=5em] (tdot) at ($(t01.south)+(0,-.25)$) {$\vdots$};
\node[minimum size=5em] (tdot) at ($(t10.south)+(0,-.25)$) {$\vdots$};
\node[minimum size=5em] (tdot) at ($(t10.south)+(0,-.25)$) {$\vdots$};
\node[minimum size=5em] (tdot) at ($(t110.south)+(0,-.25)$) {$\vdots$};
\node[minimum size=5em] (tdot) at ($(t111.south)+(0,-.25)$) {$\vdots$};
\node[font=\tiny] at (t) {
\begin{tabular}{c}
$B_{a,c}$ \\
$B_{b,c}$ \\
$R_{a,b}$ \\
$B_{a,d}$ \\
$R^\trans_{a,c}$
\end{tabular}
};
\node[font=\tiny] at (t1) {
\begin{tabular}{c}
$B_{a,c}$ \\
$B_{b,c}$ \\
$R_{a,b}$ \\
$B_{a,d}$ \\
$R^\trans_{a,c}$
\end{tabular}
};
\node[font=\tiny] at (t11) {
\begin{tabular}{c}
$B_{a,c}$ \\
$B_{b,c}$ \\
$R_{a,b}$ \\
$B_{a,d}$ \\
$R^\trans_{a,c}$
\end{tabular}
};
\node[font=\tiny] at (t111) {
\begin{tabular}{c}
$B_{a,c}$ \\
$B_{b,c}$ \\
$R_{a,b}$ \\
$B_{a,d}$ \\
$R^\trans_{a,c}$
\end{tabular}
};
\node[font=\tiny] at (t10) {
\begin{tabular}{c}
$C_{1,2,c}$ \\
$R_{c,2}$ \\
$R_{1,2}$
\end{tabular}
};
\node[font=\tiny] at (t0) {
\begin{tabular}{c}
$B_{c,c}$ \\
$R_{b,1}$
\end{tabular}
};
\node[font=\tiny] at (t01) {
\begin{tabular}{c}
$B_{c,c}$ \\
$R_{b,1}$
\end{tabular}
};
\node[font=\tiny] at (t00) {
\begin{tabular}{c}
$B_{c,3}$ \\
$R_{1,3}$ \\
$R_{3,c}$
\end{tabular}
};
\node[font=\tiny] at (t110) {
\begin{tabular}{c}
$B_{1,1}$ \\
$R_{2,2}$ \\
$R_{1,2}$
\end{tabular}
};
\draw[link] (t.south) -- (t0.north);
\draw[link] (t.south) -- (t1.north);
\draw[link] (t0.south) -- (t00.north);
\draw[link] (t0.south) -- (t01.north);
\draw[link] (t1.south) -- (t10.north);
\draw[link] (t1.south) -- (t11.north);
\draw[link] (t11.south) -- (t110.north);
\draw[link] (t11.south) -- (t111.north);
\end{tikzpicture}
\caption{An $\instance_0$-rooted tree decomposition, and part of its encoding (see Examples~\ref{ex:tree}~and~\ref{ex:encoding})}
\label{fig:tree}
\end{figure}
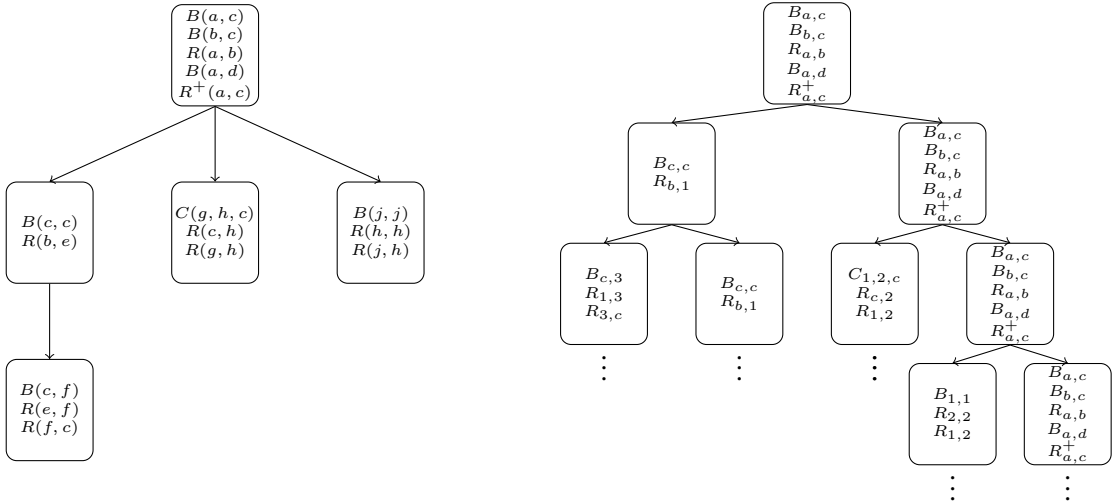

\begin{example}\label{ex:tree}
  Let $\instance_0 = \set{B(a,c), B(b,c), R(a,b), B(a,d), R^\trans(a,c)}$.
Figure~\ref{fig:tree} shows an $\instance_0$-rooted tree decomposition
for a set $\instance$ of facts extending $\instance_0$.
The width of the tree decomposition is~2.
$\instance$ satisfies $\instance_0 \wedge \Sigma \wedge \neg Q$
where
\begin{align*}
\Sigma &= \set{\exists x y z (R(x,y) \wedge R(y,z) \wedge \neg R^\trans(z,y) )} \\
Q &= \exists x y  (B(x,y) \wedge R(x,y)) .
\end{align*}
For instance, $R^\trans(a,c) \in \instance_0$ is satisfied in $\instance$ because of the chain of facts
$R(a,b)$, $R(b,e)$, $R(e,f)$, $R(f,c)$.
Observe that the only transitive facts that explicitly appear
in the tree decomposition are at the root, but other transitive facts
are implied by the facts (e.g., $R^\trans(b,f)$) and must be taken into account
when reasoning about $\instance_0 \wedge \Sigma \wedge \neg Q$.
\end{example}

We can show that $\agnf$ sentences have these transitive-closure friendly $k$-tree-like witnesses
for an easily computable $k$. 
The proof uses a standard technique,
involving an unravelling based on a form of ``guarded negation bisimulation'',
so we defer the proof of this result to Appendix~\ref{app:tctreelike}.
The result does not follow directly from the fact that $\gnf$ has tree-like witnesses \cite{gnficalp} since
we must show that this unravelling preserves $\agnf$ sentences
even when interpreting each distinguished $R^\trans$ as the transitive closure of $R$,
rather than just preserving $\gnf$ sentences without these special interpretations.
However, adapting their proof to our setting is straightforward.

\newcommand{\transdecomp}{
  Every sentence $\phi$ in $\agnf$ has  transitive-closure friendly $k$-tree-like witnesses,
    where $k \leq \mysize{\phi}$.
  }
\begin{proposition}
  \label{prop:transdecomp}
  \transdecomp
\end{proposition}

Hence,
it suffices to test satisfiability for $\agnf$ restricted
to sets of facts with tree decompositions
of width $\mysize{\phi}-1$.
It is well known that sets of facts of bounded tree-width
can be encoded as trees over a finite alphabet.
This makes the satisfiability problem amenable to tree automata
techniques,
since we can design a tree automaton that runs on
representations of these tree decompositions
and checks whether some sentence holds in
the corresponding set of facts.

\begin{theorem}\label{thm:automata}
Let $\phi$ be a sentence in $\agnf$ over signature $\sigma$,
and let $\instance_0$ be a finite set of $\sigma$-facts.
We can construct in $\twoexp$
a $\ptwowayaltinf$ $\calA_{\phi,\instance_0}$
such that\\[.3em]
\null\hfill$ \text{$\instance_0 \wedge \phi$ is satisfiable
\quad iff \quad
$L(\calA_{\phi,\instance_0}) \neq \emptyset$}$\hfill\null\\[.3em]
when each $R_i^\trans \in \sigmad$ is interpreted as the transitive closure
of $R_i \in \sigmab$.
The number of states and priorities\footnote{
With some additional work, it is possible to construct
an automaton with only priorities $\set{1,2}$ (also known as a B\"uchi automaton),
but this optimization is not important for our results and is omitted.}
  of $\calA_{\phi,\instance_0}$ is
at most exponential in $\mysize{\phi} \cdot \mysize{\instance_0}$.
\end{theorem}

We present the details of this construction in the next section.
It can be viewed
as an extension of work by \citeA{CalvaneseGV05},
and incorporates ideas from automata for guarded logics
by, e.g., \citeA{GradelW99}.
It can also be viewed as an optimization of the construction  by
\citeA{lics16-gnfpup},
which we discuss in Section~\ref{sec:gnfpup}.

Theorem~\ref{thm:automata} implies that
the language of the automaton $\cA_{\Sigma \wedge \neg Q,\instance_0}$ is empty iff $\owqatc(\instance_0,\Sigma,Q)$ holds.
Because 2-way tree automata emptiness is decidable in time
polynomial in the overall size and exponential in the number of states and priorities
\cite{Vardi98},
this yields the $\twoexp$ bound
for Theorem~\ref{thm:decidtransautomata}.

\subsection{Automata for $\agnf$ (Proof of Theorem~\ref{thm:automata})}\label{sec:aut-construction}

Fix the signature $\sigma = \sigmab \sqcup \sigmad$.
As described above in Proposition~\ref{prop:transdecomp},
in order to test satisfiability of sentences in $\agnf$ over $\sigma$,
it suffices to consider only
sets of facts with tree decompositions of
some bounded width.
We first describe how to encode such sets of facts
using trees over a finite tree signature,
and then describe how to construct the automaton
to prove Theorem~\ref{thm:automata}.

\subsubsection{Tree Encodings}
Consider a set of $\sigma$-facts $\instance$,
with an $\instance_0$-rooted tree decomposition of width $k-1$
with countable branching,
specified by a tree $(T,\child)$ and a function $\lambda$.
For technical reasons in the automaton construction,
it is more convenient to use binary trees,
so we want to convert to an alternative tree decomposition of $\instance$
based on $(T',\child')$ with a labelling $\lambda'$
such that $(T',\child')$ is a \emph{full binary tree}, i.e., every node has exactly two children.
This can be done by duplicating and rearranging parts of the tree.
First, for each node $u$, we add infinitely many new children to~$u$, each child
being the root of an 
infinite full binary tree where each node has the same label as~$u$ in~$T$.
This ensures that each node of~$T$ now has infinitely many (but still countably many)
children.
Second, we convert~$T$
into a full binary tree:
starting from the root, each node $u$ with children $(v_i)_{i \in \mathbb{N}}$
is replaced by
the subtree consisting of $v_1,v_2,\dots$ and new nodes $u_1, u_2, \dots$
such that
the label at each $u_i$ is the same as the label at $u$,
the left child of $u_i$ is $v_i$ and the right child of $u_i$ is $u_{i+1}$.
In other words, instead of having a node $u$ with infinitely many children $(v_i)_{i \in \mathbb{N}}$,
we create an infinite spine of nodes with the same label as $u$,
and attach each $v_i$ to a different copy $u_i$ of $u$ on this spine.

Now we can start encoding this infinite full binary tree decomposition.
To achieve this, we specify a finite set $\paramsk$ of \emph{names} that
can be used to describe the possibly infinite number of elements in $\instance$;
we will fix the size of this set momentarily.
We include $\elems{\instance_0}$ in $\paramsk$;
these are precisely the names used to describe elements from the initial set of facts $\instance_0$.
Then we map the elements in $\elems{\instance} \setminus \elems{\instance_0}$ to a name in $\paramsk$
such that the following condition is satisfied:
if $u$ and $v$ are neighboring nodes of $T$,
then distinct elements of
$\elems{\lambda(u)} \cup \elems{\lambda(v)}$
are mapped to distinct names in $u$ and $v$.
Note that every bag either has names $\elems{\instance_0}$
or has at most $k$ names.
Hence,
letting $l \colonequals \mysize{\elems{\instance_0}}$, we know that
$2k+l$ possible names suffice
to be able to choose different names
for distinct elements in neighboring nodes,
in a way that does not conflict with the names of elements in~$\instance_0$.
So we choose $U$ to be of size $2k + l$.
This assignment of names is encoded
using unary relations $D_a$ for each $a \in \paramsk$,
so that
$D_a(v)$ holds iff $a$ is a name that was assigned to an element in $v$.
Facts in $\instance$ are encoded
using unary relations $R_{\vec{a}}$ for each
$R \in \sigma$ of arity $n$ and each $n$-tuple $\vec{a} \in \paramsk^n$,
so that
$R_{\vec{a}}(v)$ holds iff
$R$ holds of the tuple of elements named by $\vec{a}$ at $v$.

\begin{example}\label{ex:encoding}
Figure~\ref{fig:tree} shows an encoding of the tree decomposition of $\instance$
  from Example~\ref{ex:tree} (omitting the $D_a$-relations).
\end{example}

The encodings will sometimes need to specify a valuation
for free variables in a formula,
so we also introduce relations for this purpose.
Recall that we can assume that formulas of width $k$
use some fixed set of $k$ variable names.
For each such variable $z$ and each $c \in \paramsk$,
we introduce a relation $V_{c / z}$;
if $V_{c / z}(v)$ holds, then this indicates that the valuation for $z$
is the element named by $c$ at~$v$.
We refer to these relations that give a valuation for the free variables
as \emph{free variable markers}.

As mentioned in Section~\ref{sec:automata-tools}, we also assume that there are unary relations
$\dleft$ and $\dright$ to indicate whether a node is a left or right child of its parent.

This concludes the definition of our encoding scheme.
We let $\sigcode{\sigma}{k,l}$ denote the \emph{encoding signature} containing the relations described above,
and we use the term \emph{$\sigcode{\sigma}{k,l}$-tree}
to refer to an infinite full binary tree over the tree signature $\sigcode{\sigma}{k,l}$.

\subsubsection{Tree Decodings}
If a $\sigcode{\sigma}{k,l}$-tree satisfies certain consistency properties,
then it can be decoded into a set of $\sigma$-facts
that extends $\instance_0$.

Formally, let $\bagnames{v} := \set{ a \in \paramsk : \text{$D_a (v)$ holds}}$ be the
set of names used for elements in bag~$v$ in some tree;
we will abuse notation and write $\vec{a} \subseteq \bagnames{v}$ to
mean that $\vec{a}$ is a tuple over names from $\bagnames{v}$.
Then a \emph{consistent tree} $\tree$ with respect to $\sigcode{\sigma}{k,l}$ and $\instance_0$
is a $\sigcode{\sigma}{k,l}$-tree such that:
\begin{enumerate}[(i)]
\item coded facts respect the domain:
for all $R_{\vec{a}} \in \sigcode{\sigma}{k,l}$ and for all nodes $v$, if $R_{\vec{a}} (v)$ holds then $\vec{a} \subseteq \bagnames{v}$;
\item there is a bijection between the elements and facts represented at the root node and
the elements and facts in $\instance_0$:
\begin{inparaenum}[]
\item $\bagnames{\epsilon} = \elems{\instance_0}$;
\item for each fact $R(c_1\dots c_n) \in \instance_0$,
the fact $R_{c_1 \dots c_n}(\epsilon)$ holds in the tree;
\item for every $R_{c_1 \dots c_n}(\epsilon)$,
the fact $R(z_1 \dots z_n)$ is in $\instance_0$;
\end{inparaenum}
\item if there is a free variable marker for $z$, then it is unique:
for each variable $z$, there is at most one node $v$ and one name $c \in \paramsk$
such that $V_{c/z}(v)$ holds.
\end{enumerate}

Given a consistent tree $\tree$,
we say nodes $u$ and $v$ are \emph{$a$-connected}
if
there is a sequence of nodes $u = w_0, w_1, \dots, w_j = v$
such that $w_{i+1}$ is a neighbor (child or parent) of $w_i$,
and $a \in \bagnames{w_{i}}$ for all $i \in \set{0,\dots,j}$.
We write $[v,a]$ for the equivalence class of $a$-connected nodes of $v$.
For $\vec{a} = a_1 \dots a_n$,
we often abuse notation and write $[v,\vec{a}]$ for the tuple 
$[v,a_1],\dots,[v,a_n]$.

The \emph{decoding} of $\tree$ is the
set of $\sigma$-facts $\mydecode{\tree}$
using elements $\set{ [v,a] : \text{$v \in \tree$, $a \in \bagnames{v}$}}$,
where we identify $c \in \elems{\instance_0}$ with $[\epsilon,c]$.
For each relation $R \in \sigma$ of arity $j$,
we have $R([v_1,a_1],\dots,[v_j,a_j]) \in \mydecode{\tree}$ iff
there is some $w \in \tree$ such that
$R_{a_1,\dots,a_j}(w)$ holds and $[w,a_i] = [v_i,a_i]$ for all~$1 \leq i \leq j$.
Note that a single fact $R([v_1,a_1],\dots,[v_j,a_j])$ in $\decode(\tree)$ might be coded in multiple nodes in $\tree$,
but there is no requirement that this fact is coded in all nodes that
represent $[v_1,a_1], \dots, [v_j,a_j]$.

\subsubsection{Automaton Construction}
The goal is to build an automaton for a sentence $\varphi$ in $\agnf$
as needed for Theorem~\ref{thm:automata}.
An \emph{automaton for a formula $\psi(x_1,\dots,x_n)$} in $\agnf$ is an automaton $\cA_\psi$ such that
for all $\sigcode{\sigma}{k,l}$-trees $\tree$
with free variable markers $V_{a_1/x_1}(v_1), \dots, V_{a_n/x_n}(v_n)$,
the automaton $\cA_\psi$ accepts $\tree$ iff
$\mydecode{\tree}$ satisfies
$\psi([v_1,a_1],\dots,[v_n,a_n])$
when each $R^\trans \in \sigmad$ is interpreted as the transitive closure of $R \in \sigmab$.
As a warm-up and an example,
we start by constructing a $\ptwowayaltinf$ $\cB_{R^\trans(x_1,x_2)}$
for an atomic formula $R^\trans(x_1,x_2)$ using $R^\trans \in \sigmad$.

\begin{example}\label{ex:aut}
The $\ptwowayaltinf$ $\cB_{R^\trans(x_1,x_2)}$
runs on $\sigcode{\sigma}{k,l}$-trees
with free variable markers for $x_1$ and $x_2$.

We define the state set to be
$(U \times \set{\choosenext,\findend}) \cup \set{\findstart,\accsink,\rejsink}$.
The idea is that Eve must navigate to the node $v_1$
carrying the free variable marker for~$x_1$,
find a series of $R$-facts,
and then show that the last element on this guessed $R$-path
corresponds to the element with the free variable marker for $x_2$.
The initial state is $\findstart$, when she is finding the marker for $x_1$.
The states of the form $(a,\choosenext)$
are used to track that $a$ is the name of the furthest element on the $R$-path that has been found so far.
The state
$(a,\findend)$
is used when she is trying to show that the element with name $a$
is identified by the free variable marker for $x_2$.
If the automaton is in a state of the form $(a,\findend)$ or $(a,\choosenext)$
and moves to a node where $a$ is not represented ($D_a$ is not in the label),
then she moves to a sink state $\rejsink$.
The state $\accsink$ is a sink state used when she successfully finds an $R$-path.

This is implemented by the following transition function, which describes how the automaton should behave in a state
from $(U \times \set{\choosenext,\findend}) \cup \set{\findstart,\accsink,\rejsink}$ and in a node with label $\treelab$:
\begin{align*}
\delta(\findstart,\treelab) &:=
\begin{cases}
(\dstay, (a,\choosenext)) \quad \text{if $V_{a/x_1} \in \treelab$ } \\
(\dup,\findstart) \vee (\dleft,\findstart) \vee (\dright,\findstart) \quad \text{otherwise}
\end{cases}
\\
\delta((a,\choosenext),\treelab) &:=
\begin{cases}
(\dstay,\rejsink) \quad \text{if $D_a \notin \treelab$}
\\
\begin{aligned}
&\textstyle\bigvee \set{ (\dstay, (a',\choosenext)) \vee  (\dstay,(a',\findend)) : \text{$R_{a,a'} \in \treelab$}} \\
&\quad \vee (\dup,(a,\choosenext)) \vee (\dleft,(a,\choosenext)) \vee (\dright,(a,\choosenext))
\end{aligned} \quad \text{otherwise}
\end{cases}
\\
\delta((a,\findend),\treelab) &:=
\begin{cases}
(\dstay,\rejsink)  \quad \text{if $D_a \notin \treelab$} \\
(\dstay,\accsink) \quad \text{if $D_a \in \treelab$ and $V_{a/x_2} \in \treelab$ } \\
(\dup,(a,\findend)) \vee (\dleft,(a,\findend)) \vee (\dright,(a,\findend)) \quad \text{otherwise}
\end{cases}
\\
\delta(\accsink,\treelab) &:= (\dstay,\accsink) \qquad
\delta(\rejsink,\treelab) := (\dstay,\rejsink)
\end{align*}
Only two priorities are needed.
The state $\accsink$ is assigned priority 0;
all of the other states are assigned priority~1.
This prevents Eve from cheating and forever delaying
her choice of the elements in the $R$-path.
\end{example}

As another building block, the next lemma describes how to construct an automaton for a CQ.
This lemma is stated for a CQ over signature $\sigma'$ rather than just $\sigma$;
the reason for this will become clear when we use this in the inductive case in Lemma~\ref{lemma:aut}.
As usual, we assume that $\sigma'$ is partitioned into base relations $\sigmapb$ and distinguished relations $\sigmapd$.

\begin{lemma}\label{lemma:autcq}
Given a CQ $\chi(x_1,\dots,x_j) = \exists \vec{y} ( \eta(x_1,\dots,x_j,\vec{y}) )$
of width at most $k$
over signature $\sigma'$,
and given a natural number $l$,
we can construct a
  $\ponewayndinf$ $\calN_{\chi}$ for the formula $\chi$ (over
  $\sigcode{\sigma'}{k,l}$-trees).

Furthermore, there is a polynomial function $g$ independent of~${\chi}$
such that the number of states of~$\calN_{\chi}$
is at most $2^{g(K r_{\chi})}$
and the number of priorities is at most $g(K r_{\chi})$,
where
$r_{\chi}$ is the CQ-rank of~${\chi}$ (i.e., the number of conjuncts in $\eta$), and $K = 2k+l$.
The overall size of the automaton
and the running time of the construction
is at most exponential in $\mysize{\sigma'} \cdot 2^{g(K r_{\chi})}$.
\end{lemma}

\begin{proof}
Each conjunct of $\eta$ is an atomic formula.
For each such atomic formula $\psi$,
we will first describe a
$\ptwowayaltinf$ $\cB_\psi$ that runs
on trees with the free variable markers for $\vec{x}$ and $\vec{y}$
written on the tree.
We describe these informally (in terms of choices by Adam and Eve),
but they could all be translated into formal automaton definitions
in the style of Example~\ref{ex:aut}.

  \begin{itemize}
    \item \emph{Base atom.}
Suppose $\psi$ is a $\sigmapb$-atom $A(\vec{z})$,
where $\vec{z}$ is a tuple of variables from $\vec{x}$ and $\vec{y}$.
Eve tries to
navigate to a node $v$ whose label includes fact $A_{\vec{b}}$.
If she is able to do this,
Adam can then challenge Eve to show that $[v,\vec{b}]$ is the valuation for $\vec{x}$.
Say he challenges her on $b_i \in \vec{b}$.
Then Eve must navigate from $v$ to
the node carrying the marker $b_i/x_i$.
However, she must do this
by passing through a series of nodes
that also contain $b_i$.
If she is able to do this,
$\cB_\psi$ enters a sink state with priority 0, so she wins.
The other states are assigned priority~1
to force Eve to actually witness $A(\vec{z})$.
The number of states of $\cB_\psi$ is linear in $K$,
since the automaton must remember the name $b_i$ that
Adam is challenging.
Only two priorities are needed.
\item \emph{Equality.}
Suppose $\psi$ is an equality $x_1 = x_2$.
Eve navigates to the node $v$ with a free variable marker $a/x_1$.
She is then required to navigate from $v$ to
the node carrying the marker for $x_2$.
She must do so by passing through a series of nodes
that also contain $a$.
If she is able to reach the marker $a/x_2$ in this way,
then $x_1$ and $x_2$ are marking the same element
in the underlying set of facts,
so $\cB_\psi$ moves to a sink state with priority 0 and she wins.
The other states have priority 1, so if Eve is not able to do this,
then Adam wins.
The state set is of size linear in $K$, in order to remember the name~$a$.
There are two priorities.
\item \emph{Distinguished atom.}
Suppose $\psi$ is a $\sigmapd$-atom $R^\trans(x_1,x_2)$.
Then we use the construction in Example~\ref{ex:aut}.
The number of states in $\cB_\psi$ is again linear in $K$,
since it must remember the name $a$
that is currently being processed along this path.
There are only two priorities.
  \end{itemize}
This means that we can construct a $\ptwowayaltinf$ for each conjunct in $\eta$.
We can also easily construct a $\ptwowayaltinf$ that checks
that there is a unique free variable marker for each variable in $\vec{y}$;
the number of states of this automaton is linear in $K$,
since we must remember which variable and name we are checking.
By using closure under intersection of $\ptwowayaltinf$ (see Proposition~\ref{prop:closure-union-intersection} in Appendix~\ref{app:automata}),
we can construct an automaton $\cB_{\eta}$ that checks that
$\eta(\vec{x},\vec{y})$ holds and that there is a unique free variable marker for each variable in $\vec{y}$.
This can be converted to a $\ponewayndinf$ $\calN_{\eta}$ using Theorem~\ref{thm:2way-to-nd},
such that the number of states is exponential in $Kr_{\eta}$,
and the number of priorities is polynomial in $Kr_{\eta}$.

Finally, the desired $\ponewayndinf$ $\calN_{\chi}$ for $\chi(\vec{x}) = \exists \vec{y} ( \eta( \vec{x},\vec{y}) )$
can be constructed from $\calN_{\eta}$ as follows:
it runs on trees with the free variable markers for $\vec{x}$,
and simulates $\calN_{\eta}$ while allowing Eve to guess the free variable marker for each variable in $\vec{y}$
(in other words, it is the projection of $\calN_{\eta}$ with respect to the free variable markers for $\vec{y}$;
see Proposition~\ref{prop:closure-projection} in Appendix~\ref{app:automata}).
The idea is that Eve tries to guess a valuation for $\vec{y}$ that satisfies $\eta$.
Note that for this guessing procedure to be correct,
it is essential that $\calN_{\eta}$ is a 1-way nondeterministic automaton,
rather than an alternating automaton.
The construction of $\calN_{\chi}$ from $\calN_{\eta}$ can be done in
polynomial time, with no increase in the number of states or priorities.

Overall, this means that there is some polynomial function $g$ independent of~${\chi}$
such that the number of states of~$\calN_{\chi}$
is at most $2^{g(K r_{\chi})}$
and the number of priorities is at most $g(K r_{\chi})$
where 
$r_{\chi}$ is the CQ-rank of~${\chi}$, and $K = 2k+l$.
It can be checked that the overall size of the automaton
and the running time of the construction
is at most exponential in $\mysize{\sigma'} \cdot 2^{g(K r_{\chi})}$.
\end{proof}

Note that in the base cases of our construction we utilized a simple
kind of $\ptwowayaltinf$ with only two priorities.
However, in applying Theorem~\ref{thm:2way-to-nd} we will increase the number of priorities,
and thus we are really utilizing the power of $\ptwowayaltinf$ in this construction.

This shows that we can construct an automaton for a CQ over $\sigma$,
and this could be used to achieve the desired $\twoexptime$ bound
for satisfiability testing of a CQ.
This $\twoexptime$ bound can be extended to a UCQ,
simply by using closure under union for $\ponewayndinf$
(see Proposition~\ref{prop:closure-union-intersection} in Appendix~\ref{app:automata}).
However, $\agnf$ allows nesting of UCQs with base-guarded negation.
The fear is that if we iterate this process for each level of nesting,
we will get an exponential blow-up each time,
which would lead to non-elementary complexity for satisfiability testing
of $\agnf$.

In order to avoid these additional exponential blow-ups,
we take advantage of the fact that the nesting of UCQs
allowed in $\agnf$ is restricted:
the free variables in a nested UCQ-shaped formula
must be base-guarded,
and hence must be represented locally
in a single node in the tree code.
Recall that a UCQ-shaped formula in $\agnf$ with negation depth greater than 0 is of the form
$\delta[Y_1 := \alpha_1 \wedge \neg \psi_1, \dots , Y_s := \alpha_s \wedge \neg \psi_s]$
where $\delta$ is a UCQ over the extended signature~$\sigma'$ obtained from
$\sigma$ by adding fresh base relations $Y_1, \dots, Y_s$,
each $\psi_i$ is a UCQ-shaped formula in $\agnf$,
and each $\alpha_i$ is a base-guard in $\sigma$ for the free variables of $\psi_i$.
We first construct an automaton for $\delta$, over the extended signature $\sigma'$.
The automaton for $\delta[Y_1 := \alpha_1 \wedge \neg \psi_1, \dots , Y_s := \alpha_s \wedge \neg \psi_s]$ can then
simulate the automaton for $\delta$
while allowing Eve to guess the valuations
for each $Y_i$ (i.e., valuations for
each base-guarded subformula $\alpha_i \wedge \neg \psi_i$).
In order to prevent Eve from cheating
and just guessing that every tuple satisfies these subformulas,
Adam is allowed to challenge her on these guesses
by launching automata for these subformulas.

The technical difficulty here is that the free variable markers for these
subformulas are not written on the tree code any more---they are being guessed on-the-fly by Eve.
In order to cope with this, the inductive process will construct an
automaton for $\psi$
that can be launched from some internal node $v$ in the tree 
to test whether or not $\psi$ holds with a local valuation $[v,\vec{a}]$ for $\vec{x}$.
The automaton will not specify a single initial state.
Instead,
there will be a designated initial state for each
polarity $p \in \set{+,-}$ and
each
possible ``local assignment''
for the free variables $\vec{x}$.
A \emph{local assignment} $\vec{a}/\vec{x}$
for $\vec{a} = a_1 \dots a_n \in \paramsk^n$ and $\vec{x} = x_1 \dots x_n$ is a mapping
such that $x_i \mapsto a_i$.
A node~$v$ in a consistent tree $\tree$ with $\vec{a} \subseteq \bagnames{v}$
and a local assignment $\vec{a}/\vec{x}$,
specifies a valuation for~$\vec{x}$.
We say it is local since the free variable markers for $\vec{x}$
would all appear locally in $v$.
Given a polarity $p \in \set{+,-}$,
we write $p \psi$ for $\psi$ if $p = +$
and $\neg \psi$ if $p = -$.

We will write $\cA_\psi$ for the automaton for $\psi$
(without specifying the initial state),
and will write $\cA_{\psi}^{p,\vec{a}/\vec{x}}$ for $\cA_{\psi}$
with the designated initial state for $p$ and $\vec{a}/\vec{x}$.
We call $\cA_\psi$ a \emph{localized automaton for $\psi$},
since when it is launched from a node $v$ starting from the designated state for $p$ and $\vec{a}/\vec{x}$,
it is testing whether $p\psi$ holds when the valuation for $\vec{x}$ is
$[v,\vec{a}]$, which is represented locally at $v$.
The point of localized automata is that they can test
whether a tuple of elements that appear together in a node satisfy some
property, but without having the markers for this tuple explicitly written on the
tree. This allows us to ``plug-in'' localized automata for base-guarded subformulas,
as described in the following lemma.

\newcommand{\localsubstitution}{
Let $\eta$ be in $\agnf$ over $\sigma \cup \set{Y_1,\dots,Y_s}$.
Let $\cA_\eta$ be a localized $\ptwowayaltinf$ for $\eta$
over $\sigcode{\sigma \cup \set{Y_1,\dots,Y_s}}{k,l}$-trees.

For $1 \leq i \leq s$, let $\chi_i := \alpha_i \wedge \neg \psi_i$ be a formula in $\agnf$ over $\sigma$
with the number of free variables in $\chi_i$ matching the arity of $Y_i$,
and let $\cA_{\chi_i}$ be a localized $\ptwowayaltinf$ for $\chi_i$
over $\sigcode{\sigma}{k,l}$-trees.

We can construct a localized $\ptwowayaltinf$ $\cA_{\psi}$
for $\psi := \eta[Y_1 := \chi_1,\dots,Y_s := \chi_s]$
over $\sigcode{\sigma}{k,l}$-trees
in linear time
such that the number of states (respectively, priorities) is
the sum of the number of states (respectively, priorities)
of $\cA_\eta, \cA_{\chi_1}, \dots, \cA_{\chi_s}$.
}
\begin{lemma}\label{lemma:localsubstitution}
  \localsubstitution
\end{lemma}

We defer the proof of this lemma to Appendix~\ref{app:automataproof}.
With the help of Lemma~\ref{lemma:autcq} and
Lemma~\ref{lemma:localsubstitution},
we can now prove our main lemma.

\begin{lemma}\label{lemma:aut}
Given a $\nf$ formula $\psi(\vec{x})$ in $\agnf$ of width at most $k$ over signature $\sigma$
and given a natural number $l$, we can construct a localized
$\ptwowayaltinf$ $\calA_\psi$
such that for all consistent $\sigcode{\sigma}{k,l}$-trees~$\tree$,
for all polarities $p \in \set{+,-}$, for all local assignments $\vec{a}/\vec{x}$,
and for all nodes $v$ in $\tree$ with $\vec{a} \subseteq \bagnames{v}$,
\begin{align*}
&\text{$\calA_\psi^{p,\vec{a}/\vec{x}}$ accepts $\tree$
starting from $v$} \quad
\text{iff} \quad \text{$\decode(T)$ satisfies $p\psi([v,\vec{a}])$}
\end{align*}
when each $R^\trans \in \sigmad$ is interpreted as the transitive closure
of $R \in \sigmab$.

Further, there is a polynomial function $f$ independent of~$\psi$
such that the number of states of~$\calA_\psi$
is at most $N_\psi := f(m_\psi) \cdot 2^{f(K r_\psi)}$
and the number of priorities is at most $f(K m_\psi)$,
where $m_\psi = \mysize{\psi}$, 
$r_\psi$ is the CQ-rank of~$\psi$, and $K = 2k+l$.
The overall size of the automaton
and the running time of the construction
is at most exponential in $\mysize{\sigma} \cdot N_\psi$.
\end{lemma}

\begin{proof}
We proceed by induction on the negation depth $d$ of the $\nf$ formula $\psi(\vec{x})$ in $\agnf$.
We write $m_\psi$ for $\mysize{\psi}$,
write $r_\psi$ for the CQ-rank of $\psi$,
  and write \mbox{$N_\psi\colonequals f(m_\psi) \cdot 2^{f(Kr_\psi)}$}
  for some suitably chosen (in particular, non-constant) polynomial $f$
independent of $\psi$
(we will not define $f$ explicitly).

During each case of the inductive construction,
we will describe informally how to build the desired automaton,
and we will analyze the number of priorities
and the number of states required.
We defer the analysis of the overall size
of the automaton until the end of this proof.

\mysubparagraph{Negation depth 0}

For the base case of a UCQ $\psi(\vec{x})$ (negation depth 0),
we apply Lemma~\ref{lemma:autcq}
to obtain a $\ponewayndinf$ for each CQ,
and then use closure under union to obtain a $\ponewayndinf$ $\calN_{\psi}$ for $\psi$.
This automaton has number of priorities at most polynomial in $K m_\psi$
and number of states at most polynomial in $m_\psi$ and exponential in $Kr_\psi$,
as desired.
However, this automaton runs on trees with the free variable markers for $\vec{x}$,
so it remains to show that we can construct the automaton $\cA_\psi$ required by
the lemma,
that runs on trees without these markers.

For each local assignment $\vec{a}/\vec{x}$,
we can apply the localization theorem (Theorem~\ref{thm:localization})
to the set of relation symbols of the form $V_{a_i/x_i}$,
and eliminate the dependence on any other
$V_{c/x_i}$ for $c \neq a_i$, by always assuming these relations do not hold.
This results in an automaton $\cA_{\psi}^{+,\vec{a}/\vec{x}}$ that
no longer relies on free variable markers for $\vec{x}$.
By Theorem~\ref{thm:localization},
there is only a linear blow-up in the number of states and number of priorities.
Then let $\cA_{\psi}^{-,\vec{a}/\vec{x}}$ be the dual of $\cA_{\psi}^{+,\vec{a}/\vec{x}}$,
obtained using closure under complement of $\ptwowayaltinf$ (see Proposition~\ref{prop:dual} in Appendix~\ref{app:automata}).
Finally, we take $\cA_\psi$ to be the disjoint union of~$\cA_{\psi}^{+,\vec{a}/\vec{x}}$
and $\cA_{\psi}^{-,\vec{a}/\vec{x}}$
over all local assignments $\vec{a}/\vec{x}$;
the designated initial state for each $p \in \set{+,-}$ and each localization $\vec{a}/\vec{x}$
is the initial state for $\cA_\psi^{p,\vec{a}/\vec{x}}$.
All of these automata can be seen to use priorities from the same set,
which is of size at most polynomial in $K m_\psi$.
For a suitably chosen $f$, the number of states in $\cA_\psi$ can be bounded by $f(m_\psi) \cdot 2^{f(K r_\psi)}$
(since there are at most $2K^k$ subautomata being combined)
and the number of priorities can be bounded by $f(K m_\psi)$
(since all of the different localized automata $\cA_\psi^{p,\vec{a}/\vec{x}}$ being combined use the same set of
priorities, which is polynomial in $K m_\psi$).

\mysubparagraph{Negation depth $d > 0$}
The inductive step is for a UCQ-shaped formula $\psi$ with negation depth $d > 0$.
Suppose $\psi$ is of the form
$\delta[Y_1 := \alpha_1 \wedge \neg \psi_1, \dots , Y_s := \alpha_s \wedge \neg \psi_s]$
where $\delta$ is a UCQ over the extended signature $\sigma'$ obtained from
$\sigma$ by adding fresh base relations $Y_1, \dots, Y_s$,
each $\psi_i$ is a UCQ-shaped formula in $\agnf$ with negation depth strictly less than $d$
and each $\alpha_i$ is a base-guard in~$\sigma$ for the free variables of $\psi_i$.

Since each $\alpha_i$ and $\psi_i$ have strictly smaller negation depth,
the inductive hypothesis yields localized $\ptwowayaltinf$ $\cA_{\alpha_i}$ and $\cA_{\psi_i}$ for $0 \leq i \leq s$.
As a step on the way to constructing an automaton for $\psi$,
we can use these inductively-defined automata to construct a localized $\ptwowayaltinf$ $\cA_{\phi_i}$
for each $\phi_i := \alpha_i \wedge \neg \psi_i$. Let us fix $0 \leq i \leq s$.
We start by taking the disjoint union of $\cA_{\alpha_i}$ and $\cA_{\psi_i}$.
We then add 
a fresh state $q_0^{p,\vec{a}/\vec{x}}$ to the new automaton
for each polarity $p \in \set{+,-}$ and each localization $\vec{a}/\vec{x}$, 
which becomes the designated initial state for $p$ and $\vec{a}/\vec{x}$ in the new automaton.
We keep the transition functions from the subautomata, but add in the following rules for the new states:
in state~$q_0^{+,\vec{a}/\vec{x}}$,
Adam is given the choice of staying in the current node and switching to the designated initial state for $+,\vec{a}/\vec{x}$
in $\cA_{\alpha_i}$,
or staying in the current node and switching to the designated initial state for $-, \vec{a}/\vec{x}$ in $\cA_{\psi_i}$;
likewise, in state $q_0^{-,\vec{a}/\vec{x}}$
Eve is given the choice of staying in the current node and switching to the designated initial state for $-,\vec{a}/\vec{x}$
in $\cA_{\alpha_i}$ or
or staying in the current node and switching to the designated initial state for $+, \vec{a}/\vec{x}$ in $\cA_{\psi_i}$.
The number of states is at most
\begin{align*}
&f(m_{\alpha_i})\cdot 2^{f(K r_{\alpha_i})} + f(m_{\psi_i})\cdot 2^{f(K r_{\psi_i})} + 2K^k 
\leq (f(m_{\alpha_i}) + f(m_{\psi_i}) + 1) \cdot 2^{f(K r_{\alpha_i \wedge \neg \psi_i})}
\end{align*}
which is at most $N_{\alpha_i \wedge \neg \psi_i}$ by the choice of $f$.
The number of priorities is at most the sum of the number of priorities in $\cA_{\alpha_i}$ and $\cA_{\psi_i}$,
so it is bounded by $f(K m_{\alpha_i \wedge \neg \psi_i})$.

Now let $\calA_\delta$ be the $\ptwowayaltinf$ for the underlying UCQ $\delta$
obtained as described in the base case.
Even though this is over the extended alphabet $\sigma'$,
the overall size of the automaton
and the running time of the construction
is still at most exponential in $\mysize{\sigma} \cdot N_{\psi}$,
since the number of additional relations in $\sigma'$ is $s \leq \mysize{\psi}$.
We can then apply
Lemma~\ref{lemma:localsubstitution}
to obtain $\cA_\psi$
from $\cA_\delta, \cA_{\phi_1},\dots,\cA_{\phi_s}$.
Since the automata $\cA_\delta, \cA_{\phi_1},\dots,\cA_{\phi_s}$ all satisfy the desired bounds on the number of states,
the number of states is at most
\begin{align*}
&f(m_\delta) \cdot 2^{f(Kr_\delta)} + f(m_{\phi_1}) \cdot 2^{f(Kr_{\phi_1})} + \dots + f(m_{\phi_s}) \cdot 2^{f(Kr_{\phi_s})} \\
\leq \
&(f(m_\delta)+ f(m_{\phi_1}) + \cdots + f(m_{\phi_s}) ) \cdot 2^{f(Kr_\psi)}
\end{align*}
which is at most $N_\psi$.
Likewise, since the number of priorities is at most the sum of the priorities of $\cA_\delta, \cA_{\phi_1},\dots,\cA_{\phi_s}$,
the number of priorities in $\cA_\psi$ can still be bounded by $f(K m_\psi)$.

This concludes the inductive case.

\subsubsection{Overall Size}
We have argued that each automaton has at most $N_\psi$ states
and the number of priorities is polynomial in the size of $\psi$.
It remains to argue that the overall size of $\cA_\psi$ is at most
exponential in $\mysize{\sigma} \cdot N_\psi$.
The size of the priority mapping is at most polynomial in $N_\psi$.
The size of the alphabet is exponential in $\mysize{\sigma} \cdot K^k$,
which is at most exponential in~$\mysize{\sigma} \cdot N_\psi$.
For each state and alphabet symbol,
the size of the corresponding transition function formula
can always be kept of size at most exponential in $N_\psi$.
Hence, the overall size of the transition function is at most exponential 
in~$\mysize{\sigma} \cdot N_\psi$.
Together, this means that the overall size of $\cA_\psi$ is
at most exponential in $\mysize{\sigma} \cdot N_\psi$.

It can also be checked that
the running time of the construction is polynomial
in the size of the constructed automaton,
and hence is also exponential in $\mysize{\sigma} \cdot N_\psi$.
\end{proof}

We must also construct an automaton that checks that
the input tree is consistent,
and actually represents a set of facts $\instance$
such that $\instance \supseteq \instance_0$
and where every $R^\trans$-fact in $\instance_0$ is actually witnessed
by some path of $R$-facts in $\instance$.

\begin{lemma}\label{lemma:consistency-distinguished-relations}
Given a finite set of $\sigma$-facts $\instance_0$
and natural numbers $k$ and $l$,
we can construct
a $\ptwowayaltinf$ $\calA_{\instance_0}$
in time doubly exponential in $\mysize{\sigma} \cdot K$ (for $K = 2k+l$),
such that for all $\sigcode{\sigma}{k,l}$-trees~$\tree$,
\begin{align*}
\text{$\calA_{\instance_0}$ accepts $\tree$} 
\quad \text{iff} \quad 
\text{$T$ is consistent and for all facts $S(\vec{c}) \in \instance_0$: \ }
\text{$\decode(\tree), [\epsilon,\vec{c}]$ satisfies $S(\vec{x})$}
\end{align*}
when $R^\trans \in \sigmad$ is interpreted as the transitive closure
of $R \in \sigmab$.
The number of states is at most exponential in $\mysize{\sigma} \cdot K$,
the number of priorities is two,
and
the overall size is at most doubly exponential in~$\mysize{\sigma} \cdot K$.
\end{lemma}

\begin{proof}
The automaton is designed to allow
Adam to challenge some consistency condition
or a particular fact $S(\vec{c})$
in~$\instance_0$.
It is straightforward to design automata for each of the
consistency conditions,
so we omit the details.
To check some base fact $S(\vec{c})$ from $\instance_0$,
the automaton launches $\cA_{S(\vec{x})}^{\vec{c}/\vec{x}}$
(obtained from Lemma~\ref{lemma:aut}) from the root.
To check some distinguished fact $R^\trans(c_1,c_2)$,
we launch the localized version of $\cA_{R^\trans(x_1,x_2)}$ from Example~\ref{ex:aut}
(localized to $c_1/x_1, c_2/x_2$ in the root).
Although transitive facts can, in general, mention elements that are far apart in the tree code,
we know that $c_1,c_2$ are both elements of $\instance_0$, and hence are both represented locally in the root;
hence, it is possible to localize based on this.
Note that the $R$-path witnessing this fact may require elements outside of $\elems{\instance_0}$
even though $c_1$ and $c_2$ are names of elements in $\instance_0$.

It can be checked that the automaton only needs two priorities,
the number of states is exponential in $\mysize{\sigma} \cdot K$,
and the overall size is at most doubly exponential in $\mysize{\sigma} \cdot K$.
\end{proof}

\subsubsection{Concluding the Proof}

We can now conclude the proof of Theorem~\ref{thm:automata}.
We are given some sentence $\varphi$ in $\agnf$ over signature $\sigma$
and some finite set of facts $\instance_0$.
Without loss of generality, we can assume that
$\mysize{\varphi} \cdot \mysize{\instance_0}
\geq \mysize{\sigma}$.
We construct the $\nf$ $\varphi'$ equivalent to $\varphi$ in exponential time
using Proposition~\ref{prop:nf}.
Although the size of $\varphi'$ can be exponentially larger than $\varphi$,
the CQ-rank and width is at most~$\mysize{\varphi}$.
By Proposition~\ref{prop:transdecomp}, we know that it suffices to
consider only sets of facts with full, binary $\instance_0$-rooted tree decompositions of width $\mysize{\varphi}-1$,
so we can restrict to considering automata on $\sigcode{\sigma}{k,l}$-trees
for $k := \mysize{\varphi}$ and $l = \mysize{\elems{\instance_0}}$.

Hence, we apply Lemma~\ref{lemma:aut} to
$\varphi'$, $k$, and $l$,
and construct a $\ptwowayaltinf$
$\cA_{\varphi'}$ for~$\varphi'$ (and hence $\varphi$)
in time exponential in $\mysize{\sigma} \cdot N_{\varphi'}$,
which is at most doubly exponential in $\mysize{\varphi} \cdot \mysize{\instance_0}$.
The number of states and priorities in this automaton is at most
singly exponential in $\mysize{\varphi} \cdot \mysize{\instance_0}$.

Next, we apply Lemma~\ref{lemma:consistency-distinguished-relations}
to $\instance_0$, $k$, and $l$,
to get a $\ptwowayaltinf$ $\cA_{\instance_0}$.
This can be done in time doubly exponential in $\mysize{\sigma} \cdot (2k+l)$,
which is at most doubly exponential in $\mysize{\varphi} \cdot \mysize{\instance_0}$.
The automaton has two priorities and
the number of states is at most singly exponential in $\mysize{\varphi} \cdot \mysize{\instance_0}$.

Finally, we construct the desired $\ptwowayaltinf$
$\cA_{\varphi,\instance_0}$
by taking the disjoint union of
the automaton $\cA_{\instance_0}$
and $\cA_{\varphi'}$,
and giving Adam an initial choice of which of these automata to simulate.
This automaton has a non-empty language
iff
$\varphi$ is satisfiable.
Moreover, the number of states and priorities in this automaton is still at most
singly exponential in $\mysize{\varphi} \cdot \mysize{\instance_0}$.
This concludes the proof of Theorem~\ref{thm:automata}.

\subsection{Consequences for $\owqatr$ and Other Variants}\label{sec:reflexive-transitive}
We can derive results
for $\owqatrans$
by observing that the
$\owqatc$ problem subsumes it:
to enforce that $R^\trans \in \sigmad$ is transitive, simply interpret it as the
transitive closure of a relation~$R$ that is never otherwise used.
Hence:

\begin{corollary} \label{cor:decidetransgnf}
We can decide $\owqatr(\instance_0, \Sigma,Q)$ in $\twoexp$,
where $\instance_0$ ranges over finite sets of facts,
$\Sigma$ over $\agnf$ constraints (in particular, $\afgtgd$),
and $Q$ over UCQs.
\end{corollary}

In particular, this result holds for 
\emph{frontier-one TGDs} (those with a single frontier variable),
as a single variable is always base-guarded.
This answers a question posed  in prior work \cite{mugnier15}.

As mentioned in the preliminaries,
we have defined $\owqatr$ and $\owqatc$
based on transitive relations, not reflexive and transitive relations.
However, the decidability and combined complexity results described
in Theorem~\ref{thm:decidtransautomata} and
Corollary~\ref{cor:decidetransgnf}
(as well as the data complexity results that
will be described later in Theorem~\ref{thm:conptransdataupper}
and Theorem~\ref{thm:ptimetransdataupper})
also apply to the corresponding
query answering problems
when the distinguished relations
are reflexive transitive relations
or the reflexive transitive closure of some base relation.
Adapting the proofs to this case is a straightforward exercise: the only points
that need to be changed are the precise handling of distinguished atoms in
Proposition~\ref{prop:bisim-game}, and the construction of the automaton in Example~\ref{ex:aut}
(which is used in Lemmas~\ref{lemma:autcq}~and~\ref{lemma:consistency-distinguished-relations}).

\subsection{Relationship to $\gnfpup$}\label{sec:gnfpup}

It is well-known that the transitive closure of a binary relation
can be expressed in \emph{least fixpoint logic} ($\kw{LFP}$),
the extension of $\fo$ with a least fixpoint operator.
$\kw{LFP}$ can also express that a relation is transitively closed,
or is the transitive closure of another relation.
Unfortunately, satisfiability is undecidable for $\fo$ and hence $\kw{LFP}$, so
it is not possible to rely on this connection to prove decidability of
$\owqatr$ or $\owqatc$.
On the other hand,
the fixpoint extension of $\gnf$ (called $\kw{GNFP}$) is decidable,
but it is unable to express transitive closure
(see \citeR{gnficalp,lics16-gnfpup}),
so it also cannot be used to decide $\owqatr$ or $\owqatc$.

Recently, a new fixpoint logic called
$\gnfpup$---\emph{guarded negation fixpoint logic with unguarded parameters}---was introduced 
\cite{lics16-gnfpup}.
This logic subsumes $\gnf$ and $\kw{GNFP}$,
and is expressive enough to define the transitive closure of a binary relation.
It also subsumes a number of other highly expressive Datalog-like languages
introduced by \citeA{BourhisKR15}.
However, unlike $\kw{LFP}$, satisfiability for $\gnfpup$ is decidable
\cite{lics16-gnfpup}.
Hence, $\owqatc(\instance_0,\Sigma,Q)$ for $\Sigma \in \agnf$
can be decided by converting $\instance_0 \wedge \Sigma \wedge \neg Q$
to an equivalent $\gnfpup$ sentence~$\varphi$,
and then testing for unsatisfiability of $\varphi$.

Each $\agnf$ sentence with distinguished relations $R_i^+$
can be converted to an equivalent $\gnfpup$ sentence
with only base relations,
since each occurrence of $R_i^{+}$ is replaced with a fixpoint formula
that describes the transitive closure of the corresponding base relation $R_i$.
The fixpoints in this formula are not nested in complicated ways;
using the terminology of \citeA{lics16-gnfpup},
they have ``parameter-depth'' 1.
Applying Theorem 20 in \citeA{lics16-gnfpup},
this means that $\owqatc$ is decidable in $\kw{3}\exptime$.
Thus, the approach using $\gnfpup$
gives an alternative proof of the decidability
of query answering with transitivity,
but without the optimal $\twoexptime$ complexity bound
presented here.
The automaton construction in this paper
can be viewed as an optimization of the automaton construction for $\gnfpup$
by \citeA{lics16-gnfpup}.
The results on $\gnfpup$, however, imply that
query answering is decidable for $\agnf$
not only when we have distinguished  $R^{+}_i$,
but when the distinguished relations are defined
by regular expressions over base binary relations and their inverses,
in the spirit of $\kw{C2RPQ}$s
(see Example 4 in \citeR{lics16-gnfpup}).

Due to the syntactic restrictions in $\gnfpup$,
the translation described above
would not produce a $\gnfpup$ formula if distinguished relations
were used as guards (i.e.~if we started with $\gnf$,
rather than $\agnf$).
This makes sense, since we will see in Section~\ref{sec:undecid} that 
our query answering problems
become undecidable when the distinguished relations
are allowed as guards.

\subsection{Data Complexity}
Our results in Theorem~\ref{thm:decidtransautomata} and
Corollary~\ref{cor:decidetransgnf} show upper bounds on the \emph{combined
complexity} of the $\owqatr$ and $\owqatc$ problems. We now turn to the complexity when
the query and constraints are fixed but the initial set of facts varies---the \emph{data complexity}.

We first show
a $\conp$ data complexity upper bound for $\owqatc$ for $\agnf$ constraints.

\newcommand{\conptransdataupper}{
  For any fixed $\agnf$ constraints $\Sigma$ and UCQ $Q$, given a finite
  set of facts $\instance_0$, we can decide $\owqatc(\instance_0, \Sigma, Q)$ in
  $\conp$ data complexity.
}
\begin{theorem}
  \label{thm:conptransdataupper}
  \conptransdataupper
\end{theorem}

The algorithm is based on an idea found in earlier $\conp$ data complexity
bounds, used in particular 
for a guarded variant of fixpoint logic by \citeA{vldb12}.
From Proposition~\ref{prop:transdecomp},
we know that a counterexample to $\owqatc$ for
any sentence  $\phi$ and any initial set of facts $\instance_0$
can be taken to have an $\instance_0$-rooted tree  decomposition.
While such a decomposition could be large, we can show
that to determine whether it satisfies $\phi$  it suffices
to look at a small amount of information concerning it, in the form of 
 annotations describing, for each $\mysize{\phi}$-tuple $\vec c$ in 
$\instance_0$,  sufficiently many formulas holding in the  subtree that interfaces with $\vec c$.  We can show a ``decomposition theorem''  showing that 
checking~$\phi$ on a decomposition is equivalent to checking
another sentence~$\phi'$ on the ``abstract description of
the tentacles'' with these annotations. Thus instead of guessing
a witness structure, we simply guess a consistent
set of  annotations and check that $\phi'$ holds. We defer the details
to Appendix~\ref{apx:conptransdataupperproof}.

For $\fgtgd$s, the data complexity of $\owqa$ is in $\ptime$ \cite{bagetcomplexityfg}.
We can show that the same holds, but only for $\acfgtgd$s, and for
$\owqatrans$ rather than $\owqatc$:
\newcommand{\ptimetransdataupper}{
  For any fixed $\acfgtgd$ constraints $\Sigma$ and base-covered UCQ $Q$, given a finite
  set of facts $\instance_0$, we can decide $\owqatrans(\instance_0, \Sigma, Q)$ in
  $\ptime$ data complexity.
}
\begin{theorem}
  \label{thm:ptimetransdataupper}
  \ptimetransdataupper
\end{theorem}
The proof uses a reduction to the standard $\owqa$ problem for
$\fgtgd$s, and then applies the $\ptime$ result of \citeA{bagetcomplexityfg}. The reduction again makes use of tree-likeness to show that we can replace the requirement  that the 
$R_i^\trans$ are transitive by the weaker requirement of transitivity within small sets (intuitively, within 
bags of a decomposition).
We will also use this idea for linear orders (see Proposition~\ref{prop:rewritelin}),
so we defer the proof of this result to
Appendix~\ref{apx:ptimetransdataupperproof}.

As we will see in Section~\ref{sec:hardness},
restricting to $\owqatrans$ in this result is in fact essential to make
data complexity tractable, as hardness holds otherwise.

\section{Decidability Results for Linear Orders} \label{sec:decidlin}

\renewcommand{\drel}{<}

We now move to $\owqalin$,
the setting where the distinguished relations $<_i$ of~$\sigmad$ are \emph{linear}
(total) strict orders, i.e., they are transitive, irreflexive, and total.

Unlike the previous section which used base-frontier-guarded constraints,
we restrict to \emph{base-covered constraints and queries} in this section.
We do this because we will see in Section~\ref{sec:undecidlin}
that $\owqalin$ is undecidable if we allow base-frontier-guarded constraints,
in contrast to the decidability results for $\owqatc$ and $\owqatr$ with 
such constraints.

Our main result will again be decidability of $\owqalin$ with these
additional conditions, but
the proof techniques differ from the previous section: instead of using automata,
we reduce $\owqalin$ to a traditional query answering problem (without distinguished relations),
by approximating the linear order axioms in $\gnf$.

\subsection{Deciding $\owqalin$ by Approximating Linear Order Axioms}
\label{sec:decidingapprox}

We prove the following result about the decidability
and combined complexity of $\owqalin$:

\begin{theorem}
  \label{thm:decidelindirect}
We can decide $\owqalin(\instance_0, \Sigma, Q)$ in $\twoexp$,
where
$\instance_0$ ranges over finite sets of facts,
$\Sigma$ over $\acgnf$,
and $Q$ over base-covered UCQs.
In particular, this holds when $\Sigma$ consists of $\acfgtgd$s.
\end{theorem}

Our technique here is to reduce
$\owqalin$ with $\acgnf$ constraints to traditional $\owqa$ for $\gnf$
constraints.
This implies decidability in $\twoexp$ using prior results
on $\owqa$ \cite{vldb12}.
The reduction is quite simple, and hence could be applicable to   other
constraint classes: it simply adds additional constraints
that enforce the linear order conditions.
However, as we cannot express transitivity or totality in $\gnf$,
we only add
a weakening of these properties that is expressible in $\gnf$,
and then argue that this is sufficient for our purposes.
The reduction is described in the following proposition.
\begin{proposition}
  \label{prop:rewritelin}
  For any finite set of facts $\instance_0$,
  constraints $\Sigma \in \acgnf$,
  and
  base-covered UCQ $Q$,
  we
  can compute $\instance_0'$ and $\Sigma' \in \agnf$ in $\ptime$
  such that we have
  $\owqalin(\instance_0, \Sigma, Q)$ iff we have
  $\owqa(\instance_0', \Sigma', Q)$.
\end{proposition}

Specifically,
$\instance_0'$ is $\instance_0$ together with facts $G(a,b)$ for every pair $a,b \in \elems{\instance_0}$,
where $G$ is some fresh binary base relation.
We define $\Sigma'$ as $\Sigma$ together with the \emph{$k$-guardedly linear axioms}
  for each distinguished relation~$\drel$, where $k$ is $\max(\mysize{\Sigma \wedge \neg
  Q}, \arity{\sigma \cup \{G\}})$; namely:

\begin{itemize}
\item guardedly total: 
$\forall x y ~ ( (\guardedbg(x,y) \wedge \neg (x = y)) \rightarrow x \drel y \vee y \drel x )$
\item irreflexive: 
$\neg \exists x ~ (x \drel x)$
\item $k$-guardedly transitive:
for $1 \leq l \leq k-1$: 
$\neg \exists x y ~ ( \psi_l(x,y) \wedge \guardedbg(x,y) \wedge \neg (x \drel
    y))$, and for $1 \leq l \leq k$:
$\neg \exists x ~ ( \psi_l(x,x) \wedge x = x \wedge \neg (x \drel x) )$
\end{itemize}
where:
\begin{itemize}
  \item $\guardedbg(x,y)$ is the formula expressing that $x,y$ is guarded by a relation in $\sigmab \cup \set{G}$
  (so it is an existentially-quantified disjunction over all possible atoms using a relation from $\sigmab \cup \set{G}$ and containing $x$ and $y$);
  \item $\psi_1(x,y)$ is just $x < y$; and
  \item $\psi_l(x,y)$ for $l \geq 2$ is:
$\exists x_2 \dots x_{l} ( x \drel x_2 \wedge \dots \wedge x_{l} \drel y
)$.
\end{itemize}
Unlike the property of being a linear order,
the $k$-guardedly linear axioms can be expressed in $\agnf$.
The idea is that these axioms are strong enough to enforce conditions about
transitivity and irreflexivity within ``small'' sets of elements---intuitively,
within sets of at most $k$ elements
that appear together in some bag of a $(k-1)$-width tree decomposition.

We now sketch the argument for the correctness
of the reduction.
The easy direction is where we assume $\owqa(\instance_0',\Sigma',Q)$ holds,
so any $\instance' \supseteq \instance_0'$ satisfying $\Sigma'$ must satisfy $Q$.
In this case, consider $\instance \supseteq \instance_0$ that satisfies $\Sigma$
and where all $\drel$ in $\sigmad$
are strict linear orders.
We must show that $\instance$ satisfies~$Q$.
First, observe that $\instance$ satisfies $\Sigma'$
since the \mbox{$k$-guardedly linear} axioms for $\drel$
are clearly satisfied for all~$k$
when $\drel$ is a strict linear order.
Now consider the extension of $\instance$ to $\instance'$
with facts $G(a,b)$ for all $a,b \in \elems{\instance_0}$.
This must still satisfy $\Sigma'$:
adding these facts
means there are additional $k$-guardedly linear requirements
on the elements from~$\instance_0$,
but these requirements already hold
since $\drel$ is a strict linear order.
Hence, by our initial assumption, $\instance'$ must satisfy $Q$.
Since $Q$ does not mention $G$,
the restriction of $\instance'$ back to~$\instance$ still satisfies $Q$ as well.
Therefore, $\owqalin(\instance_0,\Sigma,Q)$ holds.

For the harder direction, we prove the contrapositive of the implication,
namely, we suppose 
that $\owqa(\instance_0',\Sigma',Q)$ does not hold and
show that $\owqalin(\instance_0,\Sigma,Q)$ does not hold either.
From our assumption, there is some counterexample $\instance' \supseteq \instance_0'$ such that
$\instance'$ satisfies $\Sigma' \wedge \neg Q$.
We will again rely on the ability to  restrict to tree-like $\instance'$,
but with
a slightly different notion of tree-likeness.

We say a set $E$ of elements from $\elems{\instance}$
are \emph{base-guarded}
in $\instance$
if $\mysize{E} \leq 1$ or there is some $\sigmab$-fact or $G$-fact in $\instance$ that mentions all of the elements in $E$.
A \emph{base-guarded-interface tree decomposition} $(T, \child ,\lambda)$ for $\instance$
is a tree decomposition satisfying the following additional property:
for all nodes $n_1$ that are not the root of $T$,
if $n_2$ is a child of $n_1$
and $E$ is the set of elements mentioned in both $n_1$ and $n_2$,
then $E$ is base-guarded in $\instance$.

\begin{example}
  \label{ex:baseguarded}
The tree decomposition in Figure~\ref{fig:tree} is not base-guarded since
the set of elements in the interface between
the bag with $\set{B(c,c),R(b,e)}$
and the bag with $\set{B(c,f),R(e,f),R(f,c)}$ is $\set{c,e}$,
which is not base-guarded in the pictured set $\instance$ of facts.
\end{example}

A sentence $\phi$ has \emph{base-guarded-interface $k$-tree-like witnesses} if
for any finite
set of facts $\instance_0$,
  if there is some
  $\instance \supseteq \instance_0$ satisfying $\phi$
 then there  is
 such an $\instance$
  with an $\instance_0$-rooted $(k-1)$-width
   base-guarded-interface tree decomposition.

By adapting the proof of tree-like witnesses for $\gnf$ in \cite{gnficalp}
we can show
that $\phi$ in $\agnf$ have base-guarded-interface $k$-tree-like witnesses
for $k \leq \mysize{\phi}$.
This is stated as Proposition~\ref{prop:guarded-interface-dec-appendix} in Appendix~\ref{app:guarded-interface}.
By generalizing this slightly, we can also show the following (see
Appendix~\ref{app:guarded-interface} again for details of the proof):

\newcommand{\guardedinterfacedec}{
  The sentence $\Sigma' \wedge \neg Q$ has base-guarded-interface $k$-tree-like
  witnesses when taking
   $k \colonequals \max(\mysize{\Sigma \wedge \neg Q}, \arity{\sigma \cup \{G\}})$.
 }
\begin{lemma}
  \label{lemma:guarded-interface-dec}
  \guardedinterfacedec
\end{lemma}

Using this lemma, we can assume that we have some counterexample $\instance' \supseteq \instance_0'$
that satisfies $\Sigma' \wedge \neg Q$ and
has a $(k-1)$-width base-guarded-interface tree decomposition.
If every $\drel$ in $\sigmad$ is a strict linear order in $\instance'$,
then restricting $\instance'$ to the set of $\sigma$-facts
yields some $\instance$ that would satisfy $\Sigma \wedge \neg Q$, i.e., a
counterexample allowing
us to conclude that $\owqalin(\instance_0,\Sigma,Q)$ does not hold.
The problem is that there may be distinguished relations
$\drel$ that are not strict linear orders in $\instance'$.
In this case, we can show that $\instance'$ can 
be extended to some $\instance''$ that still satisfies $\Sigma' \wedge \neg Q$
but where all $\drel$ in $\sigmad$ are strict linear orders,
which allows us to conclude the proof.
Thus, the crucial part of the argument is about extending
$k$-guardedly linear counterexamples to genuine linear orders:

\begin{lemma} \label{lemma:reducelin}
If there is $\instance' \supseteq \instance_0'$
that has an $\instance_0'$-rooted base-guarded-interface $(k-1)$-width tree decomposition
and satisfies $\Sigma' \wedge \neg Q$ (and hence is $k$-guardedly linear),
then there is $\instance'' \supseteq \instance'$
that satisfies $\Sigma' \wedge \neg Q$
where each distinguished relation is a strict linear order.
\end{lemma}

\newcommand{\instanceg}{\instance'}
\subsection{Extending Approximate Linear Orders to Genuine Linear Orders (Proof of Lemma~\ref{lemma:reducelin})}

The proof of Lemma \ref{lemma:reducelin} 
is the main technical result in this section.
It require a few auxiliary lemmas
(Lemmas~\ref{lem:guarded-transitivity}, \ref{lem:nobadcycle}, and \ref{lemma:acov} below),
that describe the power of the $k$-guardedly linear axioms in a
$(k-1)$-width base-guarded-interface tree decomposition.

Before stating and proving these auxiliary lemmas,
we sketch the proof of Lemma~\ref{lemma:reducelin} to give an idea of how these lemmas will be used
(see page~\pageref{proof:reducelin} for more details about the proof).
First, sets of facts that have $(k-1)$-width base-guarded-interface tree decompositions
and satisfy \mbox{$k$-guardedly} linear axioms
must already be cycle-free with respect to $\drel$ (this is the Cycles Lemma, Lemma~\ref{lem:nobadcycle}).
Hence, by taking the transitive closure of $\drel$ in $\instance$,
we get a new set of facts
where every $\drel$ is a strict \emph{partial} order.
Any strict partial order can be further extended to a strict \emph{linear} order using known techniques,
so we can obtain $\instance'' \supseteq \instance'$
where $\drel$ is a strict linear order.
This $\instance''$ may have more $\drel$-facts than $\instance'$,
but the $k$-guardedly linear axioms
ensure that these new $\drel$-facts are only about pairs of elements
that are not base-guarded (this follows from the Transitivity Lemma, Lemma~\ref{lem:guarded-transitivity}).
It is clear that $\instance''$ still satisfies the $k$-guardedly linear axioms,
but the fear is that it might not satisfy $\Sigma \wedge \neg Q$.
However,
this is where use the base-covered assumption on $\Sigma \wedge \neg Q$:
satisfiability of $\Sigma \wedge \neg Q$ in $\acgnf$
is not affected by adding new $\drel$-facts
about pairs of elements that are not base-guarded
(this is the Base-coveredness Lemma, Lemma~\ref{lemma:acov}).

\subsubsection{Transitivity Lemma}
We first prove a result about transitivity
for sets of facts with
base-guarded-interface tree decompositions.

\begin{lemma}[Transitivity Lemma]\label{lem:guarded-transitivity}
Suppose $\instanceg$ is a set of facts with an
$\instanceg_0$-rooted $(k-1)$-width base-guarded-interface tree decomposition $(T, \child ,\lambda)$.
If $\instanceg$ is $k$-guardedly transitive with respect to binary relation $\drel$,
and there is a $\drel$-path $a_1 \dots a_n$
where the pair $\{a_1, a_n\}$ is base-guarded,
then $a_1 \drel a_n \in \instanceg$.
\end{lemma}

\begin{proof}
Suppose there is an $\drel$-path $a_1 \dots a_n$
and that the pair $\{a_1, a_n\}$ is base-guarded, with $v$ a node
where $a_1, a_n$ appear together.
We can assume that $a_1 \dots a_n$ is a minimal $\drel$-path
between $a_1$ and $a_n$,
so there are no repeated intermediate elements.
Consider a minimal subtree $T'$ of $T$ containing
$v$ and containing all of the elements $a_1 \dots a_n$.
We proceed by induction on the length of the path
and on the number of nodes of $T'$ (with the lexicographic order on this pair)
to show that $a_1 \drel a_n$ is in~$\instanceg$.

If all elements $a_1 \dots a_n$ are represented at $v$,
then either (i) all elements are in the root 
or (ii) the elements are in some internal node.
For (i), by construction of $\instanceg_0$,
every pair of elements in $a_1 \dots a_n$ is guarded (by $G$).
Hence, repeated application of the axiom
$
\forall x y z ( (x \drel z \wedge z \drel y \wedge \guardedbg(x,y)) \rightarrow x \drel y )
$
(which is part of the $k$-guardedly transitive axioms)
is enough to ensure that
$a_1 \drel a_n$ holds.
For (ii), since the bag size of an internal node is at most $k$,
we must have $n \leq k$, in which case an application of the
  $k$-guardedly transitive axiom to the guarded pair $\{a_1, a_n\}$ ensures that $a_1 \drel a_n$ holds.
This covers the base case of the induction.

Otherwise,
there must be some $1 \leq i < j \leq n$ such that
$a_i$ and $a_j$ are represented at $v$,
but $a_{i'}$ is not represented at $v$ for $i < i' < j$ (in particular $a_{i+1}$
is not represented at~$v$).
We claim that $a_i$ and $a_j$ must be in an interface together.

We say $a_{i+1}$ is \emph{represented in the direction of $v'$}
if $v'$ is a child of $v$ and $a_{i+1}$ is represented in the subtree
rooted at~$v'$,
or $v'$ is the parent of $v$ and $a_{i+1}$ is represented in the tree obtained
from~$T'$
by removing the subtree rooted at $v$.
Note that by definition of a tree decomposition,
since $a_{i+1}$ is not represented at $v$,
it can only be represented in at most one direction.

Let $v_{i+1}$ be the neighbor (child or parent) of $v$ such that
$a_{i+1}$ is represented in the direction of~$v_{i+1}$.
It is straightforward to show that
$a_i$ and $a_{j}$ must both be represented in the subtree in the direction of~$v_{i+1}$
in order to witness the facts $a_i \drel a_{i+1}$ and $a_{j-1} \drel a_{j}$.
But $a_i$ and $a_j$ are both in $v$,
so they must both be in $v_{i+1}$.
Hence, $a_i$ and $a_j$ are in the interface between $v$ and $v_{i+1}$.

If this is an interface with the root node, then the pair $a_i, a_j$
is guarded by $G$ by definition of~$\instanceg_0$.
Otherwise, it is base-guarded by definition of base-guarded-interface tree
  decompositions.

Hence, we can apply the inductive hypothesis to the path $a_i \dots a_j$
and the subtree $T''$ of $T'$ in the direction of $v_{i+1}$
to conclude that $a_i \drel a_j$ holds:
the reason why can apply the inductive hypothesis
is because $T''$ is smaller than $T'$ as we removed $v$, and $a_i \ldots a_j$ is no
longer than $a_1 \ldots a_n$.
If $i = 1$ and $j = n$, then we are done.
If not,
then we can apply the inductive hypothesis to the new, strictly shorter path
$a_1 \dots a_i a_j \dots a_n$ in $T'$
and conclude that $a_1 \drel a_n$ is in $\instanceg$ as desired.
\end{proof}

\subsubsection{Cycles Lemma}
We next show that within base-guarded-interface tree decompositions,
$k$-guarded transitivity and irreflexivity imply
cycle-freeness.

\begin{lemma}[Cycles Lemma] \label{lem:nobadcycle}
Suppose $\instanceg$ is a set of facts with an
$\instanceg_0$-rooted $(k-1)$-width base-guarded-interface tree decomposition $(T, \child ,\lambda)$.
If $\instanceg$ is $k$-guardedly transitive and irreflexive with respect to $\drel$,
then $\drel$ in $\instanceg$ cannot have a cycle.
\end{lemma}

\begin{proof}
Suppose for the sake of contradiction that there is a cycle
$a_1 \dots a_n a_1$ in $\instanceg$ using relation $\drel$.
Take a minimal length cycle.

If elements $a_1 \dots a_n$ are all represented in a single node in $T$,
then either (i) all elements are in the root 
or (ii) the elements are in some internal node.
For (i), by construction of $\instanceg_0$,
every pair of elements in $a_1 \dots a_n$ is guarded (by $G$).
Hence, repeated application of the axiom
$
\forall x y z ( (x \drel z \wedge z \drel y \wedge \guardedbg(x,y)) \rightarrow x \drel y )
$
(which is part of the $k$-guardedly transitive axioms)
would force $a_1 < a_1$ to be in $\instanceg$,
which would contradict irreflexivity.
  Likewise, for (ii), since the bag size of an internal node is at most $k$,
we must have $n \leq k$, so we can apply the $k$-guardedly transitive axioms
to deduce $a_1 \drel a_1$, which contradicts irreflexivity.

Even if this is not the case, then since $a_n \drel a_1$ holds,
there must be some node $v$ in which both $a_1$ and $a_n$ are represented.
Since not all elements are represented at $v$, however,
there is $1 \leq i < j \leq n$ such that
$a_i$ and $a_j$ are represented at $v$,
but $a_{i'}$ is not represented at $v$ for $i < i' < j$.
We claim that $a_i$ and $a_j$ must be in an interface together.
Observe that $a_{i+1}$ is not represented at $v$.
Let $v_{i+1}$ be the neighbor of $v$ such that
$a_{i+1}$ is represented in the subtree in the direction of $v_{i+1}$.
It is straightforward to show that
$a_i$ and $a_{j}$ must both be represented in the subtree of $T'$ in the direction of $v_{i+1}$
in order to witness the facts $a_i \drel a_{i+1}$ and $a_{j-1} \drel a_{j}$.
But $a_i$ and $a_j$ are both in~$v$,
so they must both be in $v_{i+1}$.
Hence, $a_i$ and $a_j$ are in the interface between $v$ and $v_{i+1}$.
If this is an interface with the root node, then the pair $a_i, a_j$
is base-guarded (by definition of $\instanceg_0$);
otherwise, the definition of base-guarded-interface tree decomposition
ensures that they are base-guarded.
  By the Transitivity Lemma (Lemma~\ref{lem:guarded-transitivity}),
  this means that $a_i \drel a_j$ holds.
Hence, there is a strictly shorter cycle $a_1 \dots a_i a_j \dots a_n a_1$,
contradicting the minimality of the original cycle.
\end{proof}

\subsubsection{Base-Coveredness Lemma}
Last, we note that adding only facts about unguarded sets of elements cannot impact $\acgnf$ constraints.
This is where we use the base-coveredness assumption.

\begin{lemma}[Base-coveredness Lemma] \label{lemma:acov}
Let $\instanceg' \supseteq \instanceg$
where $\instanceg'$ contains additional facts about distinguished relations, 
  but no new facts about base-guarded tuples of elements, and where we have
  $\dom{\instanceg'} = \dom{\instanceg}$.
Let $\varphi(\vec{x}) \in \acgnf$.
If $\instanceg, \vec{a}$ satisfies $\varphi(\vec{x})$
then $\instanceg', \vec{a}$ satisfies $\varphi(\vec{x})$.
\end{lemma}

\begin{proof}[Proof sketch]
We assume without loss of generality that
$\varphi$ is in $\nf$ $\acgnf$.
Let $\acgnfplus$ (respectively, $\acgnfminus$)
denote the $\nf$ $\agnf$ formulas
where the covering requirements
(distinguished atoms in CQ-shaped subformulas are appropriately base-guarded)
are required for
positively occurring (respectively, negatively occurring)
CQ-shaped formulas.
Using induction on the negation depth of $\varphi$,
we can show that:
\begin{quote}
For $\varphi(\vec{x}) \in \acgnfminus$:
$\instanceg, \vec{a}$ satisfies $\varphi(\vec{x})$ implies $\instanceg', \vec{a}$ satisfies $\varphi(\vec{x})$. \\
For $\varphi(\vec{x}) \in \acgnfplus$:
$\instanceg', \vec{a}$ satisfies  $\varphi(\vec{x})$ implies $\instanceg, \vec{a}$ satisfies  $\varphi(\vec{x})$.
\end{quote}
We omit this straightforward proof.
The desired result immediately follows, since $\acgnf = \acgnfminus$.
\end{proof}

\subsubsection{Final Proof of Lemma~\ref{lemma:reducelin}}\label{proof:reducelin}
\newcommand{\instancegtrans}{\mathcal{G}}
\newcommand{\instancegext}{\instanceg'}
We are now ready to prove Lemma~\ref{lemma:reducelin}.
We start with some $\instanceg \subseteq \instanceg_0$ satisfying $\Sigma' \wedge \neg Q$
with an
$\instanceg_0$-rooted $(k-1)$-width base-guarded-interface tree decomposition.
We prove that there is an extension $\instancegext$ of $\instanceg$ satisfying $\Sigma' \wedge \neg Q$
in which each distinguished relation is a strict linear order.
Note that because $\instanceg$ satisfies $\Sigma'$,
we know that $\instanceg$ is $k$-guardedly linear.

We present the argument when
there is one $\drel$ in $\sigmad$ that is not a strict linear order in $\instanceg$,
but the argument is similar if there are multiple distinguished relations like
this, as we can handle each distinguished relation independently
with the method
that we will present.
Let $\instancegtrans$ be the extension of $\instanceg$ obtained by taking $\drel$ in $\instancegtrans$
to be the
transitive closure of $\drel$ in $\instanceg$.
Suppose for the sake of contradiction that there is a $\drel$-cycle in $\instancegtrans$.
We proceed by induction on the number of facts from \mbox{$\instancegtrans
\setminus \instanceg$} used in this cycle.
If there are no facts from $\instancegtrans \setminus \instanceg$ in the cycle,
the Cycles Lemma (Lemma~\ref{lem:nobadcycle}) yields the contradiction.
Otherwise,
suppose that there is a cycle involving $(a_1,a_n)$,
where $(a_1,a_n)$ is a $\drel$-fact in $\instancegtrans \setminus \instanceg$ coming from
facts $(a_1,a_2), \dots, (a_{n-1},a_n)$ in $\instanceg$.
By replacing $(a_1,a_n)$ in this cycle with $(a_1,a_2), \dots, (a_{n-1},a_n)$,
we get a (longer) cycle with fewer facts from~\mbox{$\instancegtrans \setminus
\instanceg$},
which is a contradiction by the inductive hypothesis.

Since $\drel$ is transitive in $\instancegtrans$ and cycle-free,
the relation $\drel$ in $\instancegtrans$ must be a strict partial order.
We now apply the \emph{order extension principle} or
\emph{Szpilrajn extension theorem} \cite{Szpilrajn30}:
any strict partial order can be extended to a strict total order.
From this, we deduce that $\instancegtrans$ can be further extended by additional $\drel$-facts to obtain some $\instancegext$
where $\drel$ is a strict total order.

We must prove that $\instancegext \supseteq \instancegtrans \supseteq \instanceg \supseteq \instanceg_0$
does not include any new $\drel$-facts
about base-guarded tuples.
Suppose for the sake of contradiction that there is a new fact $a \drel b$
in $\instancegext \setminus \instanceg$,
where $\{a,b\}$ is base-guarded in $\instanceg$.
By the guardedly total axiom, it must be the case that there was already
$b \drel a$ in $\instanceg$, and hence also in $\instancegext$.
But $a \drel b$ and $b \drel a$ in $\instancegext$ would together imply
$a \drel a$ in $\instancegext$,
contradicting the fact that $\instancegext$ is a strict linear order.

Hence, $\instanceg$ and $\instancegext$ agree on all facts about base-guarded tuples.
Since $Q$ is base-covered and $\Sigma \in \acgnf$,
$\Sigma \wedge \neg Q \in \acgnf$.
Thus, the Base-Coveredness Lemma (Lemma~\ref{lemma:acov}) guarantees
that $\Sigma \wedge \neg Q$ is still satisfied in $\instancegext$.
Since $\instancegext$ also trivially satisfies all of the $k$-guardedly linear axioms,
it satisfies $\Sigma' \wedge \neg Q$ as required.
This concludes the proof of Lemma~\ref{lemma:reducelin}.

\subsection{Data Complexity}
The result of Theorem~\ref{thm:decidelindirect} is a combined complexity
upper bound. However, as it
works by reducing to traditional $\owqa$ in $\ptime$,
data complexity upper bounds follow
from  prior work \cite{vldb12}.

\begin{corollary}
  For any $\acgnf$ constraints $\Sigma$ and base-covered UCQ $Q$, given a finite
  set of facts $\instance_0$, we can decide $\owqalin(\instance_0, \Sigma, Q)$ in
  $\conp$ data complexity. 
\end{corollary}

This is similar to the way data complexity bounds were shown for $\owqatr$
(in Theorem~\ref{thm:ptimetransdataupper}).
However, unlike for the $\owqatr$ problem, the constraint rewriting in this
section introduces disjunction, so rewriting a $\owqalin$ problem for $\acfgtgd$s does not produce a classical query answering problem for
$\fgtgd$s. Thus the rewriting does not  imply a $\ptime$ data complexity upper bound for
$\acfgtgd$; indeed, we will see in the next section
(in Proposition~\ref{prop:lindatacompl}) that it is $\conp$-hard.

\section{Hardness Results}\label{sec:hardness}

We now show complexity lower bounds. We already know that each one of our variants of
$\owqa$ are $\twoexp$-hard in combined complexity, and $\conp$-hard in data
complexity, when $\gnf$ constraints are allowed: this follows from existing bounds on $\gnf$
reasoning even without distinguished relations \cite{vldb12}.
However, in some cases, we can show the same hardness results for weaker
languages, using the distinguished relations.

In this section, we first summarize our hardness results in
Sections~\ref{sec:hardness-tc}~and~\ref{sec:hardness-lin},
and then
present the proofs in Section~\ref{sec:hardness-proofs}.

\subsection{Hardness for $\owqatc$}\label{sec:hardness-tc}

In the setting where we have distinguished relations interpreted as the
transitive closure of other relations, we can show
$\twoexp$-hardness in combined complexity, and $\conp$-hardness in data
complexity,
for the much
weaker language of $\aincd$s.
This is in contrast with Theorem \ref{thm:ptimetransdataupper}, which showed $\ptime$
data complexity for $\owqatr$ with the more expressive language of $\acfgtgd$s.

We show hardness via a reduction from $\owqa$ with
\emph{disjunctive inclusion dependencies} ($\did$s): recall their definition in
Section~\ref{sec:guardedlogics}. $\did$s
are known to be $\twoexp$-hard in combined complexity \mcite[Thm.
2]{bourhispieris} and $\conp$-hard in data complexity
\cite{calvanese2006data,bourhispieris}, even without distinguished relations. We
use transitive closure to emulate disjunction---as was already suggested in
the description logic context by \citeA{oldhorrocks}---by creating an
$R^\trans_i$-fact and limiting the length of a witness $R_i$-path using $Q'$.
The choice of the length of the witness path among the possible lengths is used to
mimic the disjunction.  We thus show:

\newcommand{\tcdisj}{
  For any finite set of facts $\instance_0$,
  $\did$s $\Sigma$,
  and UCQ $Q$ on a signature $\sigma$,
  we can compute in $\ptime$ a set of facts $\instance_0'$,
  $\aincd$s $\Sigma'$,
  and a base-covered CQ $Q'$ on a signature $\sigma'$
  (with a single distinguished relation),
  such that $\owqa(\instance_0, \Sigma,
  Q)$ iff $\owqatc(\instance_0', \Sigma', Q')$.
}

\begin{theorem}
  \label{thm:tcdisj}
  \tcdisj
\end{theorem}

With the results of \citeA{calvanese2006data} and \citeA{bourhispieris}, this immediately implies our hardness
result:

\begin{corollary}
  The $\owqatc$ problem with $\aincd$s and base-covered CQs is $\conp$-hard in data
  complexity and $\twoexp$-hard in combined complexity.
\end{corollary}

In fact, the data complexity lower bound for $\owqatc$ even holds in the absence of
constraints:

\newcommand{\lindatacompltrans}{
  There is a base-covered CQ $Q$ such that 
  $\owqatc(\instance_0, \emptyset, Q)$ is $\conp$-hard in data complexity.
}
\begin{proposition} \label{prop:lindatacompltrans}
  \lindatacompltrans
\end{proposition} 

We prove this by reducing the problem of $3$-coloring a directed graph, known to be
$\np$-hard, to the complement of $\owqatc$: we can easily do this using
dependencies with 
disjunction in the head. Hence, as in the proof of
Theorem~\ref{thm:tcdisj}, we simulate this disjunction by using a choice of  the length of
paths that realize transitive closure facts asserted in~$\instance_0$.

All of these hardness results are first proven using UCQs rather than CQs, and
then strengthened by eliminating the disjunction in the query, by adapting a prior
trick (see, e.g., \citeR{georgchristos}) to code the intermediate truth values of
disjunctions within a CQ. We state in Appendix~\ref{app:ucqtocq} the general
lemmas about this transformation, and explain why the proofs of this section
still hold when using a CQ rather than a UCQ.

\subsection{Hardness for $\owqalin$}\label{sec:hardness-lin}

Our hardness results for $\aincd$s and $\owqatc$ also apply to $\owqalin$, using
the same technique of translating from $\did$s. What changes is the technique used
to code disjunction: rather than the length of a path in the transitive closure,
we use the totality of the order relation between elements to code disjunction
in the relative ordering of elements.
We can thus show the following analogue
to Theorem~\ref{thm:tcdisj}:

\newcommand{\lindisj}{
  For any finite set of facts $\instance_0$,
  $\did$s $\Sigma$,
  and UCQ $Q$ on a signature $\sigma$,
  we can compute in $\ptime$ a set of facts $\instance_0'$,
  $\aincd$s $\Sigma'$ (not mentioning the distinguished relations),
  and base-covered CQ $Q'$ on a signature $\sigma'$ (with a single distinguished relation),
  such that $\owqa(\instance_0, \Sigma,
  Q)$ iff $\owqalin(\instance_0', \Sigma', Q')$.
}

\begin{theorem}
  \label{thm:lindisj}
  \lindisj
\end{theorem}

Hence, we can conclude our hardness result
using prior work \cite{calvanese2006data,bourhispieris}:

\begin{corollary}
  The $\owqalin$ problem with $\aincd$ and
  base-covered CQs is $\conp$-hard in data
  complexity and $\twoexp$-hard in combined complexity.
\end{corollary}

We can also use a reduction from $3$-coloring to show hardness in data complexity
even without constraints:

\newcommand{\lindatacompl}{
  There is a base-covered CQ $Q$ such that 
  $\owqalin(\instance, \emptyset, Q)$ is $\conp$-hard in data complexity.
}
\begin{proposition} \label{prop:lindatacompl}
  \lindatacompl
\end{proposition}

Again, we will prove the results with UCQs in this section, and explain in
Appendix~\ref{app:ucqtocq} how to prove these results with a CQ instead.

\subsection{Proof of Theorems~\ref{thm:tcdisj} and~\ref{thm:lindisj}}\label{sec:hardness-proofs}

We now start to prove the results of Sections~\ref{sec:hardness-tc}
and~\ref{sec:hardness-lin}. In this section, we first prove the results about the translation
from $\did$s to $\owqatc$ and $\owqalin$, namely,
Theorems~\ref{thm:tcdisj} and~\ref{thm:lindisj}. In the next section, we show
the data complexity hardness results without constraints
(Propositions~\ref{prop:lindatacompltrans}~and~\ref{prop:lindatacompl}).
We start by proving Theorem~\ref{thm:tcdisj}, and we will adapt the proof
afterwards to show Theorem~\ref{thm:lindisj}.
Recall the claim:

\fakethm{Theorem}{thm:tcdisj}{\tcdisj}
 \medskip
 
We will establish a weaker form of the result where $Q'$ is allowed to be
a UCQ: the extension where we only use a CQ is shown in
Appendix~\ref{app:ucqtocqhardness}.

\subsubsection{Defining $\sigma'$ from $\sigma$}
We create the signature $\sigma'$ (featuring both base and distinguished
relations) from the signature $\sigma$ of the $\did$s and
from the $\did$s $\Sigma$ themselves by:
\begin{itemize}
  \item creating, for each relation $R$ in~$\sigma$, 
a base relation $R'$ in~$\sigma'$ whose arity is $\arity{R}+2$;
  \item adding a fresh binary base relation $E$, and taking the transitive closure
$E^\trans$ of~$E$ as the one
distinguished relation of~$\sigma'$;
\item creating, for each $\did$ $\tau$ in $\Sigma$ written \[\forall \vec{x} ~ R(\vec{x}) \rightarrow \bigvee_{1 \leq i \leq n}
  \exists \vec{y_i} ~ R_i(\vec{x}, \vec{y_i}),\]
a base relation
$\witness_\tau$
    in~$\sigma'$ of arity $\card{\vec x} + \sum_i \card{\vec{y_i}} + 2n + 2$. For
    simplicity, we will always use the same variables when writing
$\witness_\tau$-atoms, namely, we will write them 
$\witness_\tau(\vec{x}, e, f, \vec y_1, e_1, f_1, \ldots, \vec y_n, e_n, f_n)$.
\end{itemize}
\subsubsection{Defining $\Sigma'$ from $\Sigma$ and~$\sigma$}
We then create the $\aincd$s $\Sigma'$ from the $\did$s $\Sigma$. First, 
for each relation $R$ in~$\sigma$, we create the following $\aincd$, asserting
that the two additional positions of the base relation~$R'$ must be connected by
an $E$-path.
\[
  \tau'_R: \forall \vec x \, e \, f ~ R'(\vec x, e, f) \rightarrow E^+(e, f)
  \]
The intuition is that the failure of the query will impose that this $E$-path have length at
most~$2$, so it has length either~$1$ or~$2$. Facts with 
a path of length~$1$ will be called \emph{genuine facts}, which intuitively
means that they really hold, and those with a path of length~$2$ will be called
\emph{pseudo-facts}, intuitively meaning that they will be ignored.

Then, for each $\did$
$\tau: \forall \vec{x} ~ R(\vec{x}) \rightarrow \bigvee_{1 \leq i \leq n}
\exists \vec{y_i} ~ R_i(\vec{x}, \vec{y_i})$, we create multiple $\aincd$s. First, we
create
a $\aincd$ $\tau'$ with a $\witness_\tau$-fact in the head:
\[
  \tau' : 
\forall \vec{x} \, e \, f \, R'(\vec{x}, e, f) \rightarrow 
\exists \vec{y_1} \, e_1 \, f_1 \, \ldots \, \vec{y_n} \, e_n \, f_n  ~
	\witness_\tau(\vec{x}, e, f, \vec y_1, e_1, f_1, \ldots, \vec y_n, e_n, f_n)
      \]
Then, for $1 \leq i \leq n$, we create the following $\aincd$ $\tau'_i$:
\[
  \tau'_i: \forall \vec{x} \, e \, f \, \vec{y_1} \, e_1 \, f_1 \, \ldots \, \vec{y_n} \,
  e_n \, f_n  ~
\witness_\tau(\vec{x}, e, f, \vec y_1, e_1, f_1, \ldots, \vec y_n, e_n, f_n)
  \rightarrow R'_i(\vec{x}, \vec{y_i},e_i, f_i)
\]
In other words, whenever a $\did$ $\tau$ would be applicable on a fact $R'(\vec c,
e, f)$, we will create a fact
$\witness_\tau(\vec c, e, f, \vec d_1, e_1, f_1, \ldots, \vec d_n, e_n, f_n)$,
which will cause \emph{all} head atoms
$R'_i(\vec c, \vec d_i, e_i, f_i)$
for the $\did$ to be instantiated.
However, thanks to the two additional positions, we will be free to choose which
of these facts are pseudo-facts, and which are genuine. The query will then
enforce the correct semantics for $\did$s, by prohibiting $\witness_\tau$-facts
whose match was genuine but where all instantiated heads are pseudo-facts.

\subsubsection{Defining $Q'$ from $Q$, $\sigma$, and $\Sigma$}
The UCQ $Q'$ contains the following disjuncts (existentially
closed):
\begin{itemize}
\item \emph{$Q$-generated disjuncts}: For each disjunct $\psi$ of the original
  UCQ~$Q$, we create one disjunct $\psi'$ in the UCQ $Q'$ obtained by replacing 
each atom $R(\vec{x})$ of~$\psi$ by the conjunction $R'(\vec{x}, e, f)
\wedge E(e, f)$, where $e$ and $f$ are fresh. That is, the query $Q'$ matches
    whenever we have a witness for $Q$ consisting of genuine facts.
\item \emph{$E$-path length restriction disjuncts}:
    For each relation $R$ in~$\sigma$, we create the following disjunct
    in~$Q'$:
    \[R'(\vec{x}, e, f) \wedge E(e, y_1) \wedge E(y_1, y_2) \wedge E(y_2,
    y_3).\]
    This disjunct succeeds if
    the $E$-path annotating an $R'$-fact has length~$\geq 3$. Hence, for any
    fact $R'(\vec{a}, e, f)$, the
    $E^\trans$-fact from~$e$ to~$f$ enforced by the $\did$ $\tau'_R$ in~$\Sigma$
    must make $R'(\vec{a}, e, f)$ 
    either a genuine fact or a pseudo-fact. 
\item  \emph{$\did$ satisfaction disjuncts}:
For every $\did$ $\tau: \forall \vec{x} ~ R(\vec{x}) \rightarrow
\bigvee_i \exists \vec{y_i} ~ R_i(\vec{x}, \vec{y_i})$ in
$\Sigma$, 
    we create the following disjunct in~$Q'$:
    \[
  Q_\tau: 
  \witness_\tau(\vec{x}, e, f, \vec y_1, e_1, f_1, \ldots, \vec y_n, e_n, f_n)
  \wedge E(e,f) \wedge \bigwedge_i \left(E(e_i, w_i) \wedge E(w_i, f_i)\right).
    \]
Informally, the failure of $Q_\tau$ enforces that we cannot have the body of
    $\tau$ holding as a genuine fact and  each head disjunct realized by a
    pseudo-fact.
\end{itemize}
Observe that all of these disjuncts are
trivially base-covered (since they do not use $E^\trans$).

\subsubsection{Defining $\instance_0'$ from $\instance_0$}
We now explain how to rewrite the facts of an initial fact set $\instance_0$ on
$\sigma$ to a fact set $\instance_0'$ on $\sigma'$.
Create $\instance_0'$ by replacing each fact $F = R(\vec{a})$ of $\instance_0$
by the facts $R'(\vec{a}, b_F, b'_F)$, and $E(b_F, b'_F)$, where $b_F$ and
$b'_F$ are fresh. Hence, all facts of $\instance_0$ are created as genuine facts.

\medskip

We have now defined $\sigma'$, $\Sigma'$, $Q'$, and $\instance_0'$.
We now show that the claimed equivalence holds:
$\owqa(\instance_0, \Sigma, Q)$ holds 
iff $\owqatc(\instance_0', \Sigma', Q')$ holds.

\subsubsection{Forward Direction of the Correctness Proof}
First, let $\instance \supseteq \instance_0$ satisfy $\Sigma$ and
violate $Q$. We must construct $\instance'$
that satisfies $\Sigma'$ and violates $Q'$ when interpreting $E^+$ as the
transitive closure of~$E$.

We construct $\instance'$ using the following steps:
\begin{itemize}
  \item Modify $\instance$ in the same way that we used to build $\instance_0'$
    from $\instance_0$ (i.e., expand each fact with two fresh elements with
an $E$-edge between them, to make them genuine facts), yielding $\instance_1$.
    The result of this process consists only of genuine facts, and satisfies all
    $\aincd$s of the form $\tau'_R$.

  \item 
    Expand $\instance_1$ to a superset of facts $\instance_2$ by adding
    facts that solve violations of all dependencies in~$\Sigma'$ of the
    form~$\tau'$.

    Specifically, for every $\aincd$ of the form $\tau'$, letting $\tau : 
    \forall \vec x ~ R(\vec x) \rightarrow \bigvee_{1 \leq i \leq n} \exists
  \vec{y_i} ~ S_i(\vec x, \vec{y_i})$ be the corresponding $\did$ in~$\Sigma$,
    consider a fact $F' = R(\vec c, e, f)$ of~$\instance_1$ that matches the body
    of~$\tau'$. From the way we constructed $\instance_1$, we know that it must
    contain $E(e, f)$, and that $\instance$ must contain the fact $F = R(\vec
    c)$.
    Now, as $\instance$ satisfies $\tau$, we know that there is $1 \leq i_0 \leq n$
    such that $R_{i_0}(\vec c, \vec d_{i_0})$ holds in~$\instance$ for some
    choice of $\vec d_{i_0}$.
    Hence, by construction of~$\instance_1$ again, we know that it contains
    $F_{i_0}' = R'_{i_0}(\vec c, \vec d_{i_0}, e_{i_0}, f_{i_0})$
    and $E(e_{i_0}, f_{i_0})$ for some
    $e_{i_0}$ and $f_{i_0}$. 
    For every $i \in \{1, \ldots, n\} \backslash \{i_0\}$,
    create fresh elements $\vec d_i, e_i, f_i, w_i$ in the domain of
    $\instance_2$. Now,
    add to $\instance_2$ the
    fact $F_{\mathrm{w}} = \witness_\tau(\vec c, e, f, \vec d_1, e_1, f_1, \ldots, \allowbreak \vec d_n, e_n, f_n)$: 
    in this fact, $\vec c$ is as in~$F'$, 
    $\vec d_{i_0}, e_{i_0}, f_{i_0}$
    are as in~$F'_{i_0}$, and for
    the $\vec d_i, e_i, f_i$ 
    for $i \neq i_0$ are the fresh elements that we just created.
    
    It is easy to see now that $\instance_2$ now satisfies all $\aincd$s of the form
    $\tau'$, and it still satisfies those of the form $\tau'_R$. Further, it is
    easy to see that for any $\witness$-fact $F_{\mathrm{w}}$ of $\instance_2$ that violates a
    dependency of the form $\tau'_i$ in $\Sigma'$, the value $i$ must be
    different from the value $i_0$ used when creating~$F_{\mathrm{w}}$ (as for $i = i_0$ the
    fact $F_{i_0}'$ considered when creating $F_{\mathrm{w}}$ witnesses that
    $F_{\mathrm{w}}$ is not a violation of~$\tau_{i_0}$). Hence, we have the following
    property: for any violation of a dependency of~$\Sigma'$ in~$\instance_2$,
    the elements to be exported are in $\dom{\instance_2} \backslash
    \dom{\instance_1}$, and they only occur in one fact and in one position
    of~$\instance_2$.

  \item We now create $\instance_3$ from $\instance_2$ by taking care of the remaining
    violations by performing the \emph{chase} \cite{ahv} with $\Sigma'$ wherever applicable,
    always creating fresh elements (see Appendix~\ref{apx:chase} for details
    about the chase).
    Whenever
    we need to create a witness for some $E^{\trans}$ requirement, we always create an $E$-path of length~$2$ with a fresh element
    in the middle, that is, all facts created in $\instance_3 \backslash
    \instance_2$ are $\witness_\tau$-facts and pseudo-facts.
\end{itemize}
Let $\instance' \defeq \instance_3$.
We now check that $\instance'$ is a counterexample to $\owqatc(\instance_0',
\Sigma', Q')$.
As $\instance \supseteq \instance_0$, it is clear that
$\instance_1 \supseteq \instance_0'$, so that $\instance' \supseteq
\instance_0'$. Further, 
it is immediate by definition of the
chase that $\instance'$ satisfies $\Sigma'$.
There remains to check that $\instance'$
violates $Q'$.
To this end, we will first observe that, by construction of
$\instance'$, the only $E$-facts that we create are paths of length 1 
on fresh elements in the construction of $\instance_1$ from $\instance$, and
paths of length~2 on fresh elements in the chase in~$\instance_3$ (these
elements were either created as nulls in the chase in~$\instance_3$, or they
were created in $\instance_2$ where they only occurred in one fact and at one
position). In
particular, observe that, whenever we create an $E$-fact at any point, its endpoints are
fresh (they have just been created), so that $E$-paths have length~$1$ or~$2$ and
are on pairwise
disjoint sets of elements. Hence, as $\instance'$ satisfies the $\tau'_R$, any
fact $R'(\vec c, e, f)$ in~$\instance'$ is either a genuine fact (i.e., $E(e, f)$
holds in $\instance'$) or a pseudo-fact (i.e., there is an $E$-path of length~2
from~$e$ to~$f$ in~$\instance'$), and \emph{these two properties are mutually
exclusive}.

We now check that $Q'$ is violated, by considering each possible kind of
disjuncts.
For the \emph{$E$-path length restriction disjuncts}, we just explained that the
interpretation of~$E$ in~$\instance'$
consists of disjoint paths of length~$1$ or~$2$, 
so there is no $E$-path of length~$3$ at all in~$\instance'$.

For the \emph{$\did$ satisfaction disjuncts}, we will first observe that there
are two kinds of $\witness_\tau$-facts in~$\instance'$. Some
$\witness_\tau$-facts 
$\witness_\tau(\vec c, e, f, \vec d_1, e_1, f_1, \ldots,
\allowbreak \vec d_n, e_n, f_n)$ were created in $\instance_2$, and for these we
always have $E(e, f)$ in $\instance_1$ (hence in $\instance'$), and the same is
true also of $E(e_{i_0}, f_{i_0})$ for the $1 \leq i_0 \leq n$ considered when
creating them (using the fact that $\instance$ satisfied $\Sigma$).
All other $\witness_\tau$-facts of $\instance'$ are created in
$\instance_3$ and include only elements from $\dom{\instance_3} \backslash
\dom{\instance_2}$ or elements occurring only in one position at one fact
in~$\instance_2$ (and not occurring in~$\instance_1$): hence, for these
$\witness_\tau$-facts, neither $E(e, f)$ holds in
$\instance'$ nor does $E(e_i, f_i)$ hold for any $1 \leq i \leq n$.
This suffices to ensure that disjuncts of the form $Q_\tau$ in~$Q'$ cannot have
a match in~$\instance'$, because their $\witness_\tau$-atom can neither match
$\witness_\tau$-facts created in~$\instance_3$ (as $E(e, f)$ does not hold for
them, unlike what $Q_\tau$ requires, remembering that paths of length~1 and~2
are mutually exclusive)
nor
$\witness_\tau$-facts created in~$\instance_2$ (because, for $i = i_0$, the fact
$E(e_{i_0}, f_{i_0})$ holds for them, violating again what $Q_\tau$ requires).
Hence, the $\did$ satisfaction
disjuncts have no match in~$\instance'$.

Finally, for the \emph{$Q$-generated disjuncts},
observe that any match of them must be on genuine facts, i.e., on
facts of~$\instance'$ created for facts of~$\instance$, 
so we can conclude because $\instance$ violates
$Q$.

Hence, $\instance'$ satisfies $\Sigma'$ and violates $Q'$, which concludes the
forward direction.

\subsubsection{Backward Direction of the Correctness Proof}
In the other direction,
let $\instance' \supseteq \instance_0'$ be a counterexample
to $\owqatc(\instance_0', \Sigma', Q')$.
Consider the set of~$R'$-facts from $\instance'$
such that
$R' \in \sigma'$ corresponds to some $R \in \sigma$
and the elements in the last two positions of this $R'$-fact are connected by an
$E$-fact, i.e., the genuine facts.
Construct a set of facts $\instance$ on $\sigma$ by
projecting away the last two positions from these $R'$-facts,
and discarding all of the other facts.

It is clear by
construction of~$\instance_0'$ that $\instance \supseteq \instance_0$.
Further, as
$\instance'$ violates $Q'$, it is clear that $\instance$ violates $Q$, because any
match of a disjunct of~$Q$ on~$\instance$ implies a match of the corresponding
$Q$-generated disjunct $Q'$ in~$\instance'$. So it suffices to show
that $\instance$ satisfies $\Sigma$.

Hence, assume by contradiction that there is a $\did$
$\tau: \forall \vec x ~ R(\vec x) \rightarrow \bigvee_{1 \leq i \leq n} \exists
  \vec{y_i} ~ S_i(\vec x,
  \vec{y_i})$ of~$\Sigma$ and a fact $R(\vec c)$ of~$\instance$ which violates
  it.
Let $F' = R'(\vec{c},e,f)$
be the fact in~$\instance'$ from which we created $F$;
we know that $E(e, f)$ holds in~$\instance'$.
Since $\instance'$ satisfies $\tau'$ in $\Sigma'$,
we know that there are $\vec{d}_1, e_1, f_1, \dots, \vec{d}_n, e_n, f_n$
such that
$\witness_\tau(\vec{c}, e, f, \vec d_1, e_1, f_1, \ldots, \vec d_n, e_n, f_n)$
holds. Further, as $\instance'$ satisfies the $\tau'_i$ for $1 \leq i \leq n$,
we know that $S_i(\vec d_i, e_i, f_i)$ hold in~$\instance'$ for all $1 \leq i
\leq n$, and as $\instance'$ satisfies the $\tau'_{S_i}$, we know that
$E^\trans(e_i, f_i)$ holds, so that there is at least one $E$-path connecting $e_i$ and $f_i$.
As the $E$-path length-restriction disjuncts are violated in~$\instance'$,
these $E$-paths all have length in $\{1, 2\}$, and as the $\did$ satisfaction disjunct
$Q_\tau$ is violated in~$\instance'$, there is $1 \leq i_0 \leq n$ such that no
path from~$e_{i_0}$ to~$f_{i_0}$ in~$\instance$ has length~$2$, so
that some path must have length~$1$. Hence,
$\instance$ contains $S_{i_0}(\vec d_{i_0}, e_{i_0}, f_{i_0})$ and $E(e_{i_0},
f_{i_0})$, so $\instance$ contains $S_{i_0}(\vec d_{i_0})$, which witnesses that
$\tau$ is satisfied on $R(\vec c)$ in~$\instance$, a contradiction.
Hence, $\instance$ satisfies $\Sigma$, which concludes the proof of Theorem~\ref{thm:tcdisj}.

\bigskip

We now prove Theorem~\ref{thm:lindisj}, which states:

\fakethm{Theorem}{thm:lindisj}{\lindisj}
\medskip

The entire proof is shown by adapting the proof of Theorem~\ref{thm:tcdisj}.
Again, we show the claim with a base-covered UCQ, and we show the result for a
CQ in Appendix~\ref{app:ucqtocqhardness}

Intuitively, instead of using $E^\trans$ to emulate a disjunction on the length of
the path to encode genuine facts and pseudo facts, we will use the order
relation to emulate disjunction on the same elements: 
$e < f$ will indicate a genuine fact,
whereas $f < e$ will indicate a pseudo-fact, and $e = f$ will be
prohibited by the query.

\subsubsection{Defining $\sigma'$ from $\sigma$}
The signature $\sigma'$ is defined as in the proof of Theorem~\ref{thm:tcdisj}
except that we do not add the relations $E$
and $E^\trans$, but add a relation ${<}$ as a distinguished relation instead.

\subsubsection{Defining $\Sigma'$ from $\Sigma$ and $\sigma$}
We also define $\Sigma'$ as before except that we do not create the $\aincd$s of
the form $\forall \vec x \, e \, f ~ R'(\vec x, e, f) \rightarrow E^\trans(e,
f)$. Constraints like this are not necessary because the totality of $<$ already enforces the corresponding property.
This means that $\Sigma'$ does not mention the distinguished relations.

\subsubsection{Defining $Q'$ from $Q$, $\sigma$, and $\Sigma$}
The UCQ $Q'$ contains the following disjuncts (existentially closed), which are
clearly base-covered:
\begin{itemize}
\item \emph{$Q$-generated disjuncts}:
For each disjunct $\psi$ of the original UCQ $Q$, we create one disjunct $\psi'$
    in the UCQ $Q'$ where
each atom $R(\vec{x})$ is replaced by the conjunction $R'(\vec{x}, e, f)
\wedge e < f$, where $e$ and $f$ are fresh.
That is, the query $Q'$ matches whenever we have a witness for $Q$
consisting of genuine facts.

\item \emph{Order restriction disjuncts}: For each relation $R$ in~$\sigma$, we
  create a disjunct $R'(\vec{x},e,e)$. Intuitively, failure of this disjunct
    imposes that, for each relation
    $R' \in \sigma'$ that stands for a relation $R \in \sigma$, the elements
    in the two last positions must be different; so every fact must be either a genuine fact or a pseudo-fact.

\item \emph{$\did$ satisfaction disjuncts}:
For every $\did$ $\tau: \forall \vec{x} ~ R(\vec{x}) \rightarrow
\bigvee_i \exists \vec{y_i} R_i(\vec{x}, \vec{y_i})$ in
$\Sigma$, 
    we create the following disjunct in $Q'$:
    \[
Q_\tau: \witness_\tau(\vec{x}, e, f, \vec y_1, e_1, f_1, \ldots, \vec y_n, e_n,
    f_n)  \wedge
e < f \wedge \bigwedge_{1 \leq i \leq n} f_i < e_i.
  \]
Intuitively, $Q_\tau$ is satisfied if the body of  $\tau$ is matched to a genuine
    fact but each
of the head disjuncts of~$\tau$ is matched to 
    a pseudo-fact. 
\end{itemize}
\subsubsection{Defining $\instance'_0$ from $\instance_0$}
The process to define $\instance_0'$ from $\instance_0$ is
like in the proof of Theorem~\ref{thm:tcdisj}
except that, instead of creating the facts $E(b_F, b'_F)$, we create facts $b_F < b'_F$.

\medskip

The proof that $\owqa(\instance_0, \Sigma,
Q)$ holds iff $\owqalin(\instance_0', \Sigma', Q')$ holds
is similar to the proof for Theorem~\ref{thm:tcdisj},
so we sketch the proof and highlight the main differences.

\subsubsection{Forward Direction of the Correctness Proof}
Let $\instance \supseteq \instance_0$ satisfy $\Sigma$ and
violate $Q$, and construct $\instance'$ that satisfies $\Sigma'$ and violates
$Q'$ and in which $<$ is an order relation. We do so as follows:

\begin{itemize}
  \item Build $\instance'$ from $\instance$
    by expanding each fact $F$ with two fresh elements $b_F$ and $b'_F$
    and adding the fact $b_F < b'_F$ to make it a genuine fact.
  \item Create $\instance_2$ and $\instance_3$ as before, except that
    pseudo-facts and genuine facts are annotated with $<$-facts rather than
    $E^\trans$-facts.
  \item Add one step where we construct $\instance'$ from $\instance_3$ by
    completing $<$ to
    be a total order. To do so, observe that $<$ in $\instance_3$ must be a
    partial order, because all order facts that we have created are on disjoint
    elements (they are of the form $b_F < b_F'$ or $b_F' < b_F$ where $b_F$ and
    $b_F'$ are elements specific to a fact $F$). Hence, we define $\instance'$
    by simply completing $<$ to a total order using the order
    extension principle \cite{Szpilrajn30}. Note that this can never change a
    genuine fact in a pseudo-fact or vice-versa.
\end{itemize}
As before it is clear that $\instance' \supseteq \instance'_0$ and that
$\instance'$ satisfies $\Sigma'$ (note that the additional order facts created
from $\instance_3$ to $\instance'$ cannot create a violation of $\Sigma'$, as it
does not mention $<$), and we have made sure that $<$ is a total
order. To see why $Q'$ is not satisfied in~$\instance'$, we proceed exactly as
before for the $\did$ satisfaction disjuncts and $Q$-generated disjuncts,
but replacing ``having an $E$-fact between $e$ and $f$'' by ``having $e
< f$'', and replacing ``having an $E$-path of length~$2$ between $e$ and $f$'' by
``having $e > f$'', and likewise for $e_i$ and $f_i$. For the order-restriction
disjuncts, we simply observe that for any $R'$-fact $R'(\vec a,
e, f)$ in $\instance'$, by construction we always have $e \neq
f$.

\subsubsection{Backward Direction of the Correctness Proof}
Suppose we have some counterexample $\instance'$
to $\owqatc(\instance_0',\Sigma',Q')$.
We construct $\instance$ from $\instance'$ by
keeping all facts whose last two elements $e$ and $f$ are such that $e < f$. The
result still clearly satisfies $\instance \supseteq \instance_0$, and the proof
of why it violates
$Q$ is unchanged. To show that $\instance$ satisfies $\Sigma$, we adapt the
argument of the proof of Theorem~\ref{thm:tcdisj}, but instead of the
$\tau'_{S_i}$ we rely on totality of the order to deduce that either $e_i <
f_i$, $e_i = f_i$, or $f_i < e_i$ for all~$i$, and we rely on the
order-restriction disjuncts (rather than the $E$-path length-restriction
disjuncts) to deduce that either $e_i < f_i$ or $f_i < e_i$. We conclude as
before by the $\did$ satisfaction disjuncts that we must have $e_i < f_i$ for
some~$i$. Thus, we deduce from the satisfaction of $\Sigma'$ by $\instance'$
that $\instance$ satisfies $\Sigma$, which concludes the backward direction of
the correctness proof, and finishes the proof of Theorem~\ref{thm:lindisj}.

\subsection{Proof of Propositions~\ref{prop:lindatacompltrans}~and~\ref{prop:lindatacompl}}

We now give data complexity lower bounds that show
$\conp$-hardness even in the absence of constraints.
We first prove Proposition~\ref{prop:lindatacompltrans}:

\fakethm{Proposition}{prop:lindatacompltrans}{\lindatacompltrans}

\begin{proof}
  We will show the result for a UCQ $Q$, and we extend it to a CQ in
  Appendix~\ref{app:ucqtocqhardness}. We show $\conp$-hardness by reducing the
  $3$-colorability problem in $\ptime$ to the negation of the $\owqatc$ problem: this well-known
  $\np$-hard problem asks, given an undirected graph $\calG$, whether it is
  $3$-colorable, i.e., whether there is a mapping from the vertices of $\calG$ to
  a set of 3 colors (without loss of generality the set $\{1, 2, 3\}$) such that
  no two adjacent vertices are assigned the same color. Observe that we can
  modify slightly the definition of this problem to allow vertices to carry
  multiple colors, i.e., be colored by \emph{non-empty} subsets of $\{1, 2,
  3\}$): the use of multiple colors on a vertex imposes more constraints on the
  vertex, so makes our life harder. In other words, we can restrict the search
  for solutions to colorings where each vertex has one single color, but when
  encoding the $3$-colorability problem to $\owqatc$ we do not need to impose that
  vertices carry \emph{exactly} one color (we must just impose that they carry
  \emph{at least} one color).
  
\subsubsection{Definition of the Reduction}
  We define the signature $\sigma$ as containing:
  \begin{itemize}
    \item One binary relation $G$ to code the edges of the graph 
      provided as input to the reduction;
    \item One binary relation $E$ and its transitive closure $E^\trans$
      (playing a similar role as in the proof of Theorem~\ref{thm:tcdisj});
    \item For each $\chi \in \{1, 2, 3\}$, a ternary relation $C_\chi$.
      Intuitively, the first position of $C_\chi$-facts will contain the element
      that codes a vertex (and occurs in the $G$-facts that describe the edges
      incident to that vertex),
      and the positions 2
      and~3 will contain elements playing a similar role to elements $e$ and
      $f$ in $R'$-facts in the proof of Theorem~\ref{thm:tcdisj}. Namely, for a
      fact $C_\chi(a, e, f)$, if $e$ and $f$ are connected by a path of length
      1, this will indicate that vertex~$a$ has color $\chi$, while if they are
      connected by a path of length $2$ this will indicate that $a$ does not
      have color $\chi$.
  \end{itemize}
  We then define the UCQ $Q$ to contain the following disjuncts (existentially
  closed):
  \begin{itemize}
  \item \emph{$E$-path length restriction disjuncts}:
    For each $\chi \in \{1, 2,
      3\}$, a disjunct that holds if the $E$-path for $C_\chi$-facts has length
      $\geq 3$: \[C_\chi(x, e, f) \wedge E(e, y_1) \wedge E(y_1, y_2) \wedge E(y_2,
      y_3)\]
  \item \emph{Adjacency disjuncts:} For $\chi \in \{1, 2, 3\}$, a disjunct $Q_i$ that holds if two
      adjacent vertices were assigned the same color:
        \[C_\chi(x, e, f) \wedge E(e, f) \wedge G(x, x') 
        \wedge C_\chi(x', e', f') \wedge E(e', f')\]
    \item \emph{Coloring disjunct:} A disjunct that holds if a vertex was not
      assigned any color: \[\bigwedge_{\chi \in \{1, 2, 3\}}
      C_\chi(x, e_\chi, f_\chi) \wedge E(e_\chi,
      w_\chi) \wedge E(w_\chi, f_\chi)\]
  \end{itemize}
  Given an undirected graph $\calG$, we then code it in~$\ptime$ as the set of
  facts $\instance_0$
  defined by having:
  \begin{itemize}
    \item One fact $G(x, y)$ and one fact $G(y, x)$ for each edge $\{x,y\}$ in~$\calG$
    \item One fact $C_\chi(x, e_{x,\chi}, f_{x,\chi})$ and one fact
      $E^\trans(e_{x,\chi}, f_{x,\chi})$ 
      for each vertex~$x$ in~$\calG$
      and for each $\chi \in \{1, 2, 3\}$,
      where all the $e_{x,\chi}$ and $f_{x,\chi}$ are fresh.
  \end{itemize}
\subsubsection{Correctness Proof for the Reduction}
  We now show that $\calG$ is $3$-colorable iff $\owqatc(\instance_0, \emptyset,
  Q)$ is false, completing the reduction.

  For the forward direction, consider a $3$-coloring of $\calG$. Construct
  $\instance \supseteq \instance_0$ as follows.
  For each vertex~$x$ of
  $\calG$
  (with facts $C_\chi(x, e_{x,\chi}, f_{x, \chi}) \in \instance_0$
  as defined above
  for all $\chi \in \{1, 2, 3\}$),
  create the facts
  $E(e_{x,\chi}, f_{x,\chi})$ where $\chi$ is the color assigned to $x$, and the
  facts $E(e_{x,\chi'},
  w_{x,\chi'})$ and $E(w_{x,\chi'}, f_{x,\chi'})$ for the two other colors $\chi' \in \{1, 2, 3\} \backslash
  \{i\}$ (with the $w_{x,\chi'}$ being fresh). It is clear that $\instance$ thus
  defined is such that $\instance \supseteq \instance_0$,
  and that $E^\trans$ is
  the transitive closure of~$E$ in~$\instance$.
  The $E$-path length restriction disjuncts of $Q$ do not match in~$\instance$
  (note that we only create $E$-paths whose endpoints are pairwise distinct),
  and the coloring disjunct does not match either because each vertex has some
  color.
  Finally, the definition of a $3$-coloring ensures that the adjacency disjuncts
  do not match either.
  Hence, $\instance$ is a set of facts violating $Q$.

  \medskip

  For the backward direction, consider some $\instance \supseteq \instance_0$
  where $E^\trans$ is the transitive closure of $E$ that
  violates $Q$.
  For any vertex~$x$ of~$\calG$ and for all $\chi \in \{1, 2, 3\}$, letting 
  $C_\chi(x, e_{x,\chi}, f_{x,\chi})$ be the facts as defined in the construction,
  as the fact $E^\trans(e_{x,\chi}, f_{x,\chi})$ holds in $\instance_0$ for each $\chi \in \{1,
  2, 3\}$ and $E^\trans$ must be the transitive closure of~$E$ in~$\instance'$,
  there must be an $E$-path from 
  $e_{x,\chi}$ to $f_{x,\chi}$. Further, 
  as $\instance$ violates the $E$-path length restriction
  disjuncts of $Q$, this path must be of length~$1$ or~$2$. 
  Further, as
  $\instance$ violates the coloring disjunct of~$Q$, for every vertex~$x$
  of~$\calG$, there must be one $\chi \in \{1, 2, 3\}$ such that the paths for
  $x$ and $\chi$ has length~1. We define a coloring of $\calG$ by choosing for
  each vertex~$x$ a
  color $\chi$ for which the $E$-path has length~$1$, i.e., the fact $E(e_{x,\chi}, f_{x,\chi})$ holds. (As
  pointed out in the beginning of the proof, for each $x$, there could be
  multiple such $\chi$, but we can take any of them.)
  This indeed defines a $3$-coloring, as any
  violation of the $3$-coloring witnessed by two adjacent vertices of color~$\chi$
  would imply a match of the $\chi$-th adjacency disjunct of~$Q$ in~$\instance$.
  This concludes the backward direction of the correctness proof of the
  reduction, and concludes the proof.
\end{proof}

We then modify the proof to show Proposition~\ref{prop:lindatacompl}:

\fakethm{Proposition}{prop:lindatacompl}{\lindatacompl}

\begin{proof}
  Again we show the result for a UCQ $Q$, and extend it to a CQ in
  Appendix~\ref{app:ucqtocqhardness}.
  We define $\sigma$ as in the proof of Proposition~\ref{prop:lindatacompltrans}
  but with an order relation $<$ instead of the two relations~$E$ and~$E^\trans$.
  We define $Q$ as in the proof
  of Proposition~\ref{prop:lindatacompltrans} but
  without the $E$-path length restriction disjunct, and replacing in the other
  disjuncts the length-1 paths $E(e, f)$ and $E(e', f')$ by $e < f$ and $e' <
  f'$ 
  and the paths $E(e_\chi, w_\chi) \wedge E(w_\chi, f_\chi)$ by $f_\chi <
  e_\chi$: the resulting UCQ is
  clearly base-covered. Note that, unlike
  in the proof of Theorem~\ref{thm:lindisj}, we need not worry about equalities
  (so we need not add order restriction disjuncts),
  as all the relevant elements are already created 
  as distinct elements in~$\instance_0$.
  We define $\instance_0$ in the same fashion as in the proof of
  Proposition~\ref{prop:lindatacompltrans} but without the
  $E^\trans$-facts.

  We show the same equivalence as in that proof, but for $\owqalin$. We do it by
  replacing $E$-paths of length $1$ from an $e$ to an $f$ by $e
  < f$, and $E$-paths of length $2$ by $f < e$. In the forward direction, we
  build a counterexample set of facts from a coloring as before, and extend~$<$ to be an
  arbitrary total order
  (this is possible because all $<$-facts that we create are on disjoint pairs).
  In the backward direction, we use totality of~$<$
  to argue that a counterexample set of facts must choose some order 
  between the $e_{x,\chi}$ and the $f_{x,\chi}$, and so must decide 
  which colors are assigned to each vertex, in a way that yields a coloring
  (because the adjacency and coloring disjuncts are violated).
\end{proof}

\section{Undecidability Results}\label{sec:undecid}

We now show how slight changes to the constraint languages
and query languages used for the results
in Sections~\ref{sec:decid}~and~\ref{sec:decidlin}
lead to undecidability of query answering.

The undecidability proofs in this section are by reduction from an
\emph{infinite tiling problem},
specified by a set of colors $\mathbb{C} = C_1,
  \ldots, C_k$, a set of forbidden \emph{horizontal patterns} $\mathbb{H}
  \subseteq \mathbb{C}^2$ and a set of forbidden \emph{vertical patterns}
  $\mathbb{V} \subseteq \mathbb{C}^2$. It asks, given a sequence $c_0, \ldots,
  c_n$ of colors of $\mathbb{C}$, whether there exists a function $f :
  \mathbb{N}^2 \rightarrow \mathbb{C}$ such that we have $f((0, i)) = c_i$ for all $0
  \leq i \leq n$, and for all $i, j \in \mathbb{N}$, we have $(f(i, j), f(i+1,
  j)) \notin \mathbb{H}$ and $(f(i, j), f(i, j+1)) \notin \mathbb{V}$.
It is well-known that there are fixed 
$\mathbb{C}$, $\mathbb{V}$, $\mathbb{H}$ for which this problem is undecidable
\cite{classicaldecision}.

\subsection{Undecidability Results for $\owqatr$ and $\owqatc$} \label{sec:undecidtrans}

We have shown in Section~\ref{sec:decid}
that query answering is decidable with transitive relations (even
with transitive closure), $\afgtgd$s, and UCQs
(Theorem~\ref{thm:decidtransautomata}).
Removing the base-frontier-guarded
requirement makes $\owqatc$ undecidable, even when
constraints are inclusion dependencies:

\newcommand{\undectrans}{
  There is a signature $\sigma = \sigmab \sqcup \sigmad$ with a single
  distinguished relation $S^\trans$ in~$\sigmad$,
  a set $\Sigma$ of
  $\incd$s on~$\sigma$, and a CQ $Q$ on $\sigmab$,
  such that the following problem is undecidable:
  given a finite set of facts $\instance_0$,
  decide $\owqatc(\instance_0, \Sigma, Q)$.
}

\begin{theorem}
  \label{thm:undectrans}
  \undectrans
\end{theorem}

We can also show that the $\owqatr$ problem is undecidable if we allow
disjunctive inclusion dependencies which are not base-guarded:

\newcommand{\undectransb}{
  There is an arity-two signature $\sigma = \sigmab \sqcup \sigmad$
  with a single distinguished relation $S^\trans$ in~$\sigmad$,
  a set $\Sigma$ of $\did$s on~$\sigma$,
  a CQ $Q$ on $\sigmab$, such that the following problem is
  undecidable:
  given a finite set of facts $\instance_0$,
  decide $\owqatr(\instance_0, \Sigma, Q)$.
}
\begin{theorem}
  \label{thm:undectransb}
  \undectransb
\end{theorem}

The two results are incomparable: the second one applies to the $\owqatr$
problem rather than $\owqatc$, and does not require a higher-arity signature,
but it uses more expressive constraints that feature disjunction.
To prove both results, we reduce from a tiling problem, using a transitive
successor relation to code a grid, and using the query to test for forbidden
adjacent tile patterns. We first present the proof of the second result, because
it is simpler. We then adapt this proof to show the first result.

\begin{proof}[Proof of Theorem~\ref{thm:undectransb}]
Fix $\mathbb{C}$, $\mathbb{V}$, $\mathbb{H}$
such that the infinite tiling problem is undecidable.
We will give a reduction from this infinite tiling problem to $\owqatr(\instance_0,\Sigma,Q)$.
We prove the result with a UCQ instead of a CQ, and explain how we can use
a CQ instead in Appendix~\ref{app:ucqtocqundecid}.

  \subsubsection{Definition of the Reduction}
  The base relations of the signature are a binary relation $S'$ (for ``successor''),
  one binary relation $K_i$ for each color $C_i$, and one unary relation $K_i'$
  for each color $C_i$. We also use one distinguished transitive relation, $S^\trans$.
  The idea is that we will create an infinite chain of~$S'$ and assert that it
  is included in
  $S^\trans$: hence, $S^\trans$ will be a transitive super-relation of~$S'$, so
  it will contain at least its transitive closure.
  From $S^\trans$, we will
  define a grid structure on which we can encode the tiling problem, with
  grid positions represented as a pair of elements in an $S^\trans$-fact.

  Let $\Sigma$ consist of the following $\did$s (omitting universal quantifiers for brevity):
  \begin{align*}
    S'(x, y) & \rightarrow \exists z ~ S'(y, z) &\qquad
    S'(x, y) & \rightarrow S^\trans(x, y)\\
    S^\trans(x, y) & \rightarrow \bigvee_i K_i(x, y) &\qquad
    S^\trans(x, y) & \rightarrow \bigvee_i K_i(y, x) &\qquad
    S^\trans(x, y) & \rightarrow \bigvee_i K_i'(x) .
  \end{align*}

  The $K_i(x, y)$ describe the assignment of colors to grid positions
  represented as pairs on the infinite chain as we explained. The point of
  $K_i'(x)$ is that it stands for $K_i(x, x)$: we need a different relation symbol
  because variable reuse is not allowed in inclusion dependencies. Note that
  some of these constraints are not base-guarded.

  Let the UCQ $Q$ be a disjunction of the following sentences
  (omitting existential quantifiers for brevity):
  \begin{inparaitem}[]  
    \item for each forbidden horizontal pair $(C_i, C_j) \in \mathbb{H}$, with
      $1 \leq i, j \leq k$, the disjuncts
      \begin{align*}
        K_i(x, y) \wedge S'(y, y') \wedge K_j(x, y') \quad
        K_i'(y) \wedge S'(y, y') \wedge K_j(y, y') \quad
        K_i(y', y) \wedge S'(y, y') \wedge K_j'(y') 
      \end{align*}
    \item and for each forbidden vertical pair $(C_i, C_j) \in \mathbb{V}$, the
      analogous disjuncts
      \begin{align*}
        K_i(x, y) \wedge S'(x, x') \wedge K_j(x', y) \quad
        K_i'(x) \wedge S'(x, x') \wedge K_j(x', x) \quad
        K_i(x, x') \wedge S'(x, x') \wedge K_j'(x') .
      \end{align*}
   \end{inparaitem}
  Given an initial instance $c_0, \ldots, c_n$ of the tiling problem,
  let the initial set of facts $\instance_0$ consist of 
  the fact $K_j'(a_0)$ such that $C_j$ is the color of $c_0$, and
  for $0 \leq i < n$, the fact $S'(a_i, a_{i+1})$
  and the fact $K_j(a_0, a_i)$ such that $C_j$ is the color of initial element $c_i$.

  \subsubsection{Correctness Proof for the Reduction}
  We claim that the tiling problem has a solution iff
  there is a superset of $\instance_0$ that satisfies $\Sigma$ and
  violates $Q$ and where $S^\trans$ is transitive. From this we conclude the
  reduction and deduce the undecidability of $\owqatr$ as stated.

  For the forward direction, from a solution $f$ to the tiling problem for
  input $\vec{c}$, we construct the counterexample $\instance \supseteq \instance_0$ as follows.
  We first 
  extend the initial chain of $S'$-facts in~$\instance_0$ to
  an infinite chain $S'(a_0, a_1), \ldots, \allowbreak S'(a_m, a_{m+1}),
  \ldots$,
  and fix $S^\trans$ to be the
  transitive closure of this $S'$-chain (so it is indeed transitive).
  For all $i, j \in \mathbb{N}$ such that $i \neq j$, we create the
  fact $K_l(a_i, a_j)$ where $l = f(i,j)$. For all $i \in \mathbb{N}$, we create
  the fact $K_l'(a_i)$ where $l = f(i,i)$. This clearly satisfies the
  constraints in $\Sigma$, and does not satisfy the query
  because $f$ is a tiling.

  For the backward direction, consider an $\instance \supseteq \instance_0$ that
  satisfies $\Sigma$ and violates $Q$. Starting at the chain of $S'$-facts
  of~$\instance_0$, we can deduce, using the constraints,  the existence of an infinite chain $a_0, \ldots, a_n,
  \ldots$ of~$S'$-facts
  (whose elements may be distinct or not, this does not matter). Define a
  tiling $f$ matching the initial tiling problem instance as follows. For all
  $i < j$ in $\mathbb{N}$, as there is a path of $S'$-facts from $a_i$ to $a_j$,
  we infer that $S^\trans(a_i, a_j)$ holds, so that $K_l(a_i, a_j)$ holds for
  some $1 \leq l \leq k$; pick one such fact, taking the fact of $\instance_0$
  if $i = 0$ and $j \leq n$, 
  and fix $f(i, j) \defeq l$. For $i > j$ we can likewise see
  that $S^\trans(a_j, a_i)$ holds whence $K_l(a_i,a_j)$ holds for some~$l$, and
  we continue as before. For $i \in \mathbb{N}$, as $S'(a_i, a_{i+1})$ holds, we
  know that $K'_l(a_i)$ holds for some $1 \leq l \leq k$ (again we take the
  fact of~$\instance_0$ if $i=0$), and fix accordingly
  $f(i, i) \defeq l$. The resulting $f$ clearly satisfies the initial tiling
  problem instance $c_0, \ldots, c_n$, and it is clearly a solution to the
  tiling problem, as any forbidden pattern in~$f$ would witness a match of a
  disjunct of~$Q$ in~$\instance$. This shows that the reduction is correct, and
  concludes the proof.
\end{proof}

We now prove the first result, drawing inspiration from the previous proof, but using 
the transitive closure to emulate disjunction as we did in Theorem~\ref{thm:tcdisj}.

\begin{proof}[Proof of Theorem~\ref{thm:undectrans}]
  We reuse the notations for tiling problems from the previous proof.
  We first prove the result with two distinguished relations $S^\trans$ and
  $C^\trans$ and with a UCQ, and then explain how the proof is modified to use only a single
  transitive relation $S^\trans$.
  The extension to a CQ is explained in Appendix~\ref{app:ucqtocqundecid}.

  \subsubsection{Definition of the Reduction}
  We define a binary relation $S$ (for ``successor'') of which $S^\trans$ is
  interpreted as the transitive closure, one binary relation $S'$,
  one 3-ary relation $G$ (for ``grid''),
  one binary relation $G'$ (standing for cells on the diagonal of the grid, like
  $K_i'$ in the previous proof),
  one binary relation $T$ (a terminal for gadgets that we will define to
  indicate colors)
  and one binary relation $C$ of which
  $C^\trans$ is interpreted as the transitive closure. The distinction between
  $S$ and $S'$ is not important for now but will matter when we adapt the
  proof later to use a single distinguished relation.

  We write the following inclusion dependencies $\Sigma$ (omitting universal quantifiers for brevity):
  \begin{align*}
    S'(x, y) & \rightarrow \exists z ~ S'(y, z) &\qquad
    S'(x, y) & \rightarrow S(x, y) \\
    S^\trans(x, y) & \rightarrow \exists z ~ G(x, y, z)&\qquad
    S^\trans(x, y) & \rightarrow \exists z ~ G(y, x, z)&\qquad
    S^\trans(x, y) & \rightarrow \exists z ~ G'(x, z)\\
    G(x, y, z) & \rightarrow \exists w ~ T(z, w)&\qquad
    G'(x, z) & \rightarrow \exists w ~ T(z, w)&\qquad
    T(z, w) & \rightarrow C^\trans(z, w)
  \end{align*}
  In preparation for defining the query $Q$, we define $Q_i(z)$
  for all $i > 0$ to match the left endpoint of $T$-facts covered by a $C$-path
  of length $i$ (intuitively coding color~$i$):
  \[
    \exists z_1 \ldots z_i \, w ~ C(z, z_1) \wedge C(z_1, z_2) \wedge \ldots,
    C(z_{i-1}, z_i) \wedge
    T(z, z_i),
  \] 
The query $Q$ is a disjunction of the following disjuncts (existentially closed):
  \begin{itemize}
    \item \emph{$C$-path length restriction disjuncts:} One disjunct written as follows,
      where $k$ is the number of colors \[
    S'(x, w) \wedge G(x, y, z) \wedge\, T(z, z') \wedge C(z, z_1) 
      \wedge\, C(z_1, z_2) \wedge \cdots \wedge C(z_{k-1}, z_k), C(z_k, z_{k+1})\]
      and one disjunct defined similarly but with $G(x, y, z)$ replaced by
      $G'(x, z)$. Intuitively, these disjuncts impose that $C$-paths annotating
      $T$-facts must code colors between $1$ and $k$ (i.e., they cannot have
      length $k+1$ or greater), and the distinction
      between $G$ and $G'$ is for reasons similar to the distinction between the
      $K_i$ and $K'_i$ in the proof of Theorem~\ref{thm:undectransb}.
    \item \emph{Horizontal adjacency disjuncts:}
      For each forbidden horizontal pair $(C_i, C_j) \in \mathbb{H}$, with
      $1 \leq i, j \leq k$, the disjuncts:
        \begin{align*}
G(x, y, z) \wedge G(x, y', z')
\wedge Q_i(z) \wedge Q_j(z') \wedge S'(y, y')\\
G'(y, z) \wedge G(y, y', z') \wedge Q_i(z) \wedge Q_j(z') \wedge S'(y, y')\\
G(y', y, z) \wedge G'(y', z') \wedge Q_i(z) \wedge Q_j(z') \wedge S'(y, y')
\end{align*}
    \item \emph{Vertical adjacency disjuncts:}
      For each $(C_i, C_j) \in \mathbb{V}$, the same queries but replacing
      atoms $S'(y, y')$ by $S'(x, x')$ and the two first atoms of the last two
      subqueries by $G'(x, z) \wedge G(x, x', z')$ and $G(x, x', z) \wedge G'(x', z')$.
  \end{itemize}
  Given an initial instance of the tiling problem $c_0, \ldots, c_n$,
  the initial set of facts $\instance_0$ consists of the following:
  \begin{inparaenum}[(i)]
    \item $S'(a_i, a_{i+1})$ for $0 \leq i < n$;
    \item $G(a_0, a_i, b_{0,i})$ for $0 < i \leq n$;
    \item $G'(a_0, b_{0,0})$
    \item for all $0 \leq i \leq n$, letting $l$ be such that $c_i$ is the
      $l$-th color $C_l$, 
      we create the \emph{length-$l$ gadget on~$b_{0,i}$}:
      we create a path $C(b_{0,i}, d_{0,i}^1), 
      C(d_{0,i}^1, d_{0,i}^2),
      \ldots
      C(d_{0,i}^{l-1}, d_{0,i}^l)$,
      and the fact $T(b_{0,i}, d_{0,i}^l)$, where the elements $b_{0,i}$ and
      $d_{0,i}^j$ are all fresh;
  \end{inparaenum}

  \subsubsection{Correctness Proof for the Reduction}
  We claim that the tiling problem has a solution iff
  there is a superset of $\instance_0$ that satisfies $\Sigma$ and
  violates $Q$, where the $S^\trans$ and $C^\trans$ relations are interpreted as the
  transitive closure of $S$ and $C$, from which we conclude the reduction and
  deduce the undecidability of $\owqatc$ as stated.

  For the forward direction, from a solution $f$ to the tiling problem for input
  $\vec{c}$, we construct  $\instance \supseteq \instance_0$ as follows.
  We first create 
  an infinite chain $S'(a_0, a_1), \ldots, S'(a_m, a_{m+1}), \ldots$ to complete
  the initial chain of $S'$-facts in~$\instance_0$, we create the implied
  $S$-facts, and make $S^\trans$ the transitive closure of $S$.
  We then create one fact $G(a_i, a_j, b_{i,j})$ for all $i \neq j$
  in~$\mathbb{N}$ and one fact $G'(a_i, b_{i,i})$ for all $i \in
  \mathbb{N}$.
  Last, for all $i, j \in \mathbb{N}$, letting $l \defeq f(i, j)$, 
  we create the length-$l$ gadget on $b_{i,j}$ with fresh elements.

  It is clear that $\instance$ contains the facts of~$\instance_0$. It is easy to verify
  that it satisfies $\Sigma$. To see that we do not satisfy the query, observe
  that:
  \begin{itemize}
    \item The $C$-path length restriction disjuncts have no match 
      because all $C$-paths created have length $\leq k$ and are on disjoint
      sets of elements;
    \item For the horizontal adjacency disjuncts, it is clear that, in any
      match, $z$ must be of the form $b_{i,j}$ and $z'$ of the form $b_{i,j+1}$;
      the reason for the three different forms is that the cases where $i = j$
      and where
      $i \neq j$ are managed differently. Then, as $f$ respects
      $\mathbb{H}$, we know that the $Q_i$ and $Q_j$ subqueries cannot be
      satisfied, because for any $l \in \mathbb{N}$ and $i', j' \in \mathbb{N}$,
      we have $Q_l(b_{i',j'})$ iff $f(i', j') = l$ by construction;
    \item The reasoning for the vertical adjacency disjuncts is analogous.
  \end{itemize}
  Hence, $\instance \supseteq \instance_0$, satisfies $\Sigma$, and violates
  $Q$, which concludes the proof of the forward direction of the implication.

  For the backward direction, consider an $\instance \supseteq \instance_0$ that
  satisfies $\Sigma$ and violates $Q$. Starting at the chain of $S'$-facts
  of~$\instance_0$, we can see that there is  an infinite chain $a_0, \ldots, a_n,
  \ldots$ of $S'$-facts
  (whose elements may be distinct or not, this does not matter), and hence
  we infer the existence of the corresponding $S$-facts.  We can also infer the
  existence of elements $b_{i,j}$ for all $i, j \in \mathbb{N}$ (again, these
  elements may be distinct or not) such that $G'(a_i, b_{i,i})$ holds and $G(a_i, a_j,
  b_{i,j})$ holds if $i \neq j$. From this we conclude that there is a fact $T(b_{i,j}, c_{i,
  j})$ for all $i, j \in \mathbb{N}$, with a $C$-path
  from~$b_{i,j}$
  to~$c_{i,j}$. As the $C$-path length restriction disjuncts are violated, there cannot be such a $C$-path
  of length $> k$, so we can define a function $f$ from $\mathbb{N} \times
  \mathbb{N}$ to $\mathbb{C}$ by setting $f(i, j)$ to be $c_l$ where $l$ is the
  length of one such path, for all $i, j\in \mathbb{N}$; this can be performed in a
  way that matches~$\instance_0$ (by choosing the path that appears
  in~$\instance_0$ whenever there is one).

  Now, assume by contradiction that $f$ is not a valid tiling. If there are $i,
  j \in \mathbb{N}$ such that $(f(i, j), f(i, j+1)) \in \mathbb{H}$, then
  consider the match $x \defeq a_i$, $y \defeq a_j$, $y' \defeq a_{j+1}$, $z \defeq
  b_{i,j}$, and $z' \defeq
  b_{i,j+1}$. If $i
  \neq j$ and $i \neq j+1$, we know that $G(a_i, a_j, b_{i,j})$ and $G(a_i,
  a_{j+1}, b_{i,j+1})$ hold, and
  taking the witnessing paths used to define $f(i,j)$ and $f(i,j+1)$, we obtain
  matches of~$Q_{f(i,j)}(b_{i,j})$ and $Q_{f(i,j+1)}(b_{i,j+1})$, so that we
  obtain a match of one of the disjuncts of~$Q$ (one of the first horizontal
  adjacency disjuncts), a contradiction. The cases where $i = j$ and where $i
  = j+1$ are similar and correspond to the second and third kinds of horizontal
  adjacency disjuncts. The case
  of $\mathbb{V}$ is handled similarly with the vertical adjacency disjuncts.
  Hence, $f$ is a valid tiling, which
  concludes the proof of the backward direction of the implication, shows the
  equivalence, and concludes the reduction and the undecidability proof.

  \subsubsection{Adapting to a Single Distinguished Relation}
  To prove the result with a single distinguished relation $S^\trans$, simply
  replace all occurrences of $C$ and $C^\trans$ in the query and constraints by
  $S$ and $S^\trans$. The rest of the construction is unchanged.
  The proof of the backwards direction is unchanged, using $S$ in place of $C$;
  what must be changed is the proof of the forward direction.
  
  Let $f$ be the
  solution to the tiling problem. We start by
  constructing a set of facts $\instance_1$ as before from~$f$ to complete
  $\instance_0$, replacing the $C$-facts
  in the gadgets by
  $S$-facts. Now, we complete $S^\trans$ to add the transitive closure of these
  paths (note that they are disjoint from any other $S$-fact), and 
  complete this to a set of facts to satisfy $\Sigma$: create $G$- and $G'$-facts, and create gadgets,
  this time taking all of them to have length $k+1$: this yields $\instance_2$.
  We repeat this last process indefinitely
  on the path of $S$-facts created in the gadgets of the previous iteration, and
  let $\instance$ be the result of this infinite process, which
  satisfies~$\Sigma$.
 
  We
  justify as before that $Q$ has no matches: as we create no
  $S'$-facts in $\instance_i$ for all $i > 1$, it suffices to observe that no
  new matches of $Q$ can include any of the new facts, because each disjunct
  includes an $S'$-fact. Hence, we can conclude as before.
\end{proof}

\subsubsection{Related Undecidability and Decidability Results}
The results that we have just shown in Theorem~\ref{thm:undectrans}~and~\ref{thm:undectransb}
complement the undecidability results of 
\citeA{andreaslidia}. Their Theorem~2 shows that $\owqatr$ is undecidable for
guarded TGDs, two transitive relations and atomic CQs, even with an empty set of
initial facts.
Their Corollary~1 shows that 
$\owqatr$ is undecidable with guarded disjunctive TGDs
(TGDs with disjunction in the head, and with an atom in the body
that guards all of the variables in the body) and UCQs,
even when restricted to arity-two signatures with
a single transitive relation that occurs only in guards,
and an empty set of initial facts.

Our results contrast with the decidability results of 
\citeA{mugnier15}, which apply to $\owqatr$ with linear rules (under a safety
condition, which they conjecture is not necessary for decidability): 
our Theorem~\ref{thm:undectrans} shows that $\owqatc$ with linear rules (without
imposing their condition) is undecidable.

\subsection{Undecidability Results for $\owqalin$} \label{sec:undecidlin}

Section~\ref{sec:decidlin} has shown that $\owqalin$ is decidable for
base-covered CQs and $\acgnf$ constraints. We now show that dropping the
base-covered requirement on the query leads to undecidability:

\newcommand{\undeccq}{
  There is a signature $\sigma = \sigmab \sqcup \sigmad$ where
  $\sigmad$ is a single strict linear order
  relation, a CQ $Q$ on~$\sigma$, and a set $\Sigma$ of
  inclusion dependencies on $\sigmab$ (i.e., not mentioning the linear order, so in
  particular base-covered), such that the following problem is undecidable:
  given a finite set of facts $\instance_0$, decide
  $\owqalin(\instance_0, \Sigma, Q)$.
}
\begin{theorem}
  \label{thm:undeccq}
  \undeccq
\end{theorem}

\begin{proof}
  We show the claim for a UCQ rather than a CQ, and explain in
  Appendix~\ref{app:ucqtocqundecid} how the proof extends to a CQ.
  As in the proof of Theorem~\ref{thm:undectrans}, we fix an undecidable
  infinite tiling problem $\mathbb{C}$, $\mathbb{V}$, $\mathbb{H}$,
  and will reduce that problem to the $\owqalin$ problem.

  \subsubsection{Definition of the Reduction}
  We consider the signature consisting of two binary relations~$R$ and~$D$ (for
  ``right'' and ``down''), $k-1$ unary relations $K_1, \ldots, K_{k-1}$
  (representing the colors), and one unary relation $S$ (representing the fact
  of being a vertex of the grid --- this  relation could be rewritten away and is
  just used to make the inclusion dependencies shorter to write).
  We also introduce the following abbreviations:
  \begin{inparaenum}[(i)]
    \item we let $K_1'(x)$ stand for $\exists y ~ x < y \land K_1(y)$;
    \item we let $K_k'(x)$ stand for $\exists y ~ x > y \land K_{k-1}(y)$;
    \item for all $1 < i < k$, we let $K_i'(x)$ stand for $\exists y y' ~
      K_{i-1}(y) \land y < x \land x < y' \land K_i(y')$.
  \end{inparaenum}
  Intuitively, the $K_i'$ describe the color of elements, which is encoded in
  their order relation to elements labeled with the~$K_i$.  
  
  We put the following inclusion dependencies in $\Sigma$:
  \begin{align*}
  \forall x \, S(x) \rightarrow \exists y \, R(x, y) &\qquad
  \forall x \, S(x) \rightarrow \exists y \, D(x, y) \\
  \forall x y \, R(x, y) \rightarrow S(y) &\qquad
  \forall x y \, D(x, y) \rightarrow S(y)
  \end{align*}
  We consider a UCQ formed of the following disjuncts (existentially closed):
  \begin{align*}
  R(x, y) \wedge D(x, z) \wedge R(z, w) \wedge D(y, w') \wedge w < w' \\
  R(x, y) \wedge D(x, z) \wedge R(z, w) \wedge D(y, w') \wedge w' < w \\
  \text{for each $(c, c') \in \mathbb{H}$}: R(x, y) \wedge K'_c(x) \wedge
      K'_{c'}(y) \\
  \text{for each $(c, c') \in \mathbb{V}$}: D(x, y) \wedge K'_c(x) \wedge
      K'_{c'}(y)
  \end{align*}
  Intuitively, the first two disjuncts enforce a grid structure, by saying that
  going right and then down must be the same as going down and then right. The
  other disjuncts enforce that there are no bad horizontal or vertical
  patterns.

  Given an instance $c_0, \ldots, c_n$ of
  the tiling problem, we construct an initial set of facts $\instance_0$ consisting of:
  \begin{inparaenum}[(i)]
    \item $S(a_0), \ldots, S(a_n)$ for fresh elements $a_0, \ldots, a_n$;
    \item $R(a_{i-1}, a_i)$ for $1 \leq i \leq n$ on these elements;
    \item $K_i(b_i)$ for $1 \leq i \leq k$ for fresh elements $b_1, \ldots, b_k$;
    \item for each $i$ such that $c_i$ is the color $C_1$, set $a_i < b_1$ (on
      the previously defined elements) ;
    \item for each $i$ such that $c_i$ is the color $C_k$, set $a_i > b_{k-1}$
      (on the previously defined elements);
    \item for each $1 < j < k$ and $i$ such that $c_i$ is $C_j$, set $b_{j-1} < a_i$
      and $a_i < b_j$ (on the previously defined elements).
  \end{inparaenum}

  \subsubsection{Correctness Proof for the Reduction}
  Let us show that the reduction is sound. Let us first assume that the tiling
  problem has a solution $f$. We construct a counterexample $\instance \supseteq \instance_0$ as a grid
  of the $R$ and $D$ relations, with the first elements of the first row being
  the $a_0, \ldots, a_n$, and with the color of elements being coded as their
  order relations to the $b_j$ like when constructing $\instance_0$
  above. Complete the
  interpretation of $<$ to a total order by choosing one arbitrary total order
  among the elements labeled with the same color, for each color. The resulting
  interpretation is indeed a total order relation, formed of the following: some total order on
  the elements of color $1$, the element $b_1$, some total order on the elements
  of color $2$, the element $b_2$, \ldots, the element $b_{k-2}$, some total order
  on the elements of color $k-1$, the element $b_{k-1}$, some total order on the
  elements of color $k$.

  It is immediate that the result satisfies $\Sigma$. To see why it does not
  satisfy the first two disjuncts of the UCQ, observe that any match of
  $R(x,y) \wedge D(x,z) \wedge R(z,w) \wedge D(y,w')$
  must have $w = w'$, by
  construction of the grid in $\instance$.
  To see why it does not satisfy the other disjuncts,
  notice that any such match must be a pair of two vertical or two horizontal
  elements; since the elements can match only one $K'_c$ which reflects their
  assigned color, the absence of matches follows by definition of $f$ being
  a tiling.

  Conversely, let us assume that there exists a counterexample $\instance \supseteq \instance_0$ which satisfies
  $\Sigma$ and violates $Q$. Clearly, if the first two disjuncts of $Q$ are
  violated, then, for any element where $S$ holds, considering its $R$ and
  $D$ successors that exist by $\Sigma$, and respectively their $D$ and $R$
  successors, we reach the same element. Hence, from $a_0, \ldots, a_n$, we can
  consider the part of $\instance$ defined as a grid of the $R$ and $D$ relations,
  and it is indeed a full grid ($R$ and $D$ edges occur everywhere they
  should), except that some elements may be reused at multiple places (but this
  does not matter). Now, we observe that any element except the $b_j$ must be inserted at
  some position in the total suborder $b_1 < \cdots < b_{k-1}$, so that at least
  one relation $K'_j$ holds for each element of the grid (several $K'_j$ may
  hold in case $\instance$ has more elements than the $b_i$ that are labeled
  with the $K_i$). Choose one of them, in a way
  that assigns to $a_0, \ldots, a_n$ the colors that they had in~$\instance_0$, and use
  this to define a function $f$ that extends $a_0, \ldots, a_n$. We claim
  that this $f$ indeed describes a tiling.

  Assume by contradiction that it does not. If there are   two horizontally
  adjacent values $(i, j)$ and $(i+1, j)$ realizing a configuration $(c, c')$
  from $\mathbb{H}$, by completeness of the grid there is an $R$-edge between
  the corresponding elements $u, v$ in $\instance$. Further, by the fact that $(i,j)$
  and $(i+1, j)$ were given the color that they have in $f$, we must have
  $K'_c(u)$ and $K'_c(v)$ in $\instance$, so that we must have had a match of a
  disjunct of $Q$, a contradiction. The absence of forbidden vertical patterns
  is proven in the same manner.
\end{proof}

Theorem~\ref{thm:undeccq}
implies that
the base-covered requirement is also necessary for constraints:

\begin{corollary}\label{cor:undeclin}
  There is a signature $\sigma = \sigmab \sqcup \sigmad$ where $\sigmad$ is a single strict
  linear order relation, and a set $\Sigma'$ of $\afgtgd$ constraints
  on~$\sigma$,
  such that, letting $\top$ be the tautological query, the following problem is
  undecidable:
  given a finite set of facts $\instance_0$, decide $\owqalin(\instance_0,
  \Sigma', \top)$.
\end{corollary}

\begin{proof}
To prove Corollary~\ref{cor:undeclin} from Theorem~\ref{thm:undeccq},
we take constraints $\Sigma'$
that are equivalent to $\Sigma \wedge \neg Q$,
  where $\Sigma$ and $Q$ are as in the previous theorem (in particular, $Q$ is a
  CQ).
Recall that $\Sigma$ is a set of inclusion dependencies on $\sigmab$,
and therefore are $\afgtgd$s.
Hence, it only remains to argue that $\neg Q$ can be written as a $\afgtgd$.
Indeed, if $Q = \exists \vec{x} ~ \phi(\vec{x})$
then consider the constraint 
$
\forall \vec{x} ( \phi(\vec{x}) \rightarrow \exists y ( y < y ) )
$
where $<$ is the distinguished relation.
Since $<$ must be a strict linear order in $\owqalin$,
$\exists y ( y < y)$ is equivalent to $\bot$ and
this new constraint is logically equivalent to $\neg Q$.
Moreover, this constraint is trivially in $\afgtgd$
since there are no frontier variables.
Hence, $\Sigma \wedge \neg Q$ can be written as a set of $\afgtgd$ constraints as claimed.
\end{proof}

\subsubsection{Related Undecidability Results}
The result in Theorem~\ref{thm:undeccq} is related to prior work by \shortciteA{Rosati07}
and \shortciteA{datatypes1},
which deals with query answering for UCQs and CQs with inequalities. These
results are related because
we can transform a query $Q$ using an inequality $x \neq y$
into a new UCQ query $Q' \vee Q''$, where $Q'$ and $Q''$
is the result of replacing $x \neq y$ in $Q$ with $x < y$ and $y < x$, respectively.
Further, we can also express
constraints of the form $\forall x y ( S_1(x,y) \wedge S_2(x,y) \rightarrow
\bot)$,
which are part of the description logics considered in those earlier papers,
as inclusion dependencies $\forall x y ( S_1(x,y) \wedge S_2(x,y) \rightarrow x < x)$.
Hence, we could use these prior results to show the undecidability of
$\owqalin$ for inclusion dependencies and a UCQ over $\sigma = \sigmab \sqcup \sigmad$
when $\sigmad$ is a single strict linear order. This 
is weaker than the result stated in Theorem~\ref{thm:undeccq}, which uses a CQ,
and in Corollary~\ref{cor:undeclin}, which uses a tautological query.

\section{Conclusion} \label{sec:conc}
We have given a detailed picture of the impact of transitivity, transitive closure,
and linear order restrictions on query answering problems for a broad class of
guarded constraints.
We have shown that transitive relations and transitive closure restrictions
can be handled in  guarded constraints as long
as the transitive closure relation is not needed as a guard. For linear
orders, the same is true if order atoms are covered by base atoms. This
implies the analogous results for frontier-guarded TGDs, in particular
frontier-one TGDs. We have built upon some known polynomial data complexity upper bounds 
for classes of guarded constraints without
distinguished relations, and have shown how to extend them to the setting
of distinguished relations that are required to be transitive
or a transitive closure.
However, in the case of distinguished relations required to be
linear orders, we have shown that $\ptime$ data complexity does not
always carry over.

All our results were shown in the absence of constants, so we leave open the question of
whether they still hold when constants are allowed in the constraints or
queries, though we believe that it should be possible to adapt the proofs.
A more important open question is that of deciding entailment over \emph{finite} sets of facts. 
There are few techniques for deciding entailment over finite sets of facts for logics where
it does not coincide with general entailment (and for the constraints considered here it does not coincide).
One exception can be found in an earlier work  \cite{kieronski2007finite}, which establishes decidability for
the guarded fragment with transitivity, under the two-variable restriction and
assuming that transitive relations appear only in guards.
Another exception is in the context of guarded logics 
(see \citeA{vincemikolaj}),
but it is not clear if the techniques there can be extended to our constraint languages.

\section*{Acknowledgements}

This is an extended version of the conference paper by \citeA{ijcai16-lintrans}.
It was originally published at JAIR \cite{amarilli2018query}. The present
version is identical to the JAIR paper, except that we fix an error in the proof of Theorem~\ref{thm:undectransb} (see
Appendix~\ref{app:fix}).
We are grateful to Andreas Pieris and Bailey Andrew for pointing out
the error to our attention.

Amarilli was partly funded by the T\'{e}l\'{e}com ParisTech Research Chair on Big Data
and Market Insights.
Bourhis was supported by CPER Nord-Pas de Calais/FEDER DATA Advanced Data Science and Technologies 2015-2020 and ANR Aggreg Project ANR-14-CE25-0017, INRIA Northern European Associate Team Integrated Linked Data.
Benedikt was sponsored by the Engineering and Physical Sciences
Research Council of the United Kingdom (EPSRC), under grants EP/M005852/1 and EP/L012138/1.
Vanden Boom was partially supported by EPSRC grant EP/L012138/1.

\renewcommand{\theHsection}{A\arabic{section}}
\appendix 
\bookmarksetup{startatroot}
\appendix

\section{Details about Tree Decompositions}\label{app:treedecomp}

Recall the definition of \emph{tree decompositions} from
Section~\ref{sec:decidtransautomata}.
This appendix presents two kinds of tree decompositions
for $\agnf$ that were used in the body of the paper.
The proofs use a standard technique,
involving an unravelling related to a variant of guarded negation bisimulation
due to \citeA{gnficalp}.
A related result and proof also appears in \citeA{lics16-gnfpup}.

\subsection{Proof of Proposition~\ref{prop:transdecomp}: Transitive-closure Friendly Tree Decompositions for $\agnf$}\label{app:tctreelike}

\newcommand{\gnk}{\gn^k}
\newcommand{\mydom}[1]{\kw{Dom}(#1)}

Recall the definition of having \emph{transitive-closure friendly $k$-tree-like
witnesses} from Section~\ref{sec:decidtransautomata}.
A sentence~$\phi$ over~$\sigma$ is said to have
transitive-closure friendly $k$-tree-like witnesses if:
for every finite set of $\sigmad$-facts $\instance_0$,
if there is a set of facts $\instance$
(finite or infinite)
extending $\instance_0$ with additional $\sigmab$-facts
such that $\instance$ satisfies $\phi$
when each $R^\trans$ is interpreted as the transitive closure of $R$,
then there is such an $\instance$
that has an $\instance_0$-rooted $(k-1)$-width
tree decomposition with countable branching.

In this section, we prove Proposition~\ref{prop:transdecomp}:

\quoteresult{Proposition~\ref{prop:transdecomp}}{\transdecomp}

If $\mysize{\phi} < 3$, then $\phi$ is
necessarily a single $0$-ary relation or its negation, in which case the
result is trivial, with $k = 1$.
Hence, in the rest of the proof,
we assume that $\mysize{\varphi} \geq 3$,
and $k$ will be chosen such that
$3 \leq k \leq \mysize{\varphi}$; specifically,
$k$ will be an upper bound on the maximum number of free variables
in any subformula of $\varphi$.

\subsubsection{Bisimulation Game}
We say that a set $X$ of elements from $\elems{\instance}$ is \emph{base-guarded} (or $\sigmab$-guarded)
if $\mysize{X} \leq 1$ or there is a $\sigmab$-fact in $\instance$ that uses all of the elements in $X$.
A \emph{partial rigid homomorphism} is a partial homomorphism with respect to all $\sigma$-facts in $\instance$,
such that
the restriction to any $\sigmab$-guarded set of elements is a partial isomorphism.

Let $\fA$ and $\fB$ be sets of facts extending $\instance_0$.
The
\emph{$\gn^k$ bisimulation game}
between $\fA$ and $\fB$
is an infinite game 
played by two players,
Spoiler and Duplicator.
The game has two types of positions: 
\begin{itemize}
 \item[i)] partial isomorphisms $f: X \to Y$ or $g:Y \to X$, 
     where $X \subset \elems{\fA}$ and $Y \subset \elems{\fB}$ are of size at most $k$ and $\sigmab$-guarded;
 \item[ii)] partial rigid homomorphisms $f: X \to Y$ or $g:Y \to X$, 
     where $X \subset \elems{\fA}$ and $Y \subset \elems{\fB}$ are of size at most $k$.  
\end{itemize}

From a type~(i) position $h$,
Spoiler must choose a finite subset $X' \subset \elems{\fA}$ or 
a finite subset $Y' \subset \elems{\fB}$, in either case of size at most $k$, 
upon which Duplicator 
must respond with a partial rigid homomorphism $h'$ with domain $X'$ or $Y'$ accordingly.
If $h : X \to Y$ and $h' : X' \to Y'$, then $h$ and $h'$ must agree on $X \cap X'$,
and if $h : X \to Y$ and $h' : Y' \to X'$, then $h^{-1}$ and $h'$ must agree on $Y \cap Y'$.
The analogous property must hold if $h : Y \to X$. In either case, the game then continues from
position~$h'$.

From a type~(ii) position $h : X \to Y$ (respectively, $h : Y \to X$),
Spoiler must choose a finite subset $X' \subset \elems{\fA}$
(respectively, $Y' \subset \elems{\fB}$)
of size at most $k$,
upon which Duplicator 
must respond by a partial rigid homomorphism with domain $X'$
(respectively, domain $Y'$),
such that $h$ and $h'$ agree on $X \cap X'$ (respectively, $Y \cap Y'$). Again,
the game continues from position~$h'$.

Notice that a type~(i) position is a special kind of type~(ii) position
where Spoiler has the option to \emph{switch the domain} to the other set of facts,
rather than just continuing to play in the current domain.

Spoiler wins if he can force the play into a position from which Duplicator cannot 
respond, and Duplicator wins if she can continue to play indefinitely.

A winning strategy for Duplicator in the $\gn^k$ bisimulation game
implies agreement between $\fA$ and $\fB$ on certain
$\agnf$ formulas.

\begin{proposition}\label{prop:bisim-game}
Let $\fA$ and $\fB$ be sets of facts extending $\instance_0$.
Let $\varphi(\vec{x})$ be a formula in $\agnf$,
and let $k \geq 3$ be greater than or equal to
the maximum number of free variables in any subformula of~$\varphi$.

If Duplicator has a winning strategy
in the $\gn^k$ bisimulation game
between $\fA$ and $\fB$
starting from a type~(i) or (ii) position $\vec{a} \mapsto \vec{b}$
and $\fA$ satisfies $\varphi(\vec{a})$ when interpreting each $R^\trans \in \sigmad$ as the transitive closure of $R \in \sigmab$,
then $\fB$ satisfies $\varphi(\vec{b})$ when interpreting each $R^\trans \in \sigmad$ as the transitive closure of $R \in \sigmab$.
\end{proposition}

\begin{proof}
For this proof,
  when we talk about sets of facts \emph{satisfying} a formula,
we mean satisfaction when interpreting $R^\trans \in \sigmad$ as the transitive closure of $R \in \sigmab$.
  We will abuse terminology slightly and say that $\varphi$ has \emph{width $k$}
if the maximum number of free variables in any subformula of $\varphi$ is at most $k$
(this is abusing the terminology since we are not assuming in this proof that $\varphi$ is in $\nf$).

We proceed by induction on two quantities, ordered lexicographically: first, the
  number of $\sigmad$-atoms in $\varphi$; second, the size of $\varphi$.
Suppose Duplicator has a winning strategy in the $\gn^k$ bisimulation game
between $\fA$ and $\fB$.

If $\varphi$ is a $\sigmab$-atom $A(\vec{x})$,
the result follows from the fact that the position
$\vec{a} \mapsto \vec{b}$ is a partial homomorphism.

Suppose $\varphi$ is a $\sigmad$-atom $R^\trans(x_1,x_2)$,
  and $\vec{a} = (a_1, a_2)$ and $\vec{b} = (b_1, b_2)$.
If $\fA, \vec{a}$ satisfies $R^\trans(x_1,x_2)$,
there is some $n \in \mathbb{N}$ such that $n > 0$ and
there is an $R$-path of length $n$
between~$a_1$ and~$a_2$ in $\fA$.
We can write a formula $\psi_n(x_1,x_2)$ in $\agnf$
(without any $\sigmad$-atoms)
that is satisfied exactly when there is an $R$-path of length $n$.
Since we do not need to write this in $\nf$,
we can express $\psi_n$ in $\agnf$ with width $3$
(maximum of 3 free variables in any subformula).
Since
$\fA,\vec{a}$ satisfies $\psi_n$
and $\psi_n$ does not have any $\sigmad$-atoms and $k \geq\nolinebreak 3$,
we can apply the inductive hypothesis from the type~(ii) position
$\vec{a} \mapsto \vec{b}$
to ensure that
$\fB,\vec{b}$ satisfies $\psi_n$,
and hence $\fB,\vec{b}$ satisfies~$\varphi$.

If $\varphi$ is a disjunction,
the result follows easily from the inductive hypothesis.

Suppose $\varphi$ is a base-guarded negation
$A(\vec{x}) \wedge \neg \varphi'(\vec{x}')$.
By definition of $\agnf$, it must be the case that $A \in \sigmab$
and $\vec{x}'$ is a sub-tuple of $\vec{x}$.
Since $\fA,\vec{a}$ satisfies $\varphi$,
we know that $\fA,\vec{a}$ satisfies~$A(\vec{x})$,
and hence $\vec{a}$ is $\sigmab$-guarded.
This means that $\vec{a} \mapsto \vec{b}$ is actually a partial isomorphism,
so we can view it as a position of type~(i).
This ensures that
$\fB,\vec{b}$ also satisfies $A(\vec{x})$.
It remains to show that it satisfies $\neg \varphi'(\vec{x}')$.
Assume for the sake of contradiction that it satisfies $\varphi'(\vec{x}')$.
Because $\vec{a} \mapsto \vec{b}$ is a type~(i) position, we can
consider the move in the game where Spoiler switches the domain to the other set of facts,
and then restricts to the elements in the subtuple $\vec{b}'$ of $\vec{b}$
corresponding to $\vec{x}'$ in $\vec{x}$.
Let $\vec{a}'$ be the corresponding subtuple of $\vec{a}$.
Duplicator must have a winning strategy
from the type~(i) position $\vec{b}' \mapsto \vec{a}'$,
so the inductive hypothesis ensures that
$\fA, \vec{a}'$ satisfies $\varphi'(\vec{x}')$,
a contradiction.

Finally, suppose $\varphi$
is an existentially quantified formula
$\exists y ( \varphi'(\vec{x},y) )$.
We are assuming that $\fA, \vec{a}$ satisfies $\varphi$.
Hence,
there is some $c \in \elems{\fA}$ such that
$\fA,\vec{a}, c$ satisfies $\varphi'$.
Because the width of $\varphi$ is at most $k$,
we know that the combined number of elements in $\vec{a}$ and $c$ is at most $k$.
Hence, we can consider the move in the game where Spoiler
selects the elements in $\vec{a}$ and $c$.
Duplicator must respond with $\vec{b}$ for $\vec{a}$,
and some $d$ for $c$.
This is a valid move in the game,
so Duplicator must still have a winning strategy from this position $\vec{a}c \mapsto \vec{b}d$,
and the inductive hypothesis implies that $\fB, \vec{b}, d$ satisfies $\varphi'$.
Consequently, $\fB, \vec{b}$ satisfies $\varphi$.
\end{proof}

\subsubsection{Unravelling}
The tree-like witnesses for Proposition~\ref{prop:transdecomp}
can be obtained using an unravelling construction
related to the $\gn^k$ bisimulation game.
This unravelling construction is adapted from \citeA{lics16-gnfpup}.

Fix a set of facts $\instance$ that extends $\instance_0$ with additional $\sigmab$-facts.
Consider the set $\Pi$ of sequences of the form $X_0 X_1 \dots X_n$,
where $X_0 = \elems{\instance_0}$, and for all $i \geq 1$,
$X_i \subseteq \elems{\fA}$ with $\mysize{X_i} \leq k$.

We can arrange these sequences in a tree based on the prefix order.
Each sequence $\pi = X_0 X_1 \dots X_n$ identifies a unique node in the tree;
we say $a$ is \emph{represented} at node $\pi$
if $a \in X_n$.
For $a \in \elems{\fA}$, we say $\pi$ and $\pi'$ are \emph{$a$-equivalent}
if $a$ is represented at every node on the
unique minimal path between $\pi$ and $\pi'$
in this tree.
For $a$ represented at $\pi$, we write $[\pi, a]$ for the $a$-equivalence class.

The \emph{$\gn^k$-unravelling of $\fA$} is a set of facts $\structureunravelk{\fA}$
over elements $\set{ [\pi, a] : \pi \in \Pi \text{ and } a \in \elems{\fA}}$
with $S([\pi_1,a_1],\dots,[\pi_j,a_j]) \in \structureunravelk{\fA}$
iff $S(a_1, \dots, a_j) \in \fA$ and
there is some $\pi \in \Pi$ such that $[\pi,a_i] = [\pi_i,a_i]$ for $i \in \set{1,\dots,j}$.
We can identify $[\epsilon,a]$ with the element $a \in \elems{\instance_0}$,
so $\structureunravelk{\fA}$ extends $\instance_0$.
Hence, there is a natural $\instance_0$-rooted tree decomposition of width $k-1$
for $\structureunravelk{\fA}$
induced by the tree of sequences from $\Pi$.

Because this unravelling is related so closely to the $\gn^k$-bisimulation game,
it is straightforward to show that Duplicator
has a winning strategy in the bisimulation game
between $\fA$ and its unravelling.

\begin{proposition}\label{prop:unravel}
Let $\fA$ be a set of facts extending $\instance_0$ with additional $\sigmab$-facts,
and let $\structureunravelk{\fA}$ be the $\gnk$-unravelling of $\fA$.
Then Duplicator has a winning strategy in the $\gn^k$ bisimulation game
between $\fA$ and~$\structureunravelk{\fA}$.
\end{proposition}

\begin{proof}
Given a position $f$ in the $\gnk$-bisimulation game,
we say the \emph{active set} is the set of facts
containing the elements in the domain of $f$.
In other words, the active set is either $\fA$ or $\structureunravelk{\fA}$,
depending on which set Spoiler is currently playing in.
The \emph{safe positions} $f$ in the $\gnk$-bisimulation game
between $\fA$ and $\structureunravelk{\fA}$ are defined as follows:
if the active set is $\structureunravelk{\fA}$,
then $f$ is safe if for all $[\pi,a] \in \mydom{f}$, $f([\pi,a]) = a$;
if the active set is $\fA$,
then $f$ is safe if
there is some $\pi$ such that $f(a) = [\pi,a]$ for all $a \in \mydom{f}$.

We now argue that starting from a safe position $f$,
Duplicator has a strategy to move to a new safe position $f'$.
This is enough to conclude that Duplicator has a winning strategy
in the $\gnk$-bisimulation game
between $\fA$ and $\structureunravelk{\fA}$
starting from any safe position.

First, assume that the active set is $\structureunravelk{\fA}$.
\begin{itemize}
\item If $f$ is a type~(ii) position,
then Spoiler can select some new set $X'$ of elements from the active set.
Each element in $X'$ is of the form $[\pi',a']$.
Duplicator must choose $f'$ such that $[\pi',a']$ is mapped to $a'$ in $\fA$,
in order to maintain safety.
This new position $f'$ is consistent with $f$ on any elements in $X' \cap \mydom{f}$
since $f$ is safe.
This $f'$ is still a partial homomorphism since any relation holding for a tuple of
elements $[\pi_1,a_1], \dots, [\pi_n,a_n]$ from $\mydom{f'}$ must hold
for the tuple of elements $a_1, \dots, a_n$ in $\fA$ by definition of $\structureunravelk{\fA}$.
Consider some element $[\pi',a']$ in $\mydom{f'}$.
It is possible that there is some $[\pi,a']$ in $\mydom{f'}$ with $[\pi,a'] \neq [\pi',a']$;
however, $[\pi,a']$ and $[\pi',a']$ are not base-guarded in $\structureunravelk{\fA}$.
Hence, any restriction $f''$ of $f'$ to
a base-guarded set of elements is a bijection.
Moreover, such an $f''$ is a
partial isomorphism:
consider some $a_1,\dots,a_n$ in the range of $f''$ for which some relation $S$ holds in $\fA$;
since $(f'')^{-1}(a_1),\dots,(f'')^{-1}(a_n)$ must be base-guarded,
we know that there is some $\pi$ such that
$[\pi,a_1] = (f'')^{-1}(a_1)$,$\dots$, $[\pi,a_n] = (f'')^{-1}(a_n)$,
so by definition of $\structureunravelk{\fA}$,
$S$ holds of $(f'')^{-1}(a_1),\dots,(f'')^{-1}(a_n)$ as desired.
Hence, $f'$ is a safe partial rigid homomorphism.
\item If $f$ is a type~(i) position,
then Spoiler can either choose elements in the active set
and we can reason as we did for the type~(ii) case,
or Spoiler can select elements from the other set of facts.

We first argue that if Spoiler changes the active set
and chooses no new elements,
then the game is still in a safe position.
Since $f$ is a type~(i) position, 
we know that $\mydom{f}$ is guarded
by some base relation $S$, so there is some $\pi$ with $f(a) = [\pi,a]$ for all $a \in \mydom{f}$
by construction of $\structureunravelk{\fA}$.
Hence, the new position $f'=f^{-1}$ is still safe.

If Spoiler switches active sets and chooses new elements,
then we can view this as two separate moves:
in the first move, Spoiler switches active sets from $\structureunravelk{\fA}$ to $\fA$
but chooses no new elements,
and in the second move, Spoiler selects the desired new elements from $\fA$.
Because switching active sets leads to a safe position
(by the argument in the previous paragraph),
it remains to define Duplicator's safe strategy when the active set is $\fA$,
which we explain below.
\end{itemize}
Now assume that the active set is $\fA$.
Since $f$ is safe, there is some $\pi$ such that $f(a) = [\pi,a]$ for all $a \in \mydom{f}$.
\begin{itemize}
\item If $f$ is a type~(ii) position, then Spoiler can select some new set $X'$
of elements from the active set.
We define the new position $f'$ chosen by Duplicator
to map each element $a' \in X'$
to~$[\pi',a']$ where $\pi' = \pi \cdot X'$.
By construction of the unravelling, $\pi' \in \Pi$
and the resulting partial mapping $f'$ still satisfies the safety property with $\pi'$
as witness.
Note that $f'$ is consistent with $f$ for elements of~$X'$ that are also
in~$\mydom{f}$, as we have
$[\pi \cdot X',a'] = [\pi,a']$ for $a' \in X' \cap \mydom{f}$.
Now consider some tuple $\vec{a} = a_1 \dots a_n$
of elements from $\mydom{f'}$ that are in some relation $S$.
We know that $f'(a_i) = [\pi',a_i]$,
hence $S$ must hold for $f'(\vec{a})$ in $\structureunravelk{\fA}$.
Moreover,
for any base-guarded set $\vec{a} = \set{ a_1, \dots, a_n }$ of distinct elements from
$\mydom{f'}$, $f'(\vec{a})$ must yield a set of distinct elements $\set{ f'(a_1), \dots, f'(a_n) }$,
and these elements can only participate in some fact in $\structureunravelk{\fA}$
if the underlying elements from $\vec{a}$ participate in the same fact in $\fA$.
Hence, $f'$ is a safe partial rigid homomorphism.
\item If $f$ is a type~(i) position,
then Spoiler can either choose elements in the active set
and we can reason as we did for the type~(ii) case,
or Spoiler can select elements from the other set of facts.
It suffices to argue that if Spoiler changes the active set like this,
and chooses no new elements,
then the game is still in a safe position.
But in this case $f' = f^{-1}$ is easily seen to still be safe.
\end{itemize}
This concludes the proof of Proposition~\ref{prop:unravel}.
\end{proof}

\subsubsection{Countable Witnesses}

The last theorem that we need says that we can always obtain a countable witness to satisfiability of $\agnf$
with transitivity. This follows from known results about least fixed point logic.
\begin{theorem}\label{thm:countable}
For $\phi \in \agnf$ and a finite
set of facts $\instance_0$,
  if there is an
  $\instance$ extending $\instance_0$ with $\sigmab$-facts and
  satisfying $\phi$ when each relation $R^\trans \in \sigmad$ is interpreted as the transitive closure
of the corresponding $R \in \sigma$, then there  is
  such an $\instance$ 
  that has countable cardinality.
\end{theorem}

\begin{proof}
It is well-known that transitive closure can be expressed in $\kw{LFP}$, the extension of $\fo$ with a fixpoint operator.
Hence, it is possible to rewrite $\phi \in \agnf$ into $\phi' \in \kw{LFP}$ just over the base relations
such that $\phi$ is satisfiable (interpreting the distinguished relations as the transitive closure of the corresponding base relations)
iff $\phi'$ is satisfiable (with no special interpretations).
By Theorem~2.1 of \cite{Gradel02}, which is 
essentially a consequence of the L\"owenheim-Skolem property,
$\kw{LFP}$---and hence $\agnf$---has the property that if there is a satisfying set of facts,
then there is some satisfying set of facts of countable cardinality.
\end{proof}

\subsubsection{Concluding the Proof}

We can now conclude the proof of Proposition~\ref{prop:transdecomp}.
Assume that $\fA$ is a set of facts extending $\instance_0$ with additional $\sigmab$-facts
such that $\fA$ satisfies $\varphi$
when interpreting $R^\trans$ as the transitive closure of $R$.
By Theorem~\ref{thm:countable}, we can assume that $\elems{\fA}$ has countable cardinality.
Let $3 \leq k \leq \mysize{\varphi}$ be an upper bound on the
maximum number of free variables in any subformula of~$\varphi$.
It is easy to see that the unravelling $\structureunravelk{\fA}$ of a countable set of facts $\fA$
has countable branching. It remains to show that $\structureunravelk{\fA}$ satisfies $\instance_0$
and $\varphi$.

Since $\fA$ satisfies $\varphi$,
Propositions~\ref{prop:bisim-game}~and~\ref{prop:unravel} imply that
$\structureunravelk{\fA}$ also satisfies $\varphi$ when properly interpreting $R^\trans$.
Moreover, 
for each $\sigmad$-fact $R^\trans(c,d) \in \instance_0$,
there is some $n \in \mathbb{N}$ such that $n > 0$ and
there is an $R$-path of length $n$
between $c$ and $d$ in $\instance$.
We can write a formula $\psi_n(c,d)$ in $\agnf$
(without any $\sigmad$-atoms)
that is satisfied exactly when there is an $R$-path of length $n$.
Since $\psi_n$ can be expressed in $\agnf$
with at most 3 free variables in any subformula
and $\fA \models \psi_n(c,d)$,
Propositions~\ref{prop:bisim-game}~and~\ref{prop:unravel} imply that
$\structureunravelk{\fA}$ also satisfies $\psi_n$,
and hence $\structureunravelk{\fA} \models R^\trans(c,d)$.
Thus, we can conclude that the unravelling $\structureunravelk{\fA}$
is a transitive-closure friendly $k$-tree-like witness
for~$\varphi$.

\subsection{Proof of Lemma~\ref{lemma:guarded-interface-dec}: Base-guarded-interface Tree Decompositions for $\agnf$}\label{app:guarded-interface}

Recall the definition of a \emph{base-guarded-interface tree decomposition}, and
of having \emph{base-guarded-interface $k$-tree-like witnesses}, from
Section~\ref{sec:decidingapprox}. In this section, we prove
Lemma~\ref{lemma:guarded-interface-dec} from the main text, with the bulk of
the work being to prove that $\agnf$ sentences have base-guarded-interface
$k$-tree-like witnesses.

For the application that we have in mind, 
we will prove our results for a slight generalization of $\nf$ $\agnf$ formulas
that
allows
$\guardedb(x,y)$
(a disjunction over all existentially-quantified atoms that could base-guard $x$ and $y$)
in place of an explicit base-guard.
Such an atom $\guardedb(x,y)$ is called a \emph{generalized base-guard}
since
it can express that a pair of elements is guarded without indicating the exact atom
that is guarding $x$ and $y$, and without worrying about other variables that
may appear in the guard atom.
Note that allowing these generalized base-guards does not increase the expressivity of the logic,
but it is convenient to allow them (e.g.~in the definition of the $k$-guardedly linear axioms).

We can now state the following result, which we will prove in the rest of this
section, and then extend to show Lemma~\ref{lemma:guarded-interface-dec}:

\begin{proposition}\label{prop:guarded-interface-dec-appendix}
  Every sentence $\phi$ in $\nf$ $\agnf$ (possibly with generalized base-guards)
  has base-guarded-interface $k$-tree-like witnesses
  where $k$ is the width of $\phi$.
\end{proposition}

The result and proof of Proposition~\ref{prop:guarded-interface-dec-appendix}
is very similar to Proposition~\ref{prop:transdecomp}.
However, unlike Proposition~\ref{prop:transdecomp},
we do not interpret the distinguished relations in a special way here.
This allows us to prove the stronger base-guarded-interface property
about the corresponding tree decompositions,
which is important in some arguments
(e.g., Proposition~\ref{prop:rewritelin} and Theorem~\ref{thm:ptimetransdataupper}).

We first consider a variant of the $\gn^k$ bisimulation game defined earlier in
Appendix~\ref{app:tctreelike}.
The positions in the game are the same as before:
\begin{itemize}
 \item[i)] partial isomorphisms $f:X \to Y$ or $g:Y \to X$, 
     where $X \subset \elems{\fA}$ and $Y \subset \elems{\fB}$ are of size at most $k$ and $\sigmab$-guarded;
 \item[ii)] partial rigid homomorphisms $f:X \to Y$ or $g:Y \to X$, 
     where $X \subset \elems{\fA}$ and $Y \subset \elems{\fB}$ are of size at most $k$.  
\end{itemize}
However, the rules of the game are different.

From a type (i) position $h$,
Spoiler must choose a finite subset $X' \subset \elems{\fA}$ or 
a finite subset $Y' \subset \elems{\fB}$, in either case of size at most $k$, 
upon which Duplicator 
must respond by a partial rigid homomorphism $h'$ with domain $X'$ or $Y'$ accordingly
that is consistent with~$h$.
(This is the same as before).

In a type (ii) position $h$,
Spoiler is only allowed to select some base-guarded subset $X'$ of $\mydom{h}$ (rather than an arbitrary subset of size $k$),
and then the game proceeds from the type (i) position $h'$
obtained by restricting $h$ to this base-guarded subset.

Thus, unlike in the game presented in Appendix~\ref{app:tctreelike}, this game strictly alternates between type (ii) positions and base-guarded positions of type~(i).
We call this a \emph{base-guarded-interface $\gn^k$ bisimulation game},
since the interfaces (shared elements) between the domains of consecutive positions
must be base-guarded.
We can then show the analogue of Proposition~\ref{prop:bisim-game} for this
variant of the game (note that this time we do not handle the distinguished relations in
any special way):

\begin{proposition}\label{prop:bisim-game-guarded}
Let $\fA$ and $\fB$ be sets of facts extending $\instance_0$.
  Let $\varphi(\vec x)$ be a formula in $\nf$ $\agnf$ of width at most~$k$.

If Duplicator has a winning strategy in the
base-guarded-interface $\gn^k$ bisimulation game
between $\fA$ and $\fB$
starting from a type (i) position $\vec{a} \mapsto \vec{b}$
and $\fA$ satisfies $\varphi(\vec{a})$,
then $\fB$ satisfies $\varphi(\vec{b})$.
\end{proposition}

\begin{proof}
Suppose Duplicator has a winning strategy in the base-guarded-interface
$\gn^k$ bisimulation game
between $\fA$ and $\fB$.
We proceed by induction on the negation depth of $\phi$.

If $\phi$ has negation depth 0, then it is a UCQ.
We are assuming that $\fA, \vec{a}$ satisfies $\varphi$,
so there is some CQ $\delta = \exists \vec{y} ( \chi_1 \wedge \dots \wedge \chi_j )$
that is satisfied by $\fA, \vec{a}$.
Hence,
there is some $\vec{c} \in \elems{\fA}$ such that
$\fA,\vec{a},\vec{c}$ satisfies $\chi_1 \wedge \dots \wedge \chi_j$.
Because the width of $\varphi$ is at most $k$,
we know that the combined number of elements in $\vec{a}$ and $\vec{c}$ is at most $k$.
Hence, we can consider the move in the game where Spoiler
selects the elements in $\vec{a}$ and $\vec{c}$.
Duplicator must respond with some
$\vec{d} \in \elems{\fB}$
such that $\vec{a}\vec{c} \mapsto \vec{b}\vec{d}$
is a partial rigid homomorphism, a type (ii) position.
Each conjunct $\chi_i$ must be satisfied by~$\fB, \vec{b}\vec{d}$
since $\vec{a}\vec{c} \mapsto \vec{b}\vec{d}$
is a partial homomorphism with respect to $\sigma$.
Hence, $\fB, \vec{b}\vec{d} \models \delta$,
and $\fB, \vec{b} \models \varphi$ as desired.

Now assume that $\phi$ has negation depth $d > 0$
and is of the special form
$\alpha(\vec{x}) \wedge \neg \varphi'(\vec{x}')$.
By definition of $\agnf$, it must be the case that $\vec{x}'$ is a sub-tuple of $\vec{x}$,
and $\alpha$ is either a base-atom or a generalized base-guard for $\vec{x}'$.
Since $\fA,\vec{a}$ satisfies $\varphi$,
we know that $\fA,\vec{a}$ satisfies $\alpha(\vec{x})$,
which implies (by the inductive hypothesis) that
$\fB,\vec{b}$ also satisfies $\alpha(\vec{x})$.
It remains to show that $\fB$ satisfies~$\neg \varphi'(\vec{x}')$.
Assume for the sake of contradiction that it satisfies $\varphi'(\vec{x}')$.
Because $\vec{a} \mapsto \vec{b}$ is a type~(i) position, we can
consider the move in the game where Spoiler switches the domain to the other set of facts,
keeps the same set of elements,
and then collapses to the base-guarded elements in the subtuple $\vec{b}'$ of $\vec{b}$
corresponding to $\vec{x}'$ in $\vec{x}$.
Let $\vec{a}'$ be the corresponding subtuple of~$\vec{a}$.
Duplicator must still have a winning strategy
from this new type (i) position $\vec{b}' \mapsto \vec{a}'$,
so the inductive hypothesis ensures that
$\fA, \vec{a}'$ satisfies $\varphi'(\vec{x}')$,
a contradiction.

Finally consider an arbitrary UCQ-shaped formula $\phi$ with negation depth $d > 0$.
If $\fA, \vec{a} \models \varphi$,
then there is some
disjunct
$\exists \vec{y} ( \chi_1 \wedge \dots \wedge \chi_j )$ in $\varphi$
and some $\vec{c} \in \elems{\fA}$ such that
$\fA,\vec{a},\vec{c}$ satisfies $\chi_1 \wedge \dots \wedge \chi_j$.
Because the width of $\varphi$ is at most $k$,
we know that the combined number of elements in $\vec{a}$ and $\vec{c}$ is at most $k$.
Hence, we can consider the move in the game where Spoiler
selects the elements in $\vec{a}$ and $\vec{c}$.
Duplicator must respond with some
$\vec{d} \in \elems{\fB}$
such that $\vec{a}\vec{c} \mapsto \vec{b}\vec{d}$
is a partial rigid homomorphism, a type (ii) position.
Now consider the possible shape of the $\chi_i$.
If $\chi_i$ is a $\sigma$-atom, then it must be satisfied in $\fB, \vec{b}\vec{d}$
since $\vec{a}\vec{c} \mapsto \vec{b}\vec{d}$
is a partial homomorphism with respect to $\sigma$.
Otherwise, $\chi_i$ is of the form $\alpha(\vec{x}_i\vec{y}_i) \wedge \neg \varphi'$
where $\vec{x}_i$ and $\vec{y}_i$ are subtuples of variables from $\vec{x}$ and $\vec{y}$ that are actually used by this $\chi_i$,
and $\alpha$ is a base-guard for the free variables of $\varphi'$.
We know that $\fA, \vec{a}_i \vec{b}_i \models \alpha(\vec{x}_i \vec{y}_i) \wedge \neg \varphi'$.
We can consider Spoiler's restriction of $\vec{a}\vec{c}$ to the subset of elements $\vec{a}_i \vec{b}_i$ that correspond to $\vec{x}_i \vec{y}_i$
and the corresponding restriction of $\vec{b}\vec{d}$ to $\vec{b}_i\vec{d}_i$.
This is a valid move to a type (i) position $\vec{a}_i \vec{b}_i \mapsto \vec{b}_i \vec{d}_i$,
since $\fA, \vec{a}_i \vec{b}_i \models \alpha(\vec{x}_i \vec{y}_i)$
and the definition of $\agnf$ requires that $\alpha$ is a base atom.
Since Duplicator must still have a winning strategy from this new type (i) position,
the previous case implies that this $\chi_i$ is also satisfied by $\fB, \vec{b}_i \vec{d}_i$.
Since this is true for all $\chi_i$ in the CQ-shaped formula,
$\fB, \vec{b}\vec{d} \models \exists \vec{y} ( \chi_1 \wedge \dots \wedge \chi_j )$, and
$\fB,\vec{b}$ satisfies $\varphi$ as desired.
\end{proof}

We then use a variant of the unravelling based on this game.
The \emph{base-guarded-interface $\gn^k$-unravelling} $\structureunravelkint{\fA}$
is defined in a similar fashion to the $\gn^k$-unravelling,
except it uses only sequences
$\Pi \cap \set{ X_0 \dots X_n : \text{for all $i \geq 1$, $X_i \cap X_{i+1}$ is $\sigmab$-guarded}}$.
This unravelling has an $\instance_0$-rooted base-guarded-interface tree decomposition of width $k-1$.
Moreover, we can show the analogue of Proposition~\ref{prop:unravel}:

\begin{proposition}\label{prop:unravelguarded}
Let $\fA$ be a set of facts extending $\instance_0$,
  and let $\structureunravelkint{\fA}$ be
  the base-guarded-interface $\gnk$-unravelling of $\fA$.
Then Duplicator has a winning strategy in the base-guarded-interface $\gn^k$ bisimulation game
between $\fA$ and~$\structureunravelkint{\fA}$.
\end{proposition}

\begin{proof}
The proof is similar to Proposition~\ref{prop:unravel}.
The interesting part of the argument is when
Spoiler selects some new elements $X'$ in $\fA$
starting from a safe position $f$
(for which there is some $\pi$ such that $f(a) = [\pi,a]$ for all $a \in \mydom{f}$).
We need to show that
$[\pi',a']$ for $a' \in X'$ and $\pi' = \pi \cdot X'$ is well-defined in $\structureunravelkint{\fA}$.
This is well-defined only 
if the overlap
between the elements in~$\pi$ and~$\pi'$ is 
base-guarded.
But because the base-guarded-interface $\gnk$ bisimulation game
strictly alternates between type~(i) and~(ii) positions,
Spoiler can only select new elements $X'$ in a type~(i) position,
so the overlap satisfies this requirement.
The remainder of the proof is the same as in Proposition~\ref{prop:unravel}.
\end{proof}

We can conclude the proof of Proposition~\ref{prop:guarded-interface-dec-appendix} as follows.
Assume that $\fA$ is a set of facts that satisfies $\varphi$ in $\nf$ of width $k$.
Since $\fA$ satisfies $\varphi$,
Propositions~\ref{prop:bisim-game-guarded}~and~\ref{prop:unravelguarded} imply that
$\structureunravelkint{\fA}$ also satisfies $\varphi \in \agnf$.
Hence, $\structureunravelkint{\fA}$
is a base-guarded-interface $k$-tree-like witness
for $\varphi$.
This completes the proof of Proposition~\ref{prop:guarded-interface-dec-appendix}.

\myparagraph{Concluding the proof}
We can now use Proposition~\ref{prop:guarded-interface-dec-appendix} 
to prove Lemma~\ref{lemma:guarded-interface-dec},
which says that
there are base-guarded-interface $k$-tree-like witnesses
even when we
extend $\Sigma \in \acgnf$ to~$\Sigma'$ that includes the $k$-guardedly linear
axioms. Recall the formal statement:
\quoteresult{Lemma~\ref{lemma:guarded-interface-dec}}{\guardedinterfacedec}

\begin{proof}
By Proposition~\ref{prop:nf}, $\Sigma \wedge \neg Q$ is equivalent to a formula in $\nf$
with width at most $\mysize{\Sigma \wedge \neg Q}$.
Hence, Proposition~\ref{prop:guarded-interface-dec-appendix}
implies that
$\Sigma \wedge \neg Q$ has a base-guarded-interface $k$-tree-like witness
for some $k \leq \mysize{\Sigma \wedge \neg Q}$.
So, in particular, it has one for $k \colonequals \max(\mysize{\Sigma \wedge \neg Q}, \arity{\sigma \cup \{G\}})$.
To prove this lemma, then, it suffices to argue that
the \mbox{$k$-guardedly linear} axioms can also be written
in $\nf$ $\agnf$ with generalized base-guards and width at most
$k$.

The width of $\guardedbg(x,y)$ is
$\arity{\sigma \cup \{G\}} \leq k$.
The guardedly total axiom is written in $\nf$ $\agnf$
as
$
\neg \exists x y ( \guardedbg(x,y) \wedge \neg ( x = y \vee  x \drel y \vee y \drel x ) )
$
so it has width $\max (2, \arity{\sigma \cup \{G\}}) \leq k$.
The irreflexive axiom is already written in normal form
$\agnf$ with width $1 \leq k$.
For the $k$-guardedly transitive axioms,
  $\psi_l(x,y)$ has width \mbox{$\max (l+1, \arity{\sigma \cup \{G\}})$}
  and $\psi_l(x,x)$ has width \mbox{$\max (l, \arity{\sigma \cup \{G\}})$}, so 
each of the $k$-guardedly transitive axioms has width at most $k$.
Overall, this means that
the $k$-guardedly linear restriction can be expressed in $\nf$ $\agnf$ with generalized base-guards and width $k$ as required.
\end{proof}

\section{Details about Automata}\label{app:automata}

We now give details of the properties of automata used in the
body.

\subsection{Closure Properties}\label{app:closure}
We recall some closure properties of
$\ptwowayaltinf$
and $\ponewayndinf$,
omitting the standard proofs \cite{Thomas97,Loding-unpublished},
and following Section~5 of~\cite{lics16-gnfpup}.
Note that we state only the size of the automata for each property,
but the running time of the procedures constructing these automata
is always polynomial in the output size.

First, the automata that we are using
are closed under union and intersection
(of their languages).

\begin{proposition}\label{prop:closure-union-intersection}
$\ptwowayaltinf$
are closed under union and intersection,
with only a polynomial blow-up in the number of states, priorities, and overall size.
The same holds of $\ponewayndinf$.
\end{proposition}

For example, this means that
if we are given two $\ptwowayaltinf$ $\cA_1$ and $\cA_2$,
then we can construct in $\ptime$ a $\ptwowayaltinf$ $\cA$ such that
$L(\cA) = L(\cA_1) \cap L(\cA_2)$.

Another important language operation is \emph{projection}.
Let $L'$ be a language of trees over a tree signature $\Gamma \cup \set{P}$.
The \emph{projection} of $L'$ with respect to $P$
is the language of trees $\tree$ over $\Gamma$ such that
there is some $\tree' \in L'$ such that $\tree$ and $\tree'$
agree on all unary relations in $\Gamma$.
Projection is easy for nondeterministic automata
since the valuation for $P$ can be guessed by Eve.

\begin{proposition}\label{prop:closure-projection}
$\ponewayndinf$
are closed under projection,
with no change in the number of states, priorities, and overall size.
\end{proposition}

Finally, complementation is easy for alternating automata
by taking the \emph{dual} automaton,
obtained by switching conjunctions and disjunctions in the transition function,
and incrementing all of the priorities by one.

\begin{proposition}\label{prop:dual}
$\ptwowayaltinf$ are closed under complementation,
with no change in the number of states, priorities, and overall size.
\end{proposition}

\subsection{Proof of Theorem~\ref{thm:localization}: Localization}

Recall the statement:
\quoteresult{Theorem~\ref{thm:localization}}{\localization}

This result was known already in the literature.
For completeness, we present the construction here,
following the presentation from \citeA{lics16-gnfpup}.

Let $\cA'$ be a $\ponewayndinf$ on $\Gamma'$-trees,
with state set $Q_{\cA'}$, initial state $q^0_{\cA'}$,
transition function $\delta_{\cA'}$, and priority function $\Omega_{\cA'}$.
Since $\cA'$ is nondeterministic, $\delta_{\cA'}$ maps each state and label to a disjunction
of formulas of the form $(\dleft,r) \wedge (\dright,r')$.

We construct $\cA$ as follows.
The state set of $\cA$ is $\set{q_0} \cup Q_{\cA'} \cup (\set{ \dleft, \dright} \times Q_{\cA'})$.
The initial state is $q_0$.
We call states of the form $q \in Q_{\cA'}$ \emph{downwards mode states},
and states of the form $(d,q) \in (\set{ \dleft, \dright} \times Q_{\cA'})$ \emph{upwards mode states}.

We now define the transition function $\delta$ of~$\cA$.
In initial state $q_0$, we set $\delta(q_0,\treelab)$ to be:

{\scriptsize
\begin{align*}
\begin{cases}
\bigvee \set{ (d,r) \wedge (d',r') \wedge (\dup,(\dleft,q))  : \text{$q \in
  Q_{\cA'}$ and $(d,r) \wedge (d',r')$ is a disjunct in $\delta_{\cA'}(q,\treelab \cup
  \set{P_1,\dots,P_j})$}} &\!\!\!\text{if $\dleft \in \treelab$} \\
\bigvee \set{ (d,r) \wedge (d',r') \wedge (\dup,(\dright,q))  : \text{$q \in
  Q_{\cA'}$ and $(d,r) \wedge (d',r')$ is a disjunct in $\delta_{\cA'}(q,\treelab \cup
  \set{P_1,\dots,P_j})$}} &\!\!\!\text{if $\dright \in \treelab$} \\
\bigvee \set{ (d,r) \wedge (d',r') : \text{$(d,r) \wedge (d',r')$ is in
  $\delta(q_{\cA'}^0,\treelab \cup \set{P_1,\dots,P_j})$}} &\!\!\!\text{if
  $\dleft,\!\dright\!\notin\!\treelab$}
\end{cases}\end{align*}}

In other words, Eve guesses some $q \in Q_{\cA'}$
and some disjunct $(d,r) \wedge (d',r')$ consistent with the transition function of $\delta_{\cA'}$
in state $q$, assuming that the starting label also includes $P_1,\dots,P_j$.
Adam can either (i) challenge her to show that the automaton accepts from state $q$
by moving in direction $d$ or $d'$ and switching to downwards mode state $r$ or $r'$, respectively, 
or (ii) challenge her to show that she could actually reach state $q$
by switching to upwards mode $(d'',q)$ where $d''$ is $\dleft$ if the starting node was the left child of its parent
and $\dright$ if it was the right child of its parent.
In the special case that the starting node is the root ($\dleft$ and $\dright$ are not in the label),
then $q$ must be $q^0_{\cA'}$ and Adam can only challenge downwards.

In upwards mode state $(d,r)$, we set $\delta((d,r),\treelab)$ to be:

{\scriptsize
\begin{align*}
\begin{cases}
\bigvee \set{ (d',r') \wedge (\dup,(\dleft,q))  : \text{$q \in Q_{\cA'}$ and $(d,r) \wedge (d',r')$ or $(d',r') \wedge (d,r)$ is a disjunct in $\delta_{\cA'}(q,\treelab)$}} &\text{if $\dleft \in \treelab$} \\
\bigvee \set{ (d',r') \wedge (\dup,(\dright,q))  : \text{$q \in Q_{\cA'}$ and $(d,r) \wedge (d',r')$ or $(d',r') \wedge (d,r)$ is a disjunct in $\delta_{\cA'}(q,\treelab)$}} &\text{if $\dright \in \treelab$} \\
\bigvee \set{ (d',r') : \text{$(d,r) \wedge (d',r')$ or $(d',r') \wedge (d,r)$ is in $\delta(q^0_{\cA'},\treelab)$}} &\text{if $\dleft,\dright \notin \treelab$}
\end{cases}
\end{align*}}
In other words, the upwards mode state $(d,r)$ remembers the state $r$ and child $d$ that the automaton came from.
Eve guesses some state $q$ and some disjunct $(d,r) \wedge (d',r')$ or $(d',r') \wedge (d,r)$  in the transition function
$\delta_{\cA'}$ in state $q$.
That is, she is guessing a possible state in the current node that could have led to her being in state $r$ in the $d$-child of the current node.
Adam can either challenge her on the downwards run from here by moving to the $d'$ child and switching to downwards mode state $r'$,
or continue to challenge her upwards by moving up and switching to state $(d'',q)$, where $d''$ records whether the current node
is the left or right child of its parent.
In the special case that the current node is the root, then $q$ must be $q^0_{\cA'}$ and Adam can only challenge downwards.
Note that Adam is not allowed to challenge in direction $d$, since this is where the automaton came from.

In downwards mode state $r$, we set
$\delta(r,\treelab) \colonequals \delta_{\cA'}(r,\treelab)$.
In states like this, the automaton is simulating exactly the original automaton $\cA'$.

The priority assignment is inherited from $\cA'$, namely, we set
$\Omega(r) \colonequals \Omega_{\cA'}(r)$, we set $\Omega((d,r)) \colonequals \Omega_{\cA'}(r)$, and
  we set $\Omega(q_0) \colonequals 1$.

This concludes the construction of the automaton $\cA$ from $\cA'$.
We must show that when $\cA$ is launched from node $v$ in some tree $\tree$,
it behaves like $\cA'$ launched from the root of $\tree'$,
where $\tree'$ is obtained from $\tree$ by adding $\set{P_1,\dots,P_j}$ to the label at $v$. 

\newcommand{\run}{\rho}

A winning strategy $\run'$ of a nondeterministic automaton like $\cA'$ on $\tree'$
can be viewed as an annotation of $\tree'$ with states such that
\begin{inparaenum}[(i)]
\item the root is annotated with the initial state $q^0_{\cA'}$,
\item if a node $v$ with label $\treelab$ is annotated with $q$ and its left child is annotated with $r$ and right child with $r'$,
then $(\dleft,r) \wedge (\dright,r')$ or $(\dright,r') \wedge (\dleft,r)$ is a disjunct in $\delta_{\cA'}(q,\treelab)$,
and
\item the priorities of the states along every branch in $\tree'$ satisfy the parity condition.
\end{inparaenum}

So assume that there is a winning strategy $\run'$ of $\cA'$ on $\tree'$.
It is not hard to see that this induces a winning strategy $\run$ for Eve
in the acceptance game of $\cA$ on $\tree$ starting from $v$:
Eve guesses in a backwards fashion the part of the run $\run'$
on the path from $v$ to the root,
and then processes the rest of the tree in a normal downwards fashion,
using $\run'$ to drive her choices.
Using this strategy,
any play is infinite
and a suffix of this play (namely, once the play has switched to downward mode)
corresponds directly to a suffix of a play in $\run'$.
Since the priorities in these suffixes are identical,
the parity condition must be satisfied,
so $\run$ is winning, and $\cA$ accepts $\tree$ starting from $v$.

Now suppose Eve has a winning strategy $\run$ in the acceptance game
of $\cA$ on $\tree$ starting from~$v$.
Using $\run$,
we stitch together a winning strategy $\run'$ of $\cA'$ on $\tree'$.
Recall that such a winning strategy can be viewed as an annotation of $\tree'$ with
states consistent with $\delta_{\cA'}$ and satisfying the parity condition on every branch.
We construct this annotation starting at $v$,
based on Eve's guess of the state when in $q_0$.
The annotation of the subtree rooted at $v$
is then induced by the plays in $\run$ that
switch immediately to downward mode at $v$.
We can then proceed to annotate the parent $u$ of $v$, by
considering Eve's choice of state when Adam stays in upward mode and moves
to the parent $u$ of $v$.
If $v$ is the $d$-child of $u$, then
the subtree in direction $d$ from $u$ is already annotated;
the subtree in the other direction can be annotated by considering the plays
when Adam switches to downward mode at $u$.
Continuing in this fashion,
we obtain an annotation of the entire tree with states such that
$q^0_{\cA'}$ is the annotation at the root
and the other annotations are consistent with~$\delta_{\cA'}$
on $\tree'$
(the annotations are consistent with~$\tree'$ and not~$\tree$,
because in the initial state, Eve's choice of state and disjunct
is under the assumption of the extra relations $\set{P_1,\dots,P_j}$ present at $v$).
Every branch in this run tree satisfies the parity condition,
since a suffix of the branch
corresponds to a suffix of a play in $\run$ that satisfies the parity condition.
Hence, $\cA'$ accepts $\tree'$ from the root.

\subsection{Proof of Lemma~\ref{lemma:localsubstitution}: Localized Automata for $\agnf$}\label{app:automataproof}
Recall the statement:

\quoteresult{Lemma~\ref{lemma:localsubstitution}}{\localsubstitution}

Let $\psi := \eta[Y_1 := \chi_1,\dots,Y_s := \chi_s]$.
To construct $\cA_\psi$,
take the disjoint union of~$\cA_\eta$, $\cA_{\chi_1}, \dots, \cA_{\chi_s}$.
Then for each polarity $p$ and localization $\vec{a}/\vec{x}$, set the designated initial state to
the initial state for $p$ and $\vec{a}/\vec{x}$ coming from $\cA_\eta$.
Modify the transition function of $\cA_{\psi}^{p,\vec{a}/\vec{x}}$ so that
the automaton starts by simulating $\cA_{\eta}^{p,\vec{a}/\vec{x}}$, but at every node $w$:
\begin{itemize}
\item Eve guesses a valuation for each $Y_i$ at $w$, which is a (possibly empty) set of facts of the form $Y_{i,\vec{b}}$ where $\vec{b}$ is a subset of $\bagnames{w}$.
\item Adam can either accept Eve's guesses and continue the simulation of $\cA_\eta^{p,\vec{a}/\vec{x}}$,
or can challenge one of Eve's assertions of~$Y_i$ by launching the appropriate localized version of~$\cA_{\chi_i}$.
That is, if Eve guesses that $Y_{i,\vec{b}}$ holds at~$w$,
then Adam could challenge this by launching $\cA_{\chi_i}^{+,\vec{b}/\vec{z}}$
starting from~$w$.
Likewise, if Eve guesses that $Y_{i,\vec{b}}$ does not hold at $w$,
then Adam could challenge this by launching $\cA_{\chi_i}^{-,\vec{b}/\vec{z}}$.
\end{itemize}
The correctness of this construction relies on the fact that
each $Y_i$ is being replaced by a base-guarded formula $\chi_i$,
so any $Y_i$-fact must be about a $\sigmab$-guarded set of elements. 
  In particular, remember that if this set of elements has cardinality~$\leq 1$ then the
  $\sigmab$-guard may be an equality atom.
In any case, these elements must appear together in some node of the tree,
so Eve can guess an annotation of the tree that indicates where these $Y_i$-facts appear.

The proof of correctness follows. We present only the case for $p = +$
and signature $\sigma \cup \set{ Y }$, but the case for $p=-$ or multiple $Y_i$ relations is similar.
\newcommand{\strategy}{\zeta}
\newcommand{\strategya}{\zeta_A}

First, assume that $\mydecode{\tree}, [v,\vec{a}]$ satisfies $\psi$, for $\tree$ a $\sigcode{\sigma}{k,l}$-tree.
We must show that Eve has a winning strategy $\strategy$ in the acceptance game of $\cA^{+,\vec{a}/\vec{x}}_\psi$ on~$\tree$
starting from $v$.

For each localization $\vec{b}/\vec{z}$,
let $J_{\vec{b}} := \set{ w \in \mydom{\tree} : \mydecode{\tree}, [w,\vec{b}] \models \chi(\vec{z}) }$.
For each $w \in J_{\vec{b}}$, let $\strategy_{w,\vec{b}}$ denote Eve's winning strategy in the acceptance game
of $\cA^{+,\vec{b}/\vec{z}}_\chi$ on $\tree$ starting from~$w$.
Likewise, for each $w \notin J_{\vec{b}}$, let $\strategy_{w,\vec{b}}$ denote Eve's winning strategy in the acceptance game
of $\cA^{-,\vec{b}/\vec{z}}_{\chi}$ on $\tree$ starting from~$w$.
Let $\tree'$ be $\tree$ extended with
$Y_{\vec{b}}(w)$ for each $\vec{b}$ and $w$ such that $w \in J_{\vec{b}}$.
Then $\mydecode{\tree'},[v,\vec{a}]$ satisfies $\eta$,
so Eve has a winning strategy $\strategy'$ in the acceptance game of $\cA^{+,\vec{a}/\vec{x}}_\eta$ on $\tree'$
starting from $v$.

We use these strategies to define Eve's strategy $\strategy$ in the acceptance game of $\cA^{+,\vec{a}/\vec{x}}_\psi$ on $\tree$
starting from $v$:
at each node $w$, Eve should guess the set $\set{ Y_{\vec{b}} : w \in J_{\vec{b}}}$ and
use the strategy $\strategy'$ (based on the label at $w$, extended with $\set{ Y_{\vec{b}} : w \in J_{\vec{b}}}$).
If Adam never challenges her on these guesses, then the play will correspond to a play in $\strategy'$,
so Eve will win.
If Adam challenges her on some guess that $Y_{\vec{b}}$ holds at~$w$
(respectively, $Y_{\vec{b}}$ does not hold at~$w$),
then the automaton $\cA_\chi^{+,\vec{b}/\vec{z}}$ (respectively, $\cA_{\chi}^{-,\vec{b}/\vec{z}}$) is launched
and Eve should switch to the strategy $\strategy_{w,\vec{b}}$
(this is well-defined since Eve is guessing the set based on $J_{\vec{b}}$).
But $\strategy_{w,\vec{b}}$ is a winning strategy,
so once we switch to this strategy,
Eve is guaranteed to win.
Hence, $\strategy$ is a winning strategy for Eve, as desired.

Now we must prove that if Eve has a winning strategy in the acceptance game of $\cA^{+,\vec{a}/\vec{x}}_\psi$ on~$\tree$
when launched from $v$,
then $\mydecode{\tree}, [v,\vec{a}] \models \psi$.
We prove the contrapositive. Suppose $\mydecode{\tree}, [v,\vec{a}] \not\models \psi$.
We must give a winning strategy $\strategya$ for Adam in the acceptance game of
$\cA^{+,\vec{a}/\vec{x}}_\psi$ on~$\tree$
starting from $v$.
Let $J_{\vec{b}} := \set{ w \in \mydom{\tree} : \mydecode{\tree}, [w,\vec{b}] \models \chi(\vec{z}) }$,
and let $\tree'$ be $\tree$ extended with the valuation for $Y$ such that $Y_{\vec{b}}$ holds at $w$
iff $w \in J_{\vec{b}}$.

Because $\mydecode{\tree}, [v,\vec{a}] \not\models \psi$,
it must be the case that $\mydecode{\tree'}, [v,\vec{a}] \not\models \eta$.
Hence, there is a winning strategy $\strategya'$ for Adam in the acceptance game
of $\cA^{+,\vec{a}/\vec{x}}_\eta$ on~$\tree'$.
By definition of $J_{\vec{b}}$, for each pair $[w,\vec{b}]$ such that $w \notin J_{\vec{b}}$ (respectively, $w \in J_{\vec{b}}$)
it must be the case that 
\mbox{$\mydecode{\tree}, [w,\vec{b}] \not\models \chi$}
(respectively, $\mydecode{\tree}, [w,\vec{b}] \models \chi$),
and hence Adam has a winning strategy $\strategya^{w,\vec{b}}$ in
$\cA_\chi^{+,\vec{b}/\vec{z}}$ (respectively, $\cA_{\chi}^{-,\vec{b}/\vec{z}}$)
on~$\tree$ starting
from $w$.

We define $\strategya$ based on these substrategies.
Adam starts by playing according to $\strategya'$,
and continues using this strategy while the automaton is simulating $\cA^{+,\vec{a}/\vec{x}}_{\eta}$
and Eve is only guessing $Y_{\vec{b}}$ at a node $w$ for $w \in J_{\vec{b}}$.
If Eve ever deviates from this valuation based on $J_{\vec{b}}$,
then Adam challenges this guess and switches to using the appropriate strategy $\strategya^{w,\vec{b}}$.
In either case,
it is clear that the resulting plays will be winning for Adam,
so $\strategya$ is a winning strategy for Adam as desired.

\section{Data Complexity Upper Bounds for Transitivity}\label{app:ptimetransdataupper}

We now provide details on the proof of the data complexity upper bounds.

\subsection{Details on Quantifier-Rank and Pebble Games} \label{apx:pebble}
The \emph{quantifier-rank} of
a first-order formula $\phi$, written $\qr(\phi)$
is the number of nested quantifications: that is, a formula
with no quantifiers has quantifier-rank $0$,   while the inductive definition is:
\begin{align*}
\qr(\neg \phi)= \qr(\phi) \\
\qr(\phi_1 \wedge \phi_2)=\qr(\phi_1 \vee \phi_2)=\max(\qr(\phi_1), \qr(\phi_2)) \\
\qr(\exists x ~ \phi)=\qr(\forall x ~ \phi)=\qr(\phi)+1
\end{align*}
We will be interested in showing that
two sets of facts $I$ and $I'$ agree on all
formulas of quantifier-rank $j$.
This  can be demonstrated using the \emph{$j$-round pebble game on $I,I'$}.
A position in this game is given by a sequence $\vec p$ of elements from $\elems{I}$
and a sequence $\vec p\,'$ of the same length as $\vec p$ from $\elems{I'}$. 
There are two players, Spoiler and Duplicator, and a round of the game
at position~$(\vec p, \vec p\,')$ proceeds by Spoiler choosing one of the sets of facts
(e.g.~$I$)
and appending an element from that set of facts to the corresponding
sequence (e.g.~appending an element from $\elems{I}$ to $\vec p$) while
Duplicator responds by  appending an element from the other set of facts in the
other sequence (e.g.~appending an element from $\elems{I'}$ to $\vec p\,'$).
A $j$-round play of the game is a sequence of $j$ moves
as above.
Duplicator wins the game if the sequences represent a  partial isomorphism:
$p_i=p_j$ if and only if $p'_i=p'_j$ and
for any relation $R$,  $R(p_{m_1} \ldots p_{m_j}) \in I$ if
and only if $F(p'_{m_1} \ldots p'_{m_j}) \in I'$. 
A strategy for Duplicator is a response to each move of Spoiler.
Such a strategy is \emph{winning} from a given
position $\vec p, \vec p\,'$ if every $j$-round play emerging from following the strategy,
starting at these positions, is not winning for Spoiler.
The following result is well-known (see, e.g.~\citeR{fmt}):

\begin{proposition} If there a winning strategy for Duplicator in the $j$-round pebble
game on $I,I'$ starting at $\vec p, \vec p\,'$, then
for every formula $\phi$ of quantifier-rank  at most $j$ satisfied
by $\vec p$ in $I$, $\phi$ is also satisfied by $\vec p\,'$ in $I'$.
\end{proposition}

\subsection{Proof of Theorem~\ref{thm:conptransdataupper}: $\conp$ Data Complexity Bound for $\owqatc$}
\label{apx:conptransdataupperproof}
\newcommand{\veccprime}{\vec{c} \, '}
Recall the statement of Theorem \ref{thm:conptransdataupper}:

\begin{quote}
  \conptransdataupper
\end{quote}

We now prove the theorem. Within this appendix, contrary to the rest of the paper,
we will consider logics that feature constants, in which case we will explicitly
indicate it.

For a set of facts $\instance_0$ over any signature, an \emph{$\instance_0,k$-rooted set}
consists of $\instance_0$ unioned with some sets of facts $T_{\vec c}$ for
$\vec c \in \dom{\instance_0}^k$ where the domain of $T_{\vec c}$ overlaps with
the domain of $\instance_0$ only in~$\vec c$,
and for two $k$-tuples $\vec c$ and $\veccprime$,  the domain of $T_{\vec c}$ overlaps with the domain of $T_{\veccprime}$ only
within $\vec c \cap \veccprime$. 

One can picture such a set as a squid, with $\instance_0$ at the root
and the $T_{\veccprime}$ hanging off as tentacles.
Using Proposition~\ref{prop:transdecomp}, we can show
that if there is a witness to satisfiability of an $\agnf$  sentence
in a superset of a set of facts $\instance_0$, then there is such
an extension that forms an $\instance_0,k$-rooted set.

\begin{proposition} \label{prop:treelike}
For any set of $\sigmab$-facts $\instance_0$, if a $\agnf$ sentence $\Sigma$
  over~$\sigma$ (without constants) is satisfiable by some set of facts containing $\instance_0$ with
relations $R_i^\trans$ interpreted as the transitive closure of~$R_i$,
then $\Sigma$ is satisfied (with the same restriction) in a set of $\sigmab$-facts
which form an $\instance_0,k$-rooted set, where $k$ is at most $\mysize{\Sigma}$.
\end{proposition}

\begin{proof}
By Proposition~\ref{prop:transdecomp}
there is a set of facts $\instance$ consisting of $\sigmab$-facts whose extension
to distinguished facts satisfies $\Sigma$, and which has an
 $\instance_0$-rooted tree decomposition. 
For a child node $v$
of the root of the decomposition of $\instance$, its \emph{interface elements} are the values
that appear in a fact associated with the root and also in a fact associated
with $v$. For each  $\vec c \in \dom{\instance}^k$,
let $C_{\vec c}$ be the children of the root node of the decomposition  whose interface
elements are contained in $\vec c$. 
 Let $T_{\vec c}$ be all $\sigmab$-facts outside of $\instance_0$ that are in
 a descendant of a node in $C_{\vec c}$.
It is easy to verify that $\instance$ and $T_{\vec c}: \vec c \in  \dom{\instance_0}^k$
have the required property.
\end{proof}

Proposition~\ref{prop:treelike} shows that it suffices to examine
witnesses consisting of a rooted set of facts, and
a collection of tentacles indexed by $k$-tuples of the domain
of the root.  We have control over the size of the root, and also
over the index set.
But the size of the tree-like tentacles is unbounded.
We now show a decomposition result stating
that to know what happens in a tree-like set,
 we will not need to care about the details
of the tentacles, but only  a small amount of information concerning
the sentences that the tentacle satisfies in isolation.

Let $\fo(\sigma)$ denote first-order logic over the signature $\sigma$ with equality.
Let $\fo(\sigma \cup \{d_1 \ldots d_k\})$ denote first-order logic over the
signature $\sigma$ with equality and
 $k$ 
 constants, which will be used to represent the overlap elements.
Note that formulas in both $\fo(\sigma)$ and $\fo(\sigma \cup \{d_1 \ldots d_k\})$
can use the distinguished relations $R^\trans_i$ that are part of $\sigma$.

The quantifier-rank of a formula is the maximal number of nested
quantifiers; the formal definition is reviewed in Appendix~\ref{apx:pebble}.
For any fixed signature $\rho$, if we fix the
 quantifier-rank $j$,  
we also fix the number of variables
that may occur in a formula, and thus there are only finitely
many sentences up to logical equivalence. Thus 
we can let $\fo_j(\rho)$ denote a finite set containing a sentence equivalent
to each sentence of quantifier-rank at most $j$.
Given an $\instance_0,k$-rooted set
$\calI$, and  number $j$, the \emph{$j$-abstraction of $\calI$}
is the expansion of $\instance_0$ with relations $P_\tau(x_1 \ldots x_k)$ for 
each $\tau \in \fo_j(\sigma \cup \{d_1 \ldots d_k\})$.
We interpret $P_\tau(x_1 \ldots x_k)$ by the set of $k$-tuples   $\vec c$ such that $T^+_{\vec c}$
satisfies $\tau$, where $T^+_{\vec c}$ interprets the constants in $\tau$ by $\vec{c}$ and the
distinguished relations by the appropriate transitive closures.
We let $\sigma_{j,k}$ be the signature of the $j$-abstraction of such structures.

\begin{lemma} \label{lem:decomp}
For any sentence $\phi$ of $\fo(\sigma)$
and any $k\in\NN$, there is a
$j\in\NN$ having the following property:

\medskip

Let $\calI_1$ be an $\instance_1,k$-rooted set for some set of $\sigmab$-facts~$\fA_1$,
and let $\calI^+_1$ be its extension with distinguished relations $R^\trans_i$ 
interpreted as the transitive closure of the corresponding base relations~$R_i$.
Let $\calI_2$ be an $\instance_2,k$-rooted set for some set of $\sigmab$-facts~$\fA_2$,
and $\calI^+_2$ the corresponding extension with facts over the distinguished relations.
If the $j$-abstractions of $\calI_1$ and $\calI_2$ agree on all $\fo(\sigma_{j,k})$ sentences of quantifier-rank
at most $j$, then
 $\calI^+_1$ and $\calI^+_2$ agree on~$\phi$.
\end{lemma}
\begin{proof}
Let $j_\phi$ be the quantifier-rank of $\phi$. We choose $j \defeq j_\phi\nolinebreak\cdot\nolinebreak k$. We will show that $\calI_1$ and $\calI_2$ agree on
all formulas  of quantifier-rank $j_\phi$. By standard results in finite
model theory (reviewed in Appendix~\ref{apx:pebble}), it is sufficient
to give a strategy for Duplicator in the $j_\phi$-round standard pebble game for $\fo(\sigma)$ over
$\calI^+_1$ and $\calI^+_2$. 
With $i$ moves left to play, we will ensure the following invariants on a game position  consisting of a sequence $\vec p_1 \in \calI^+_1$ and $\vec p_2 \in \calI^+_2$:

\begin{itemize}
  \item Let $\vec{p_1}'$ be the subsequence of $\vec p_1$ that comes from
    $\instance_1$ and let $\vec{p_2}'$ be defined similarly
    for $\vec p_2$ and $\instance_2$. Then $\vec{p_1}'$ and $\vec {p_2}'$ should 
form a  winning position for Duplicator in the
    $(i \cdot k)$-round $\fo(\sigma_{j,k})$ game on the $j$-abstractions.

\item Fix any $k$-tuple $\vec c_1 \in \instance_1$ and let $P^1_{\vec c_1}$ be the subsequence of $\vec p_1$
that lies in $T_{\vec c_1}$ within $\calI^+_1$. Then if $P^1_{\vec c}$ is non-empty, $\vec c_1$ also lies in $\vec p_1$.
Further, letting $\vec c_2$ be
the corresponding $k$-tuple to $\vec c_1$ in~$\vec p_2$, and letting $P^2_{\vec c_2}$ be the subsequence of $\vec p_2$
that lies in $T_{\vec c_2}$ within $\calI^+_2$, then  $P^1_{\vec c_1}$ and $P^2_{\vec c_2}$ form a winning position
in the $i$-round pebble game on  $T^+_{\vec c_1}$ and $T^+_{\vec c_2}$.

The analogous property holds for any $k$-tuple $\vec c_2 \in \instance_2$.
\end{itemize}
We now explain the strategy of Duplicator, focusing for simplicity on moves of Spoiler within~$\calI_1$, with the strategy
on~$\calI_2$ being similar. If Spoiler plays within $\instance_1$, Duplicator responds using
her strategy for the games on the $j$-abstractions of $\instance_1$ and $\instance_2$. It is easy to see that the invariant is preserved.

If Spoiler plays an element  within a substructure $T_{\vec c_1}$  within~$\calI_1$ that is already inhabited, then by the inductive invariant,
$\vec c_1$ is pebbled and there is a corresponding $\vec c_2$ in~$\calI_2$ with substructure $T_{\vec c_2}$ of~$\calI_2$ such that
the pebbles within $T_{\vec c_2}$  are winning positions in the game on~$T^+_{\vec c_1}$ and $T^+_{\vec c_2}$ with $i$ moves left to play. Thus Duplicator can respond
using the strategy in this game from those positions.

Now suppose Spoiler plays an element  $e_1$ within a substructure $T_{\vec c_1}$  within $\calI_1$ that is not already inhabited.
We first use  $\vec c_1$ as  a sequence of plays for Spoiler in the game on the $j$-abstractions of~$\calI_1$ and $\calI_2$, extending the 
positions given by
$\vec{p_1}'$ and $\vec {p_2}'$.
By the inductive invariant, responses of Duplicator exist,
and we collect them  to get a tuple $\vec c_2$. 
Since a winning strategy in a game preserves atoms, and
we have a fact in the $j$-abstraction corresponding to the $j$-type of 
$\vec c_1$ in $T^+_{\vec c_1}$, we know that 
$\vec c_2$ must satisfy the same $j$-type in~$T^+_{\vec c_2}$
  that $\vec c_1$ does in $T^+_{\vec c_1}$.
Therefore $\vec c_1$ must satisfy the same $\fo(\sigma \cup \{d_1 \ldots d_k\})$ sentences of quantifier-rank at most 
$j$ in $T^+_{\vec c_1}$ as $\vec c_2$ does in $T^+_{\vec c_2}$.
Thus Duplicator can use the corresponding
strategy to respond to $e_1$ with an $e_2$ in $T^+_{\vec c_2}$
  such that $\{e_1\}$ and $\{e_2\}$ are a winning position in the $(i-1)$-round
pebble game on $T^+_{\vec c_1}$ and $T^+_{\vec c_2}$.

Since the
response of Duplicator corresponds to $k$ moves in the  game within
the $j$-abstractions, one can verify that the invariant is preserved.

We must verify that this strategy gives a partial isomorphism. Consider a fact $F$ that holds of a tuple 
$\vec t_1$ within $\calI_1$, and let $\vec t_2$ be  the  tuple obtained  using this strategy in $\calI_2$.
We first consider the case where $F$ is a $\sigmab$-fact:
\begin{itemize}
\item If $\vec t_1$ lies completely within some~$T_{\vec c_1}$, then the last invariant guarantees that $\vec t_2$
lies in some~$T_{\vec c_2}$. The last invariant also guarantees that $\sigmab$-facts of $\calI_1$ are preserved since such facts
must lie in $T_{\vec c_1}$, and the corresponding positions are winning 
in the game between   $T^+_{\vec c_1}$  and $T^+_{\vec c_2}$.

\item If $\vec t_1$  lies completely within $\instance_1$, then
the first invariant guarantees that the fact is preserved.
\end{itemize}
By the definition of a rooted set, the above two cases are exhaustive.
We now consider the case where $F$ is  of the form
 $R^\trans_i(t_1, t_2)$:
\begin{itemize}
\item If $t_1$ and $t_2$ both lie in some~$T_{\vec c_1}$, then
we reason as in the first case above, since facts over the signature
with transitive closures are also preserved in the game between
$T^+_{\vec c_1}$ and $T^+_{\vec c_2}$.
\item If $t_1$ and $t_2$ are both in $\instance_1$, we reason
as in the second case above, this time using the fact
that transitive closure facts  are taken into  account in the game on the abstraction.
\item If  $t_1$ lies in $T_{\vec c_1}$, $t_2$ lies in $T_{\vec c_2}$,
then $t_1$ reaches some $c_i \in \vec c_1$,
$c_i$ reaches some  $c_j \in \vec c_2$, and $c_j$ reaches
$t_2$ within $T_{\vec c_2}$. Then we use a combination of the first two cases
above to  conclude that $F$ is preserved.\qedhere
\end{itemize}
\end{proof}

From Lemma \ref{lem:decomp} we easily obtain:

\begin{corollary} \label{cor:composition} For any sentence $\phi$ and $k\in\NN$, there is
  $j\in\NN$
  and a sentence $\phi'$ in the language $\sigma_{j,k}$ of \mbox{$j$-abstractions} over $\sigma$ such that
for all sets of $\sigmab$-facts $\instance_0$, an   $\instance_0,k$-rooted set
  satisfies $\phi$ iff its \mbox{$j$-abstraction} satisfies~$\phi'$.
\end{corollary}

Recall that by Proposition \ref{prop:treelike}, we know it
suffices to check for a counterexample to entailment that
is an $\instance_0,k$-rooted set.
Corollary \ref{cor:composition} allows us to do this by guessing an abstraction
and checking a first-order property of it.
This allows us to finish the proof of Theorem  \ref{thm:conptransdataupper}.

\begin{proof}[Proof of Theorem~\ref{thm:conptransdataupper}]
First, consider the case where the initial set of facts $\instance_0$ is restricted to contain
only $\sigmab$-facts.
Fixing $Q$ and $\Sigma$, we give an $\np$ algorithm for the complement. Let $\phi=\Sigma \wedge \neg Q$, and
$k = \mysize{\varphi}$.
Let $j$ and $\phi'$ be the number and formula guaranteed for $\phi$ by Corollary \ref{cor:composition}.

Recall that $\fo(\sigma \cup \{d_1 \ldots d_k\})$ denotes first-order logic over
  the signature $\sigma$ of $\Sigma \wedge \neg Q$ with equality and with 
$k$ 
  constants, and
$\fo_j(\sigma \cup \{d_1 \ldots d_k\})$ denotes a finite set containing a
  sentence equivalent to each sentence of quantifier-rank at most $j$.
Let $\types_j$ be the collection of subsets $\tau$ of $\fo_j(\sigma \cup \{d_1 \ldots d_k\})$ sentences such
that the conjunction  of sentences in $\tau$ is
satisfiable.
Note that for any fixed $j$, the size of $\fo_j(\sigma \cup \{d_1 \ldots d_k\})$ is finite,
hence the size of $\types_j$ is finite.
An element of $\types_j$ can be thought of as a description of 
a tentacle, telling us everything we need to know  for the purposes
of the $j$-abstraction.  

Given $\instance_0$,  
guess a function~$f$ mapping each  $k$-tuple over $\instance_0$ to a  $\rho \in \types_j$.
We then check whether for two overlapping $k$-tuples $\vec c$ and $\veccprime$,
the types
$f(\vec c)$ and $f(\veccprime)$ are consistent on the atomic formulas
that hold on overlapping elements, and whether the atomic
formulas of $f(\vec c)$ contain each fact over $\vec c$ in $\instance_0$.
Finally, for each $\tau \in \fo(\sigma \cup \{d_1 \ldots d_k\})$ of quantifier-rank at most $j$, we form the expansion $\instance$ by
interpreting $P_\tau$ by the set of tuples   $\vec c$ such
that $\tau \in f(\vec c)$, and we 
check whether the expansion satisfies $\phi'$ with these interpretations,
and if so return true.

We argue for correctness. If the algorithm returns true with $\instance$ as the witness, then create an
  \mbox{$\instance_0,k$-rooted} set $\calI$ by picking for each $\vec c$ a structure satisfying
the sentences in $f(\vec c)$ with distinguished
elements interpreted by $\vec c$. Such a structure exists by satisfiability of $f(\vec c)$.
It assigns atomic formulas consistently on overlapping tuples, by hypothesis, and
the atomic formulas it assigns contain each fact of $\instance_0$.
We let the remaining domain elements be disjoint from the domain
of $\instance$. 
Note that by construction, $\calI$ has~$\instance$ as its $j$-abstraction.
By the choice of $j$ and $\phi'$, and the observation above, $\calI$ satisfies $\Sigma \wedge \neg Q$.
Thus   $\calI$
witnesses that $\owqatc(\instance_0, \Sigma, Q)$  is false.

On the other hand, if $\owqatc(\instance_0, \Sigma, Q)$  is false, then by Proposition \ref{prop:treelike}
we have an $\instance_0,k$-rooted set $\calI$
that satisfies $\Sigma \wedge \neg Q$. By the choice of $j$ and $\phi'$,  the $j$-abstraction of $\calI$ satisfies
$\phi'$. For each tuple $\vec c$  from $\instance_0$,
the set of formulas of quantifier-rank holding of $\vec c$ in the tentacle
of $\vec c$ must be in $\types_j$. Hence we can guess $f$
that assigns $\vec c$ to this set, and with this $f$ as a witness
the algorithm returns true.

When the initial set of facts contains also $\sigmad$-facts, then we need to
ensure that our algorithm guarantees the existence of a  witness structure which fulfills
these transitivity requirements.  We first pre-process
the sentence as follows: for each $\sigmab$ relation $R$ 
we add a new $\sigmab$ relation $R'$, and add to our theory $\Sigma$
the sentences:
\begin{align*}
\forall \vec x ~ R'(\vec x) \rightarrow R^+(\vec x)
\end{align*}
Note that these sentences are base guarded.
Now given an instance $\instance_0$ containing facts
$R^+$, we change $R^+$ to $R'$ and perform the
algorithm as before.

\myeat{
This can be achieved by adding additional requirements
to the types selected by the guessed function $f$; we must check that the guess
for a tuple $\vec c$ includes all of the transitivity facts that are in the initial set.}
\end{proof}

\subsection{Proof of Theorem~\ref{thm:ptimetransdataupper}: $\ptime$ Data Complexity Bound for $\owqatrans$}
\label{apx:ptimetransdataupperproof}
We now turn to the case where our constraints are restricted to $\acfgtgd$s and deal with
$\owqatrans$, not $\owqatc$. Recall that Theorem \ref{thm:ptimetransdataupper} states
a $\ptime$ data complexity bound for this case:

\quoteresult{Theorem~\ref{thm:ptimetransdataupper}}{\ptimetransdataupper}

The proof will follow from a reduction to traditional $\owqa$,
similar to the proof of Proposition~\ref{prop:rewritelin}:

\begin{proposition}
  \label{prop:reducetr}
  For any finite set of facts $\instance_0$,
  constraints $\Sigma \in \acgnf$,
  and
  base-covered UCQ $Q$,
  we
  can compute $\instance_0'$ and $\Sigma' \in \gnf$ in $\ptime$
  such that
  $\owqatr(\instance_0, \Sigma, Q) \text{ iff } 
  \owqa(\instance_0', \Sigma', Q)$.
  Furthermore,
  if $\Sigma$ is in $\acfgtgd$ then $\Sigma'$ is in $\fgtgd$.
\end{proposition}

\begin{proof}
\renewcommand{\drel}{R^+}
We define $\instance_0'$ and $\Sigma'$ as follows:
  \begin{itemize}
  \item $\instance_0'$ is $\instance_0$ together with facts $G(a,b)$ for every pair $a,b \in \elems{\instance_0}$
  for some fresh binary base relation $G$, and
  \item $\Sigma'$ is $\Sigma$ together with the $k$-guardedly-transitive axioms
  for each distinguished relation, where $k$ is  $\mysize{\Sigma \wedge \neg Q}$.
  \end{itemize}
These can be constructed in time polynomial in the size of the input.

As discussed in the proof of Lemma~\ref{lemma:guarded-interface-dec},
  the $k$-guardedly transitive axioms (see Section~\ref{sec:decidlin})
can be written in $\nf$ $\agnf$ with width at most $k$,
and hence in $\gnf$.

Now we prove the correctness of the reduction.
Suppose $\owqa(\instance_0',\Sigma',Q)$ holds,
so any $\instance' \supseteq \instance_0'$ satisfying $\Sigma'$ must satisfy $Q$.
Now consider $\instance \supseteq \instance_0$ that satisfies $\Sigma$
and where all $\drel$ in $\sigmad$
are transitive.
We must show that $\instance$ satisfies $Q$.
First, observe that $\instance$ satisfies $\Sigma'$
since the $k$-guardedly-transitive axioms for $\drel$
are clearly satisfied for all $k$
when $\drel$ is transitively closed.
Now consider the extension of $\instance$ to $\instance'$
with additional facts $G(a,b)$ for all $a,b \in \elems{\instance_0}$.
This must still satisfy $\Sigma'$:
adding these guards
means there are additional $k$-guardedly-transitive requirements
on the elements from $\instance_0$,
but these requirements already hold
since $\drel$ is transitively closed on all elements.
Hence, by our initial assumption, $\instance'$ must satisfy $Q$.
Since $Q$ does not mention $G$,
the restriction of $\instance'$ back to $\instance$ still satisfies $Q$ as well.
Therefore, $\owqa(\instance_0,\Sigma,Q)$ holds.

On the other hand, suppose for the sake of contradiction
that $\owqa(\instance_0',\Sigma',Q)$ does not hold,
but $\owqatr(\instance_0,\Sigma,Q)$ does.
Then there is some $\instance' \supseteq \instance_0'$ such that
$\instance'$ satisfies $\Sigma' \wedge \neg Q$, and hence also satisfies $\Sigma \wedge \neg Q$.
Since $\Sigma \wedge \neg Q$ is in $\agnf$,
Proposition~\ref{prop:guarded-interface-dec-appendix}
implies that we can take $\instance'$ to be a set of facts that has an
$\instance_0'$-rooted $(k-1)$-width base-guarded-interface tree decomposition.
Let $\instance''$ be the result of taking the transitive closure
of the distinguished relations in $\instance'$.
  By the Transitivity Lemma (Lemma~\ref{lem:guarded-transitivity}),
transitively closing like this can only add $\drel$-facts
about pairs of elements that are not base-guarded.
  Moreover, the Base-Coveredness Lemma (Lemma~\ref{lemma:acov}) ensures that adding
$\drel$-facts about these non-base-guarded pairs of elements
does not affect satisfaction of $\acgnf$ sentences,
so $\instance''$ must still satisfy $\Sigma \wedge \neg Q$.
Restricting $\instance''$ to its $\sigma$-facts results in an~$\instance$
where every distinguished relation is transitively closed
and where $\Sigma \wedge \neg Q$ is still satisfied,
since $\Sigma$ and $Q$ do not mention relation $G$.
But this contradicts the assumption that $\owqatr(\instance_0,\Sigma,Q)$ holds.
This concludes the proof of correctness.

Finally,
observe that the $k$-guardedly-transitive axioms can be written as $\fgtgd$s (in fact, $\afgtgd$s):
they are equivalent to the conjunction of $\fgtgd$s of the form
\begin{align*}
&\forall x\, y\, x_1 \dots x_{l+1}\, \vec z \,\big [ \big ( x = x_1 \wedge x_{l+1} = y \ \wedge \\
&\ \ \drel(x_1,x_2) \wedge \dots \wedge \drel(x_l,x_{l+1}) \wedge S(x,y, \vec z) \big )
 \rightarrow
\drel(x,y) \big ]
\end{align*}
for all $S \in \sigmab \cup \set{ G }$,
$1 \leq l \leq k$, and $\drel \in \sigmad$.
Therefore, if $\Sigma$ is in $\acfgtgd$ then $\Sigma'$ is in $\fgtgd$
as claimed.
\end{proof}

Theorem~\ref{thm:ptimetransdataupper} easily follows from this:

\begin{proof}[Proof of Theorem~\ref{thm:ptimetransdataupper}]
  Recall that we have fixed constraints $\Sigma$ in $\acfgtgd$ and a base-covered UCQ $Q$.
  We must show $\ptime$ data complexity of $\owqatr(\instance_0,\Sigma,Q)$
  for any finite initial set of facts $\instance_0$.
  Use Proposition~\ref{prop:reducetr} to construct $\Sigma'$ from $\Sigma$ (in constant time,
  since $\Sigma$ is fixed) and $\instance_0'$ from $\instance_0$ (in time polynomial in $\mysize{\instance_0}$).
  Since $\Sigma$ is in $\acfgtgd$, $\Sigma'$ is in $\fgtgd$.
  Therefore, the $\ptime$ data complexity upper bound for $\owqatr$ with $\acfgtgd$s follows from
  the $\ptime$ data complexity upper bound for $\owqa$ with $\fgtgd$s \cite{bagetcomplexityfg}.
\end{proof}

\section{Details about the Chase}
\label{apx:chase}

The \emph{chase} is a standard database construction~\cite{ahv},
which applies to a set of facts~$\instance_0$ and to a set $\Sigma$ of TGDs, and
constructs a set of facts $\instance_\infty \supseteq \instance_0$, possibly infinite,
which satisfies $\Sigma$.

To define the chase, we first define the notion of a \emph{trigger} 
and \emph{active trigger}.
A \emph{trigger} for a TGD $\tau: \forall \vec{x} ~ \phi(\vec{x})
\rightarrow \exists \vec{y} ~ \psi(\vec{x}, \vec{y})$ in a set of facts~$\instance$ is a 
homomorphism $h$ from~$\phi(\vec{x})$ to~$\instance$, i.e., a mapping
from~$\vec{x}$ to~$\elems{\instance}$ such that the 
facts of $\phi(h(\vec{x}))$ are in~$\instance$. We call $h$ an \emph{active
trigger} if $h$ cannot be extended to 
a homomorphism from $\psi(\vec{x}, \vec{y})$ to~$\instance$, i.e., there is no
mapping $h'$ from $\vec{x} \cup \vec{y}$ to $\elems{\instance}$ such that $h'(x) =
h(x)$ for all $x \in \vec{x}$ and such that the facts of~$\psi(h(\vec{x}, \vec{y}))$ are
in~$\instance$.

Given a TGD $\tau: \forall \vec{x} ~ \phi(\vec{x})
\rightarrow \exists \vec{y} ~ \psi(\vec{x}, \vec{y})$, a set of facts $\instance$, and an active trigger $h$
of~$\tau$
in~$\instance$, the result of \emph{firing} $\tau$ on~$h$ is a set
of facts $\psi(h(\vec{x}), \vec{b})$ where the $\vec{b}$ are fresh elements
called \emph{nulls} that
are all distinct and do not occur in~$\instance$. The
application of a \emph{chase round} by a set $\Sigma$ of TGDs on a set of facts
$\instance$ is the set of facts $\instance'$ obtained by firing simultaneously all TGDs of~$\Sigma$ on all
active triggers; formally, it is the union of the set of facts~$\instance$ and
of the facts obtained by firing each TGD $\tau$ on each active trigger $h$
of~$\tau$ in~$\instance$, using different nulls when firing each TGD.

The \emph{chase} of~$\instance$ by~$\Sigma$ is the (potentially infinite) set of
facts $\instance_\infty$ obtained by repeated applications of chase rounds.
Formally, we define $\instance_i$ for all $i>0$ as the result of applying a
chase round on~$\instance_{i-1}$, and the chase~$\instance_\infty$ is the
fixpoint of this inflationary operator.

The important points about the chase construction is that the set of facts
$\instance_\infty$ obtained as a result of the chase is a superset of the
initial set of facts $\instance_0$, that it satisfies $\Sigma$, and that it is
created by adding new facts in a way that only overlaps on~$\elems{\instance_0}$
on elements that occur at an active trigger at a position where they will be
exported.

\section{From UCQ to CQ}
\label{app:ucqtocq}
In this section, we first prove general auxiliary lemmas about reducing from $\owqa$ problems
with UCQs to $\owqa$ problems with CQs. We first give such a lemma for regular
$\owqa$, 
which formalises an existing folklore technique (see, e.g., Section~3.3 of
\citeA{georgchristos}). We then
adapt this lemma to the various $\owqa$ notions that we study for distinguished relations.
This allows us to revisit the results of Sections~\ref{sec:hardness}
and~\ref{sec:undecid}
and explain how the proofs in the main text using UCQs
can be extended to use only CQs.

The general idea to replace UCQs by CQs is to extend the arity of the relations
to include a flag that indicates whether a fact is a ``real fact'' or a ``pseudo-fact'':
the flag is propagated by the TGDs. (Note that, while the idea is similar, the
notions of ``real fact'' and ``pseudo-fact'' used in this appendix are not
related to those of ``genuine fact'' and ``pseudo-fact'' that were used within
the proofs of Section~\ref{sec:hardness}.)
We then add pseudo-facts to the instances to
ensure that each UCQ disjunct has a match that involves pseudo-facts. This
ensures that we can replace the UCQ by a \emph{conjunction} of the original
disjuncts, with an OR on the flag of the match of each disjunct: this OR can be
performed using a suitable relation which we add to the instances.

We formalize this general idea in Appendix~\ref{apx:ucqtocqgeneral}. We must then tweak the idea
to make it work for $\owqa$ with distinguished relations, as we
cannot increase the arity of these relations. When the distinguished relations
are linear orders, we first show
in Appendix~\ref{apx:ucqtocqlin} that we can adapt the idea without increasing
the arity of the distinguished relations, provided that the TGDs do not
mention these relations and that the query satisfies a condition called
\emph{base-domain-coveredness}. When the distinguished relations are transitive
or are the transitive closure of another relation, we present in
Appendix~\ref{apx:ucqtocqtc} a more general technique that allows us, under some
conditions on the TGDs and queries, to
increase the arity of a subset of the relations
(called the \emph{flagged relations}), without changing the arity of
the others (which we call the \emph{special relations}, and which include the
distinguished relations). We then use these techniques to
revisit the results of Section~\ref{sec:hardness} in
Appendix~\ref{app:ucqtocqhardness} and of Section~\ref{sec:undecid} in
Appendix~\ref{app:ucqtocqundecid}.

\subsection{UCQ to CQ for General $\owqa$}
\label{apx:ucqtocqgeneral}

We first show a translation result from UCQ to CQ for general $\owqa$:

\begin{lemma}
  \label{lem:ucqtocq}
  For any signature $\sigma$, TGDs~$\Sigma$ and UCQ $Q$, one can compute in
  $\ptime$ a signature $\sigma'$, TGDs~$\Sigma'$ and CQ $Q'$ such that the $\owqa$
  problem for $\Sigma$ and $Q$ reduces in $\ptime$ to the $\owqa$ problem for $\Sigma'$
  and $Q'$ (for combined complexity and for data complexity): namely, given a
  set of facts $\instance_0$ on~$\sigma$, we can compute in $\ptime$ in
  $\instance_0$, $\Sigma$, and $Q$ a set of facts $\instance_0'$ on~$\sigma'$
  such that $\owqa(\instance_0, \Sigma, Q)$ holds iff $\owqa(\instance_0',
  \Sigma', Q')$ holds.
\end{lemma}

\begin{proof}
  Let $\sigma_\Or$ be a constant signature consisting of a ternary relation
  $\mathrm{Or}$ and a unary relation $\mathrm{True}$. We let $\sigma'$ be the
  signature obtained from $\sigma$ by creating one relation $R'$ in~$\sigma'$
  for every $R$ in~$\sigma$
  with $\arity{R'} \colonequals \arity{R}+1$, and further adding the relations
  of~$\sigma_\Or$.

  We define $\Sigma'$ from $\Sigma$ by considering each TGD $\tau : \forall \vec x ~
  \phi(\vec x) \rightarrow \exists \vec y ~ \psi(\vec x, \vec y)$, and, letting
  $z$ be a fresh variable, replacing $\tau$ by the TGD $\tau' : \forall \vec x \, z
  ~ \phi'(\vec x, z) \rightarrow \exists \vec y ~ \psi'(\vec x, z, \vec y)$,
  where $\phi'$ and $\psi'$ are obtained from $\phi$ and $\psi$ respectively by
  replacing each $\sigma$-atom $R(\vec w)$ by the $\sigma'$-atom $R'(\vec w, z)$.
  As TGD bodies are not empty, the new variable $z$
  actually occurs in the new body $\phi'$.

  We now describe the construction of $Q'$ from $Q$. Suppose the UCQ $Q$ is
  $\bigvee_{1 \leq i \leq m} \exists \vec{x_i} ~ Q_i(\vec x_i)$, where each
  $Q_i$ is a conjunction of atoms over $\sigma$.
  Let $z_1, \ldots, z_m$ 
  be fresh variables.
  For each $1 \leq i \leq m$, we define a conjunction of atoms
  $Q_i'(\vec x_i, z_i)$
  on~$\sigma'$ which is obtained from $Q_i(\vec x_i)$ by replacing each
  $\sigma$-atom $R(\vec w)$ by the $\sigma'$-atom $R'(\vec w, z_i)$. We now
  define $Q'$ as:
  \begin{align*}
    \exists z_1 \, \ldots \, z_m \, z'_1 \, \ldots \, z'_{m-1} \,
    \vec x_1 \, \ldots \, \vec x_m ~
    \\
    \Or(z_1, z_2, z'_1) \wedge \Or(z'_1, z_3, z'_2) \wedge \cdots \wedge\,
    \Or(z'_{m-2}, z_m, z'_{m-1}) \wedge \mathrm{True}(z'_{m-1}) \wedge
    \bigwedge_{1 \leq i \leq m} Q'_i(\vec x_i, z_i)
  \end{align*}
  It is clear that the computation of $\sigma'$, $\Sigma'$, and $Q'$ from
  $\sigma$, $\Sigma$, and $Q$ is in $\ptime$.
  
  \medskip
 
  We now describe the $\ptime$
  transformation on input sets of facts. Let $\instance_0$ be a set of facts.
  Letting $\true$ and $\false$ be two fresh elements, let
  $\instance_{\mathrm{Or}}$ be the set of facts that contains the fact
  $\mathrm{True}(\true)$ and the facts $\mathrm{Or}(b, b', b'')$ for all $\{(b,
  b', b \vee b') \mid b, b' \in \{\false, \true\}\}$. Let $\instance_\false$ be the
  set of facts $\{R'(\false, \ldots, \false) \mid R' \in \sigma'\}$,
  and let $(\instance_0)_{+\true}$ be 
  $\{R'(\vec a, \true) \mid R(\vec a) \in \instance_0\}$. We define 
  $\instance_0' \colonequals \instance_{\Or}
    \sqcup \instance_\false \sqcup (\instance_0)_{+\true}$ which is clearly
    computable in $\ptime$.

  \medskip

  We now show correctness of the reduction. In the forward direction, consider a
  counterexample set of facts $\instance$ on~$\sigma$ which is a superset of
  $\instance_0$, satisfies $\Sigma$, and violates $Q$: up to renaming we can
  ensure that $\true, \false \notin \dom{\instance}$. Let us construct a
  counterexample $\instance'$ for $\instance_0'$, $\Sigma'$, and $Q'$,
  by setting $\instance' \colonequals \instance_{\Or} \sqcup \instance_{+\true} \sqcup
  \instance_\false$ where $\instance_{\Or}$ and $\instance_\false$ are as above and where 
  $\instance_{+\true} \colonequals \{R'(\vec a, \true) \mid R(\vec a) \in \instance\}$.

  It is clear that $\instance_\ptrue$, hence $\instance'$, is a superset
  of~$\instance'_0$,
  because
  $\instance$ is a superset of~$\instance_0$.
  To see why
  $\instance'$ satisfies $\Sigma'$, consider a match $M' \subseteq \instance'$ of the body of a TGD $\tau'$ of
  $\Sigma'$ in~$\instance'$. As $\Sigma'$ does not mention the
  relations of~$\sigmaor$, no fact 
  of~$\instance_{\Or}$ can occur in~$M'$. Now, the facts of $\instance_\false$ have $\false$ as their
  last element, and those of $\instance_\ptrue$ have $\true$ as their 
  last element, so, as all atoms of the body of $\tau'$ have the same variable at
  their last element, either $M' \subseteq \instance_\false$, or $M'
  \subseteq \instance_\ptrue$. In the first case, we can find a match of the
  head of~$\tau$ in $\instance_\false$ (where all variables are mapped
  to~$\false$), so we conclude that $M'$ is not a violation.
  In the second case, considering the
  preimage $M$ of~$M'$ in~$\instance$, it is clear that $M$ is a match
  of the TGD $\tau$ of~$\Sigma$, so, as $\instance$ satisfies $\Sigma$, we can
  extend $M$ to a match of the head of~$\tau$ in~$\instance$, yielding a match
  of the head of $\tau'$ in~$\instance'$, so that again $M'$ cannot be a
  violation. Hence, $\instance'$ satisfies $\Sigma$.

  Last, to see why $\instance'$ violates $Q'$, assume by contradiction that
  there is a homomorphism from~$Q'$ to~$\instance'$.
  Notice that, in our construction of~$\instance'$,
  the only element $a \in \dom{\instance'}$ such that $\True(a)$ holds is $a =
  \true$. Hence, necessarily, $h$ must map
  $z'_{m-1}$ to~$\true$. However, 
  as the only $\Or$-facts in~$\instance'$ are those
  of~$\instance_{\Or}$, it is clear that
  $h$ must map some~$z_{i_0}$ to~$\true$. Thus, a suitable restriction of $h$ is a match of 
  $\exists \vec x_{i_0} ~ Q'_{i_0}(\vec x_{i_0}, \true)$ in~$\instance'$.
  Now, as all facts in the image of~$h$ have $\true$ as their last element, the
  image of~$h$ must be contained in $\instance_\ptrue$, so we deduce that
  $\exists \vec x_{i_0} ~ Q_{i_0}(\vec x_{i_0})$ has a match in~$\instance$,
  contradicting the fact that~$\instance$ violates~$Q$.
  This concludes the forward direction of the correctness proof.

  \medskip

  In the backward direction,
  consider a counterexample set of facts $\instance' \supseteq \instance'_0$
  that satisfies $\Sigma'$ and violates $Q'$. Construct the set of facts
  $\instance \colonequals \{R(\vec a) \mid R'(\vec a, \true) \in \instance'\}$.
  As $(\instance_0)_\ptrue \subseteq \instance'_0$,
  clearly $\instance_0 \subseteq \instance$. To
  see why $\instance$ satisfies $\Sigma$, consider a match $M \subseteq
  \instance$ of the body of some TGD $\tau$ of~$\Sigma$ in~$\instance$, and
  consider its preimage $M'$ in~$\instance'$, where all facts have $\true$ as
  their last element: $M'$ is a match of the body of $\tau' \in \Sigma'$.
  Hence, as $\instance'$ satisfies $\Sigma'$, 
  $M'$ extends to a match of the head of~$\tau'$, and the last elements of all
  its facts is~$\true$, so we can find a suitable extension in~$\instance$ as
  well. Hence, $M$ is not a violation of~$\tau$ in~$\instance$, so $\instance$
  satisfies $\Sigma$.

  Last, to see why $\instance$ violates the UCQ~$Q$, assume by contradiction that the
  $Q$ has a match in~$\instance$. This means that there is $1 \leq i_0 \leq
  m$ such that the disjunct $\exists \vec x_{i_0} ~ Q_{i_0}(\vec x_{i_0})$ has a match
  $M_{i_0}$
  in~$\instance$. By construction of~$\instance$, this means that
  $\exists \vec x_{i_0} ~ Q'_{i_0}(\vec x_{i_0}, \true)$ has a match $M'_{i_0}$ in~$\instance'$.
  Now, observe that, for all $1 \leq i \leq m$, there is a match $M''_i$
  of~$\exists \vec x_i ~ Q_i(\vec x_i, \false)$ in~$\instance_\false$ obtained by
  mapping all variables to~$\false$.
  As $\instance_\false \subseteq \instance_0' \subseteq \instance'$,
  the same is true of~$\instance'$. Now, as $\instance_{\Or} \subseteq
  \instance'_0 \subseteq \instance'$, we can extend $M'_{i_0}$ and the $M''_i$
  for $i \neq i_0$ to a match of~$Q'$ in~$\instance'$
  by
  matching $z_{i_0}$ to~$\true$, every $z_i$ to~$\false$ for $i \neq i_0$, every
  $z_i'$ for $i' < i_0 -1$ to~$\false$, and the $z_i'$ for $i' \geq i_0-1$
  to~$\true$.
  This contradicts the fact that
  $\instance'$ violates $Q'$. Hence, $\instance$ violates $Q$ and is a
  counterexample to $\owqa$. This concludes the correctness proof.
\end{proof}

\subsection{UCQ to CQ for $\owqalin$}
\label{apx:ucqtocqlin}

We first adapt the general $\owqa$ result in Lemma~\ref{lem:ucqtocq}
to $\owqalin$, which avoids increasing the arity of the distinguished relations
that are interpreted as linear orders.
In order to avoid increasing the arity of the distinguished relations,
we will ban distinguished
relations in the dependencies, and 
require \emph{base-domain-coveredness} of the query (which is a weakening of
base-coveredness):

\begin{definition}
  \label{def:basedomaincov}
  A CQ $Q$ is \emph{base-domain-covered} if every variable $x$ occurring in a
  distinguished atom in~$Q$ also occurs in a base atom in~$Q$.
\end{definition}

We can now state the following variant of Lemma~\ref{lem:ucqtocq}:

\begin{lemma}
  \label{lem:ucqtocqlin}
  For any signature $\sigma$ (partitioned in base relations $\sigmab$ and
  distinguished relations $\sigmad$),
  TGDs $\Sigma$ on $\sigmab$, and base-domain-covered UCQ $Q$, one
  can compute in $\ptime$ a signature $\sigma'$ partitioned as $\sigmab' \sqcup
  \sigmad$,
  TGDs $\Sigma'$ on $\sigmab'$, and a base-domain-covered CQ $Q'$ such that the 
  $\owqalin$ problem for $\Sigma$ and $Q$ reduces in $\ptime$ to the same
  problem for $\Sigma'$ and $Q'$, in the sense of Lemma~\ref{lem:ucqtocq}. 

  Further, if $\Sigma$ is $\aincd$s then $\Sigma'$ also is; if $\Sigma$ is empty
  then $\Sigma'$ also is; if $Q$ is base-covered then $Q'$ also is.
\end{lemma}

\begin{proof}
  We first preprocess the input UCQ $Q$ without loss of generality to remove any
  disjuncts where some distinguished relation $<_i$ is not a partial order
  (i.e., it has a cycle): the rewritten $Q$ is equivalent to the original one
  for $\owqalin$ because the removed disjuncts can never be entailed because of
  the semantics of distinguished relations.

  We now adapt the proof of Lemma~\ref{lem:ucqtocq}.
  The definition of
  $\sigmaor$ is unchanged, and we define $\sigmab' \colonequals \sigmab''
  \sqcup \sigmaor$
  where $\sigmab''$ is defined by increasing the arity of the
  relations from $\sigmab$ as before.
  The definition of $\Sigma'$ is unchanged. It easy to see that if $\Sigma$ is empty
  then so is $\Sigma'$, and if $\Sigma$ consists of $\aincd$s (i.e., there are no
  repetitions of variables in the body and in the head, and only one body fact)
  then this is still the case of~$\Sigma'$.

  The definition of $Q'$ is unchanged except that
  we do not rewrite distinguished relations: as the query is
  base-domain-covered, 
  each fresh variable $z_i$ that we add must actually occur in $Q_i'$,
  and the base-domain-coveredness (resp.\ base-coveredness) of $Q'$ is easy to see from that of~$Q$.

  We modify the definition of $\instance_\Or$ to
  complete each distinguished relation to a total order on~$\instance_\Or$
  (e.g., create $\false <_i \true$ for all distinguished relations~$<_i$). We
  define $\instance_\false$ in a new fashion.
  First, letting $m$ be the maximal number of
  variables of a disjunct of~$Q$, for each distinguished relation $<_i$, we
  create fresh values $\false^i_1, \ldots,
  \false^i_m$, and create facts $\false^i_1 <_i \cdots <_i \false^i_m$ 
  in~$\instance_f$. Second, we create all facts $R'(\vec f, \false)$ where $R' \in
  \sigmab$ and $\vec f$ is any tuple of the $\false^i_j$.
  Having defined $\instance_\false$, we then define 
  $(\instance_0)_\ptrue$ as the union of
  $\{R'(\vec a, \true) \mid R(\vec a) \in \instance_0 \wedge R \in \sigmab\}$
  and of the $\sigmad$-facts of~$\instance_0$ kept as-is; then we define $\instance_0'
  \defeq \instance_\Or \sqcup (\instance_0)_\ptrue 
  \sqcup \instance_\false$.

  To adapt the forward direction of the correctness proof, let us consider a
  counterexample $\instance$ to $\owqalin(\instance_0, \Sigma, Q)$ with suitably
  interpreted distinguished relations.
  We define $\instance'' \colonequals \instance'_0 \sqcup
  \instance_\ptrue \sqcup \instance_\false$, with $\instance_\false$ as above
  and $\instance_\ptrue$ defined as above by adding $\true$ to base facts and
  keeping distinguished facts as-is. It is clear that each distinguished
  relation in~$\instance''$ is a partial order, because this is true in
  isolation
  in~$\instance_\false$, in~$\instance_\Or$, and in~$\instance_\ptrue$; further
  $\instance_\false$ and $\instance_\Or$
  overlap only on~$\false$, $\instance_\Or$ and 
  $\instance_\ptrue$ only overlap on~$\true$,
  and~$\instance_\false$ and~$\instance_\ptrue$ do not overlap at all.
  Thus, we can define $\instance'$ from $\instance''$ by
  completing each distinguished relation to be a total order.
  
  We now explain how the correctness argument of the forward direction is adapted.
  Clearly
  $\instance' \supseteq \instance'_0$. To see why $\instance'$ satisfies
  $\Sigma'$, as $\Sigma'$ does not involve the distinguished relations, we
  reason as in Lemma~\ref{lem:ucqtocq} to deduce that a match is either included
  in $\instance_\false$ or in $\instance_\ptrue$: the first case is similar as a
  head match can be found in $\instance_\false$ by definition, and the second
  case is unchanged. To see why $\instance'$ violates~$Q'$, we show
  as in Lemma~\ref{lem:ucqtocq} that there is a match $h$ of some $Q'_{i_0}$ to
  $\instance'$ that maps all \emph{base} facts of~$Q'_{i_0}$ to facts with $\true$ as their last
  element. Now, the only distinguished facts
  where each individual element occurs in $\instance'$ in base facts of this form
  are the ones from $\instance_\ptrue$, that we constructed from $\instance$,
  which violated $Q_{i_0}$; hence, we can conclude using the fact that $Q$ is
  base-domain-covered.

  We now explain how the backward direction of the correctness proof is adapted.
  We construct $\instance$ as the disjoint union of
  $\{R(\vec a) \mid R'(\vec a, \true) \in \instance \land R' \in \sigmab''\}$ and
  of the distinguished facts of~$\instance'$ kept as-is.
  It is clear that $\instance'$ suitably interprets the distinguished relations,
  because its restriction to~$\sigmad$ is the same as $\instance$, which does.
  Again we have $\instance \supseteq \instance_0$. The fact that
  $\instance$ satisfies $\Sigma$ is as before, except that the arity of
  distinguished 
  facts is not changed in~$M'$, and the witness head facts of
  $M'$ in~$\instance'$ may include distinguished facts, in
  which case they are found as-is in $\instance$.
  
  To see that $\instance$ violates $Q$, we reuse the argument of
  Lemma~\ref{lem:ucqtocq}. The only new point that is needed is that
  the new $\instance_\false$ can still be used to find matches of any $Q'_{i_0}$
  with the last variable mapped to~$\false$, but this is easy to see from that
  construction (also recall our initial preprocessing of~$Q$ to eliminate
  disjuncts where some distinguished relation was not a partial order).
  This concludes the proof.
\end{proof}

\subsection{UCQ to CQ for $\owqa$ with Special Relations}
\label{apx:ucqtocqtc}

We now adapt Lemma~\ref{lem:ucqtocq}
so it can be applied to both $\owqatr$ and $\owqatc$.
We partition the signature into
two sets of relations: \emph{flagged relations}, whose arity will be increased as in
the proof of Lemma~\ref{lem:ucqtocq} above, and \emph{special relations}, whose
arity we will not increase. This partition is different from that of the rest of
the paper, where we had base and distinguished relations. In this section, the
special relations must include all distinguished relations (so that we do not
increase their arity), but they may also
include base relations. In particular, for $\owqatc$, the base relations of
which we are taking the transitive closure must themselves be special (indeed, we cannot
increase their arities).

We prove a generalization of Lemma~\ref{lem:ucqtocq} to the
setting with flagged relations and special relations,
and where we also allow logical constraints on special relations that go beyond
the TGDs allowed in Lemma~\ref{lem:ucqtocq}: in particular this will allow us to
impose transitivity and transitive closure requirements. The tradeoff is that we
will need to impose a restriction on the TGDs and CQ: intuitively, when we use
the special relations, we must also use the flagged relations (so we cannot
simply make all relations special).

We first define the constraints that we will allow on the special
relations:

\begin{definition}\label{def:distconstrset}
  On a signature $\sigma \colonequals \sigmaf \sqcup \sigmas$ partitioned into
  flagged and special relations, a \emph{special constraint set} $\Theta$ is a
  set of logical constraints on~$\sigmas$ involving any of the following:

  \begin{itemize}
    \item Disjunctive inclusion dependencies on~$\sigmas$;
    \item \emph{Transitivity assertions}, i.e., assertions that some binary relation in~$\sigmas$ is transitive
      (i.e., a special kind of TGD);
    \item \emph{Transitive closure assertions}, i.e., assertions that some binary relation in~$\sigmas$ is the transitive
      closure of another binary relation in~$\sigmas$.
  \end{itemize}
\end{definition}

Hence, special constraint sets can be used to express the semantics of
distinguished relations in the $\owqatr$ and $\owqatc$
problems : remember that distinguished relations are always special, and for
$\owqatc$ the base relations of which we are taking the transitive closure are
also special.

In addition to the special constraint set $\Theta$, our result will allow us to
write TGDs~$\Sigma$, and the negation of a CQ, like in Lemma~\ref{lem:ucqtocq}. However, in
exchange for the freedom of keeping special relations binary,
we need to impose a condition on the TGDs and on the CQ, which we call
\emph{flagged-reachability}. Intuitively, the goal of this condition is to
ensure that we can discriminate between matches of special relations in the
query or in dependency bodies that use facts annotated by $\true$, versus the
matches whose facts are annotated by $\false$. Indeed, this information cannot
be seen on the special relations because they do not carry the flag.

\begin{definition}
  \label{def:basereach}
  Let $G$ be the graph over the atoms of~$Q$ where two atoms are connected
  iff they share a variable. 
  A CQ $Q$ is \emph{flagged-reachable} if any
  special atom $A(x, y)$ in~$Q$ has a path to some flagged atom
  $B(\vec z)$ in $G$. A UCQ is flagged-reachable if all of its disjuncts are.
  A TGD is \emph{flagged-reachable} if its body is.
\end{definition}

The flagged-reachable restriction suffices to ensure that matches of special
relations in queries and rule bodies must correspond to $\false$ or to
non-$\false$ elements, 
by looking at the flagged facts to which the
special relations must be connected. We can thus show:

\begin{lemma}
  \label{lem:ucqtocqdist}
  For any signature $\sigma \colonequals \sigmaf \sqcup \sigmas$, flagged-reachable TGDs~$\Sigma$,
  special constraint set~$\Theta$ on~$\sigmas$,
  and flagged-reachable CQ $Q$,
  one can compute in
  $\ptime$ a signature $\sigma' \colonequals \sigmaf' \sqcup \sigmas$,
  TGDs~$\Sigma'$ and CQ $Q'$ such that the $\owqa$
  problem for $\Sigma \sqcup \Theta$ and $Q$ reduces in $\ptime$ to the $\owqa$
  problem for $\Sigma' \sqcup \Theta$
  and $Q'$ (for combined complexity and for data complexity): namely, given a
  set of facts $\instance_0$ on~$\sigma$, we can compute in $\ptime$ in
  $\instance_0$, $\Sigma$, $\Theta$ and $Q$ a set of facts $\instance_0'$ on~$\sigma'$
  such that $\owqa(\instance_0, \Sigma \sqcup \Theta, Q)$ holds iff $\owqa(\instance_0',
  \Sigma' \sqcup \Theta, Q')$ holds.

  Further, the following properties transfer from $\Sigma$ to $\Sigma'$:
  being $\incd$s; being $\aincd$s; being empty. Further, if some relation in
  $\sigma$ is not mentioned in~$Q$ then $Q'$ does not mention it either.
\end{lemma}

\begin{proof}
  We amend the proof of Lemma~\ref{lem:ucqtocq}. We first explain the change in
  the construction.
  The definition of
  $\sigmaor$ is unchanged, and we define $\sigmaf' \colonequals \sigmaf''
  \sqcup \sigmaor$
  where $\sigmaf''$ is defined by increasing the arity of the
  relations from $\sigmaf$ (like $\sigma'$ from $\sigma$ in
  Lemma~\ref{lem:ucqtocq}).
  The definition of $\Sigma'$ is unchanged
  except that the special relations are not rewritten; note that
  flagged-reachability and non-emptiness of the bodies ensure that they always contain
  a flagged relation, so that the variable that we add indeed occurs in the body (but
  it may not occur in the head).
  Clearly $\Sigma'$ is still flagged-reachable.
  Further, it is clear that, if $\Sigma$ consists of $\incd$s or $\aincd$s, then $\Sigma'$ also does,
  as there is still only one fact in the body (and in the case of $\aincd$ it is
  still a flagged fact),
  and there are no repetitions of variables. It is further clear that if $\Sigma$ is empty then
  $\Sigma'$ also is.
  
  The definition of $Q'$ is unchanged except that
  special relations are not rewritten; again,
  flagged-reachability of the
  query ensures that each fresh variable $z_i$ indeed occurs in $Q_i'$,
  and the condition on~$Q'$ clearly follows from the condition on~$Q$.
  We define $\instance_\false \colonequals \{R(\false, \ldots, \false) \mid R \in
  \sigmaf' \sqcup \sigmas\}$, and we define $(\instance_0)_\ptrue$ as the union of
  $\{R'(\vec a, \true) \mid R(\vec a) \in \instance_0 \wedge R \in \sigmaf\}$
  and of the $\sigmas$-facts of~$\instance_0$ kept as-is; then we define $\instance_0'
  \defeq \instance_\Or \sqcup (\instance_0)_\ptrue 
  \sqcup \instance_\false$ as
  before.

  We now explain how to modify the correctness proof.
  For the forward direction,
  consider a counterexample $\instance$ to $\owqa(\instance_0, \Sigma \sqcup
  \Theta, Q)$, assuming without loss of generality that $\true, \false \notin
  \instance$.
  We define $\instance' \colonequals \instance'_0 \sqcup
  \instance_\ptrue \sqcup \instance_\false$, with $\instance_\false$ as above
  and $\instance_\ptrue$ defined as above by adding $\true$ to flagged facts and
  keeping special facts as-is.
  
  To verify that this construction for the forward direction is correct, we must
  first show that $\instance'$ satisfies the constraints
  $\Theta$ on~$\sigmas$. For disjunctive inclusion dependencies $\tau$,
  letting $M$ be a match of the body, as $\Theta$ does not mention
  the facts of~$\instance_\Or$, we must have $M \in \instance_\false$ or $M \in
  \instance_\ptrue$. Now, in the first case, by construction
  $\instance_\false$ must contain a match
  of some head disjunct of $\tau$, and in the second case, by considering $M$ in
  $\instance$ which satisfies $\tau$, we can also extend the match in
  $\instance_\ptrue$
  hence in $\instance'$. For transitivity
  assertions, we reason in the same way: seeing them as a TGD with a connected
  body that does not mention the relations of~$\sigma_\Or$, we deduce again that
  any match of them must be within~$\instance_\false$ or within~$\instance_\ptrue$,
  which allows us to conclude. The same reasoning works for
  transitive closure assertions.
  Hence, $\instance'$ satisfies $\Theta$.
  
  Now, we must verify the other
  conditions on~$\instance'$, for which we adapt the proof of
  Lemma~\ref{lem:ucqtocq}.
  In particular, observe that $\instance'$
  is still a superset of~$\instance_0'$. To check that there are no
  violations of~$\Sigma'$, consider a match $M'$ of a TGD $\tau'$ of~$\Sigma'$;
  as before $M'$ includes no fact of~$\instance_\Or$, and the \emph{flagged facts}
  $M'_\calB$
  of~$M'$ are either included in~$\instance_\ptrue$ or in~$\instance_\false$ depending
  on their last element. Now, as $\tau'$ is flagged-reachable,
  we observe that the \emph{special
  facts} $M'_\calD$ of~$M'$ must be connected to the flagged facts, so that,
  as $\instance_\ptrue$ and $\instance_\false$ have disjoint domains and no facts
  connect them except $\sigma_\Or$-facts which do not occur in $\tau'$, it must
  again be
  the case that the \emph{entire match} $M'$ is either included
  in~$\instance_\false$ or in~$\instance_\ptrue$. So we can conclude as before
  (in the second case,
  the preimage $M$ of~$M'$ in~$\instance$ is defined without changing the
  special
  facts, but the same reasoning applies).
  
  To check that $\instance'$ violates $Q'$, as before we reason by contradiction
  and deduce that $\exists \vec x_{i_0}
  Q'_{i_0}(\vec x_{i_0}, \true)$ holds in~$\instance'$.
  Now, its match $M'$ in $\instance'$ must
  consist of a match $M'_{\calB}$ of flagged facts of~$M'$ with $\true$ as their last
  position (so they are in~$\instance_\ptrue$) and a match $M'_{\calD}$ of
  special facts. As in the previous paragraph, 
  we now use the fact
  that $Q$, hence $Q'_{i_0}$, is flagged-reachable, so we must also have $\dom{M'_\calD} \subseteq
  \dom{M'_\calB}$. Hence, we have $M' \subseteq \instance''$, and we deduce as
  before (except that we do not increase the arity of special facts) 
  that the preimage $M$ of~$M'$ is a match of~$Q$ in~$\instance$ and conclude by
  contradiction. This proves the correctness of the forward direction.

  For the backward direction, we build $\instance$ as the disjoint union of
  $\{R(\vec a) \mid R'(\vec a, \true) \in \instance \land R' \in \sigmaf'\}$ and
  of the special facts of~$\instance'$ kept as-is.
  It is clear that $\instance'$ satisfies $\Theta$, because its restriction to~$\sigmas$
  is the same as $\instance$, which satisfies $\Theta$.
  Again we have $\instance \supseteq \instance_0$. The fact that
  $\instance$ satisfies $\Sigma$ is as before, except that the arity of special 
  facts is not changed in~$M'$, and the witness head facts of
  $M'$ in~$\instance'$ may include special facts, in
  which case they are found as-is in $\instance$. The fact that $\instance$ violates
  $Q$ is exactly as before, which concludes the correctness proof.
\end{proof}

\subsection{Revisiting the Results of Section~\ref{sec:hardness}}
\label{app:ucqtocqhardness}

We now apply the results of Appendix~\ref{apx:ucqtocqlin}
and~\ref{apx:ucqtocqtc} to show the results of Section~\ref{sec:hardness} with a
CQ instead of a UCQ.

\subsubsection{Theorem~\ref{thm:tcdisj}}
We apply Lemma~\ref{lem:ucqtocqdist} by picking as special relations $E$
and $E^\trans$, and taking all others to be flagged relations. The special
constraint set $\Theta$ asserts that $E^\trans$ is the transitive closure of~$E$.
We observe that the $\aincd$s $\Sigma'$ created in the proof are
flagged-reachable, because their bodies always consist of base facts. Now, we
observe that the UCQ $Q'$ is also flagged-reachable: the $E$-atoms in
$Q$-generated disjuncts are connected to the corresponding atom $R'(\vec x, e,
f)$, the $E$-facts in $E$-path length restriction disjuncts are connected to an
$R'$-fact, and the $E$-facts in $\did$ satisfaction disjuncts are connected to
the $\witness_\tau$-fact.

Hence, we can deduce from Lemma~\ref{lem:ucqtocqdist} that
Theorem~\ref{thm:tcdisj} extends to $\owqatc$ with CQs, both for data complexity
and combined complexity.

\subsubsection{Theorem~\ref{thm:lindisj}}
It is immediate to observe that the $\aincd$s $\Sigma'$ do not mention the
distinguished relations $<$, and we have already observed in the proof that the
UCQ that we define is base-covered, hence base-domain-covered,
so we deduce from Lemma~\ref{lem:ucqtocqlin}
that Theorem~\ref{thm:lindisj} also extends to $\owqalin$ with CQs, for data and
combined complexity.

\subsubsection{Proposition~\ref{prop:lindatacompltrans}}
We let $E$ and $E^\trans$ be the special relations, let all others be flagged
relations, and let the special constraint set $\Theta$ assert that $E^\trans$ is the transitive closure
of~$E$.
It is easy to observe that the query defined in the proof is flagged-reachable
(in particular thanks to the $C_\chi$-atom in the $E$-path length
restriction disjunct). As the constraints are empty, their image by
Lemma~\ref{lem:ucqtocqdist} also is, so we deduce that the data complexity lower
bound of Proposition~\ref{prop:lindatacompltrans} still applies to $\owqatc$
with CQs.

\subsubsection{Proposition~\ref{prop:lindatacompl}}
As the UCQ defined in the proof is base-covered, hence base-domain-covered, and the constraints are empty,
we deduce from Lemma~\ref{lem:ucqtocqlin} that the lower bound still applies to
$\owqalin$ with CQs.

\subsection{Revisiting the Results of Section~\ref{sec:undecid}}
\label{app:ucqtocqundecid}

We again apply the results of Appendix~\ref{apx:ucqtocqlin}
and~\ref{apx:ucqtocqtc} to show the results of Section~\ref{sec:undecid} with a
CQ instead of a UCQ.

\subsubsection{Theorem~\ref{thm:undectransb}}
\label{app:fix}
In this appendix, we prove Theorem~\ref{thm:undectransb}. Recall its
statement:

\medskip

\undectrans

\medskip

The same theorem was claimed in the JAIR version of this article
\cite{amarilli2018query}. In that version, we proved the result for a UCQ
instead of a CQ, and then incorrectly claimed that the result for a CQ could be
established using Lemma~\ref{lem:ucqtocqdist}. However, this is an oversight,
because the signature obtained after applying the lemma is no longer arity-two.

In this appendix, we give a revised argument explaining why the result holds for
a CQ. As in the previous results in this section, we will do so by adapting the proof for the
case of a UCQ, which was given in Section~\ref{sec:undecid}. We call this the
``UCQ case''.
We now explain how to adapt the proof without
increasing the arity of the signature:

\begin{proof}[Proof of Theorem~\ref{thm:undectransb}]
Let us fix the undecidable tiling problem from which we reduce.
The fixed signature contains, like in the UCQ case, the base relations $S'$
(binary),
$K_i$ (binary) for each color $C_i$, and $K_i'$ (unary) for each color $C_i$;
and the distinguished transitive relation~$S^\trans$. 
We also add one unary relation $\mathrm{True}$, and 
three binary relations $\mathrm{Or\_in}1$, $\mathrm{Or\_in}2$, $\mathrm{Or\_out}$.

The fixed dependencies $\Sigma$ are the same as those of the UCQ case, namely:
\begin{align*}
    S'(x, y) & \rightarrow \exists z ~ S'(y, z) &\qquad
    S'(x, y) & \rightarrow S^\trans(x, y)\\
    S^\trans(x, y) & \rightarrow \bigvee_i K_i(x, y) &\qquad
    S^\trans(x, y) & \rightarrow \bigvee_i K_i(y, x) &\qquad
    S^\trans(x, y) & \rightarrow \bigvee_i K_i'(x) .
  \end{align*}

To define the CQ, we use an idea reminiscent of 
Lemma~\ref{lem:ucqtocqdist}, but we rely on transitivity to avoid increasing the
signature arity. 
Consider all disjuncts $Q_1, \ldots, Q_n$ of the UCQ defined
in the UCQ case. Recall that these disjuncts were defined by taking,
\begin{inparaitem}[]
    \item for each forbidden horizontal pair $(C_i, C_j) \in \mathbb{H}$, with
      $1 \leq i, j \leq k$, the disjuncts:
      \begin{align*}
        K_i(x, y) \wedge S'(y, y') \wedge K_j(x, y') \quad
        K_i'(y) \wedge S'(y, y') \wedge K_j(y, y') \quad
        K_i(y', y) \wedge S'(y, y') \wedge K_j'(y')
      \end{align*}
    \item and for each forbidden vertical pair $(C_i, C_j) \in \mathbb{V}$, the
      analogous disjuncts
      \begin{align*}
        K_i(x, y) \wedge S'(x, x') \wedge K_j(x', y) \quad
        K_i'(x) \wedge S'(x, x') \wedge K_j(x', x) \quad
        K_i(x, x') \wedge S'(x, x') \wedge K_j'(x') .
      \end{align*}
   \end{inparaitem}
Note that each disjunct contains
precisely one $S'$-fact.
For each disjunct $Q_p$ with $1 \leq p \leq n$,
let $x_p$ be the variable occurring at the first
position of that fact, and let $z_p$ be a fresh variable. We let $Q'_p$ be the
conjunction of $Q_p$ and of the fact $S'(z_p, x_p)$. Now we let the CQ $Q$ be
the conjunction of all the $Q_p'$ defined before, and of the following conjuncts,
for fresh variables $f_1, \ldots, f_{n-1}$ and $z'_1, \ldots, z'_{n-1}$:
\begin{align*}
  \mathrm{Or\_in}1(z_1, f_1) \land \mathrm{Or\_in}2(z_2, f_1) \land \mathrm{Or\_out}(f_1, z_1')\\
  \mathrm{Or\_in}1(z_1', f_2) \land \mathrm{Or\_in}2(z_3, f_2) \land \mathrm{Or\_out}(f_2, z_2')\\
  \mathrm{Or\_in}1(z_2', f_3) \land \mathrm{Or\_in}2(z_4, f_3) \land \mathrm{Or\_out}(f_3, z_3')\\
  \vdots\\
  \mathrm{Or\_in}1(z_{n-2}', f_{n-1}) \land \mathrm{Or\_in}2(z_n, f_{n-1}) \land \mathrm{Or\_out}(f_{n-1}, z_{n-1}')\\
  \mathrm{True}(z_{n-1}')
\end{align*}
The intuition is that we want to code the OR Boolean operation in a table
as in the process of Lemma~\ref{lem:ucqtocqdist}, but
instead of storing the truth table in a table of arity 3, we will use  tables
of arity $2$ storing their ``reifications''. That is, we store the relationship of each row in the table to each value.
Thus, the $f_p$'s will each represent a row of the truth table, and the $z_p$'s and
$z_p'$'s will each represent a Boolean value.

We now explain how the initial set of facts $\instance_0$ is defined from the
input instance $c_0, \ldots, c_n$ of the tiling problem. It includes the facts
of the UCQ case, namely: 
  the fact $K_j'(a_0)$ such that $C_j$ is the color of $c_0$, and
  for $0 \leq i < n$, the fact $S'(a_i, a_{i+1})$
  and the fact $K_j(a_0, a_i)$ such that $C_j$ is the color of initial element $c_i$.
 But we add the following:
\begin{itemize}
  \item The fact
    $\mathrm{True}(\mathit{true})$, for a fresh element $\mathit{true}$;
  \item The facts $S'(\mathit{false}, \mathit{false})$, $K'_i(\mathit{false}, \mathit{false})$, and $K'_i(\mathit{false})$
    for each color $i$, for a fresh element $\mathit{false}$;
  \item One fact $S'(\mathit{true}, a_0)$ where $a_0$ is the element already
    defined in the UCQ case and presented above;
  \item Facts encoding the truth table of the OR-operator in arity two, that is,
    for fresh elements $\phi_{ff}$, $\phi_{ft}$, $\phi_{tf}$, $\phi_{tt}$:
    \begin{itemize}
      \item $\mathrm{Or\_in}1(\mathit{false}, \phi_{ff})$, $\mathrm{Or\_in}2(\mathit{false}, \phi_{ff})$, $\mathrm{Or\_out}(\phi_{ff}, \mathit{false})$
      \item $\mathrm{Or\_in}1(\mathit{false}, \phi_{ft})$, $\mathrm{Or\_in}2(\mathit{true}, \phi_{ft})$, $\mathrm{Or\_out}(\phi_{ft}, \mathit{true})$
      \item $\mathrm{Or\_in}1(\mathit{true}, \phi_{tf})$, $\mathrm{Or\_in}2(\mathit{false}, \phi_{tf})$, $\mathrm{Or\_out}(\phi_{tf}, \mathit{true})$
      \item $\mathrm{Or\_in}1(\mathit{true}, \phi_{tt})$, $\mathrm{Or\_in}2(\mathit{true}, \phi_{tt})$, $\mathrm{Or\_out}(\phi_{tt}, \mathit{true})$
    \end{itemize}
\end{itemize}
We now claim that this reduction is correct: the tiling problem has a solution iff there is a superset of $\instance_0$ that satisfies $\Sigma$ and violates $Q$ and where $S^\trans$ is transitive.
Intuitively, in the process of Lemma~\ref{lem:ucqtocqdist}, we added an additional
argument to some relations to track the truth value, increasing the arity. In this proof
we do not do this, to avoid increasing the arity. Thus
we will need a new trick, exploiting transitivity, to show correctness.

For the forward direction, from a solution $f$ to the tiling problem for input
$\vec{c}$, we construct the counterexample 
as follows. Starting from $\instance_0$, we first add the facts defined in
the UCQ case. Recall that this means 
extending the initial chain of $S'$-facts in~$\instance_0$ to
  an infinite chain $S'(a_0, a_1), \ldots, \allowbreak S'(a_m, a_{m+1}),
  \ldots$,
  setting $S^\trans$ to be the
  transitive closure of this $S'$-chain (so it is indeed transitive),
  creating the fact $K_l(a_i, a_j)$ where $l = f(i,j)$
  for all $i, j \in \mathbb{N}$ such that $i \neq j$, 
  and creating the fact $K_l'(a_i)$ where $l = f(i,i)$ 
  for all $i \in \mathbb{N}$. Now, we add the following facts: \begin{itemize}
  \item $S^\trans(\mathit{false}, \mathit{false})$;
  \item $S^\trans(\mathit{true}, a_i)$ for each element $a_i$ of the infinite
    chain constructed in the UCQ case as explained above;
  \item the facts $K_x'(\mathit{true})$ and $K_x(\mathit{true}, a_i)$ and $K_x(a_i, \mathit{true})$
    for each element $a_i$ of the infinite chain, for some arbitrary color~$x$.
\end{itemize}
In the resulting set of facts $\instance_0'$, the relation $S^\trans$ is
transitive, namely, it is the transitive closure of the infinite chain
$\mathit{true}, a_0, a_1, \ldots$ plus the one fact $S^\trans(\mathit{false},
\mathit{false})$. One can also check that $\instance_0'$ satisfies~$\Sigma$. 
We now explain why~$Q$ is violated, by reasoning similarly to Lemma~\ref{lem:ucqtocqdist}.
Assume by contradiction that $Q$ is satisfied.
Given that $\instance_0'$ only contains the fact $\mathrm{True}(\mathit{true})$
of the $\mathrm{True}$ relation,
and given that $\mathrm{True}(z_{n-1})$ has a match in~$\instance_0'$, 
we know that $z'_{n-1}$ was mapped to~$\mathit{true}$. Thus,
given that the only facts in $\instance_0'$ for the relations
$\mathrm{Or\_in}1$, $\mathrm{Or\_in}2$, and $\mathrm{Or\_out}$ are the facts
that were in~$\instance_0$,
necessarily some $z_p$ was mapped to~$\mathit{true}$.
Thus, there is a match of
$S'(\mathit{true}, x_p) \land Q_p$ where $Q_p$ is a disjunct of the UCQ 
from the UCQ case,
and $x_p$ is the variable at the first position of the $S'$-fact of~$Q_p$.
Given the facts in~$\instance_0'$, this means that there is a match of $Q_p$ in~$\instance_0'$ where the $S'$-fact is mapped to $S'(a_i, a_j)$ for some $i < j$. This would witness a forbidden pair in the tiling, which is a contradiction.
(Note that the facts defined with the arbitrary color~$x$ earlier can never be
part of a match of~$Q$.)
Thus, $Q$ is violated, establishing the forward direction.

For the backward direction, consider $\instance_0' \supseteq \instance_0$ that
satisfies $\Sigma$ and violates $Q$. We can define a tiling $f$ from $a_0, \ldots,
a_n, \ldots$, as in the proof for the UCQ case, namely, by intuitively following
some chain of $S'$-facts starting by $a_0, \ldots, a_n$ and continuing
indefinitely (with elements that may be distinct or not), considering the
$S^+$-facts $S^+(a_i, a_j)$ on that chain that must exist for each $a_i < a_j$, and setting $f$
according to the facts of the $K_l$- and $K_l'$-relations on that chain which must exist
according to~$\Sigma$. This tiling
satisfies the initial tiling problem instance. Let us argue that it is indeed a solution to the tiling
problem. Assume by contradiction that it is not, then as in the proof of the UCQ
case there would be a match of a UCQ disjunct $Q_p$ mapping the $S'$-fact
of~$Q_p$
to $S'(a_i, a_j)$ for some $i < j$. We can extend this to a match of the entire CQ $Q$ by:
\begin{itemize}
  \item mapping the atoms of $Q_p$ to this match;
  \item mapping $z_p$ to $\mathit{true}$, as there must be an $S^\trans$-fact $S^\trans(\mathit{true}, a_i)$ in $\instance_0'$;
  \item mapping the atoms of $Q_q$ for $q \neq p$ to the facts on
    $\mathit{false}$, which exist in $\instance_0$, in particular mapping $z_q$
    and $x_q$ to $\mathit{false}$;
  \item mapping the conjuncts featuring the $\mathrm{Or\_in}1$, $\mathrm{Or\_in}2$, and $\mathrm{Or\_out}$
    relations to the facts of $\instance_0$ in the expected way.
\end{itemize}
Thus we obtain a match of~$Q$, contradicting our assumption. Thus, we have indeed defined a solution to the tiling problem, establishing the backward direction of the correctness proof and concluding the proof.
\end{proof}

\subsubsection{Theorem~\ref{thm:undectrans}}
We use Lemma~\ref{lem:ucqtocqdist}, and pick 
$S'$ as the only flagged relation, and let all other relations be special
relations. We
let $\Sigma$ consist of the one ID applying only to~$S'$; it is flagged-reachable.
The disjuncts of the CQ that we write are always
flagged-reachable, as they are connected and always include a $S'$-fact. The
special constraint set $\Theta$ consists of all other $\incd$s used in the proof, and of the assertion that
$S^\trans$ is the transitive closure of~$S$.
Applying Lemma~\ref{lem:ucqtocqdist}, we reduce the
$\owqa$ problem with a UCQ to $\owqa$ for $\Theta$ (so $\incd$s plus the transitive
closure assertion on the one distinguished relation $S^\trans$), for the
translation of $\Sigma'$ (which are $\incd$s), and a CQ which still does not use
the one distinguished relation $S^\trans$ of the new signature.
This establishes the result of Theorem~\ref{thm:undectrans}.

\subsubsection{Theorem~\ref{thm:undeccq}}
We use Lemma~\ref{lem:ucqtocqlin}. We check that, indeed, the constraints do not
mention the distinguished relation, and that the two UCQ disjuncts which mention
these relations are
base-domain-covered. Hence, the translation shows undecidability of $\owqalin$
for $\aincd$s and a CQ, concluding the proof of Theorem~\ref{thm:undeccq}.

\vskip 0.2in
\bibliography{algs}
\bibliographystyle{theapa}

\end{document}